\documentclass[11pt,a4paper]{article}

\usepackage{preamble}
\usepackage[left=1in, right=1in, top=1in, bottom=1in]{geometry}
\usepackage{comment}
\usepackage{graphicx}
\usepackage{hyperref}
\usepackage{xcolor}
\usepackage{adjustbox}[export]
\usepackage{quantikz}
\usepackage{tikzsymbols}
\usepackage{authblk}
\usepackage{mathtools}

\usepackage{todonotes}

\newcommand{\boris}[1]{}
\newcommand{\alex}[1]{}
\newcommand{\note}[1]{}

\newcommand{\Bd}{\mathrm{Bd}}

\DeclareMathOperator{\out}{\mathrm{out}}

\newcommand{\ii}{\mathsf{i}}
\newcommand{\f}{\mathsf{f}}

\newcommand{\Dr}{\mathrm{Dr}}

\newcommand{\interval}[2]{{#2}\rightarrow{#1}}

\DeclareMathOperator{\anchor}{  \adjincludegraphics[valign=B, width = 0.25cm]{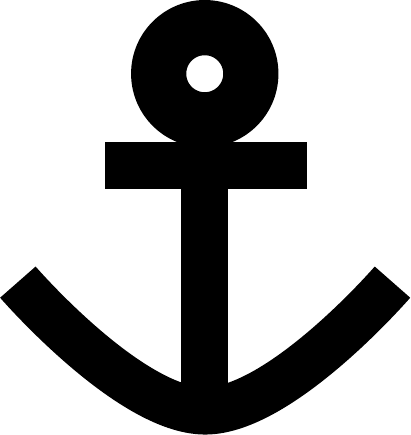}  }
\DeclareMathOperator{\anchorino}{  \adjincludegraphics[valign=B, width = 0.18cm]{anchor_symbol.pdf}  }

\DeclareMathOperator{\tr}{\mathrm{tr}}
\DeclareMathOperator{\Ob}{\mathrm{Ob}}
\DeclareMathOperator{\Irr}{\mathrm{Irr}}

\DeclareMathOperator{\Tube}{\rm{Tube}}

\newcommand{\1}{\mathbbm{1}}

\DeclareMathOperator{\op}{\mathrm{op}}
\DeclareMathOperator{\coev}{\mathrm{coev}}
\DeclareMathOperator{\ann}{\mathbb{A}}

\DeclareMathOperator{\inn}{\mathrm{in}}

\newcommand{\ZC}{Z(\mathcal{C})}
\DeclareMathOperator{\gl}{\mathrm{gl}}

\begin{document}

\title{Sector Theory of Levin-Wen Models I \\ Classification of Anyon Sectors}

\author[1]{Alex Bols \thanks{email: \href{abols01@phys.ethz.ch}{abols01@phys.ethz.ch}}}
\author[2]{Boris Kj\ae r \thanks{email: \href{bbk@math.ku.dk}{bbk@math.ku.dk}}}
\affil[1]{Institute for Theoretical Physics, ETH Z{\"u}rich}
\affil[2]{QMATH, Department of Mathematical Sciences, University of Copenhagen}

\maketitle

\begin{abstract}
We classify the irreducible anyon sectors of Levin-Wen models over an arbitrary unitary fusion category $\caC$, showing that they are in one-to-one correspondence with equivalence classes of simple objects of the unitary Drinfeld center $\ZC$. We achieve this by making explicit how the Levin-Wen Hamiltonian stabilizes subspaces isomorphic to state spaces of the corresponding Turaev-Viro TQFT, and developing a detailed understanding of these state spaces on punctured disks. In particular, we construct Drinfeld insertion operators on such spaces which can move anyons between the punctures, and can change their fusion channels. Using these Drinfeld insertions, we construct explicit string operators that excite anyons above the ground state. The fusion and braiding properties of these anyons will be analysed in a companion paper.
\end{abstract}

\section{Introduction} \label{sec:introduction}

It has been understood for a long time that the anyon content of a gapped 2+1 dimensional quantum field theory is captured by its \emph{sector theory} \cite{doplicher1971local, doplicher1974local, fredenhagen1989superselection, fredenhagen1992superselection, frohlich1988statistics, frohlich1990braid}. This insight was transferred to the setting of gapped 2+1 dimensional quantum lattice models in \cite{naaijkens2011localized, ogata2022derivation}, establishing the sector theory as an invariant of gapped phases. In this framework, anyon types are in one-to-one correspondence with \emph{superselection sectors}, equivalence classes of representations of the observable algebra that satisfy the superselection criterion with respect to the vacuum (see Definition \ref{def:anyonsector} below). Under the assumption of (approximate) Haag duality, each such representation can be obtained from the vacuum representation by acting with a \emph{string operator}, an endomorphism localized in a cone region which produces an anyon from the vacuum. These string operators allow for a natural definition of fusion and braiding structures on the category of superselection sectors, capturing the fusion and braiding properties of the model's anyonic excitations.

Recent years have seen much progress in computing the sector theory of exactly solvable toplogically ordered quantum lattice models, especially for Kitaev's quantum double models \cite{naaijkens2011localized, naaijkens2015kitaev, bols2024double, bols2025classification, bols2025category}. In this work we make further progress by investigating the sector theory of Levin-Wen models \cite{levin2005string, lin2021generalized}, which are believed to provide representatives for all topological orders in 2+1 dimensions which admit gapped boundaries. (A class of lattice models believed to capture all gapped phases in 2+1 dimensions has recently been proposed in \cite{sopenko2023chiral}.) The main result of this paper, the first in a series of two, is the classification of the irreducible anyon sectors of Levin-Wen models. Establishing fusion and braiding structures will occupy the second paper in this series \cite{sectortheoryII}. The question of establishing Haag duality for quantum double and Levin-Wen models has been addressed in \cite{Naaijkens2012, Fiedler2014, ogata2025haag, naaijkens2026local}.

The complete classification of irreducible anyons sectors was achieved for quantum double models with finite gauge group $G$ in \cite{bols2025classification}. The constructions there rely heavily on explicit localized and transportable \emph{amplimorphisms} originally introduced in \cite{naaijkens2015kitaev}. These amplimorphisms have the nice property that they provide an action of the model's anyon theory, namely the representation category of the quantum double $\caD(G)$ of the gauge group, on the quasi-local algebra of observables (see \cite[Section~5]{bols2025category}). It follows from the discussion in \cite[Section 5.1]{chen2022q} that string operators with this nice property \emph{cannot} exist for lattice models whose anyon theory has non-integer quantum dimensions, including many of the Levin-Wen models covered here. The main novelty of this work is the construction of explicit string operators for all the irreducible anyon sectors of Levin-Wen models. As the no-go result of \cite{chen2022q} suggests, our construction differs significantly from the string operators used in Kitaev's quantum double models \cite{kitaev2003fault, Bombin2008} and their generalisations based on $\rm C^*$-Hopf algebras \cite{Yan2022}.

Similar to the arguments used in \cite{bols2025classification}, a crucial stepping stone in our analysis consists of a complete characterization of the space of density matrices on annular regions obtained by restricting states satisfying Levin-Wen ground state constraints on larger regions. We show in Section \ref{sec:skein subspaces} that these are precisely the density matrices supported on so-called skein subspaces (stabilized by the local Levin-Wen Hamiltonian), which moreover satisfy \emph{maximally mixed boundary conditions} of types labelled by the irreducible objects of $\ZC$. This space of density matrices coincides with the information convex sets of the entanglement bootstrap program \cite{shi2020fusion}.

The paper is organised as follows. In section \ref{sec:model} we describe the Levin-Wen model based on a unitary fusion category $\caC$ in infinite volume, and state our main theorem classifying the irreducible anyon sectors of its ground state. In Section \ref{sec:tube algebras and skein modules} we introduce and analyse \emph{skein modules} and various algebras acting on them, in particular the $\Tube$-algebras. These skein modules are the state spaces of a Turaev-Viro TQFT. Section \ref{sec:skein subspaces} makes explicit how the Levin-Wen Hamiltonian stabilizes \emph{skein subspaces} for any region of the lattice. These skein subspaces are shown to be isomorphic to the skein modules of Section \ref{sec:tube algebras and skein modules}. In particular, we completely describe the spaces of density matrices supported on these skein subspaces, with a particular emphasis on the density matrices that satisfy \emph{maximally mixed boundary conditions}. In Section \ref{sec:anyon states}, with a good understanding of skein subspaces in hand, we construct, for every simple object $X$ of the unitary Drinfeld center $\ZC$ and any edge $e$ of the lattice, a pure state $\omega_e^X$ interpreted as an anyon excitation of type $X$ sitting near $e$. We construct string operators in Section \ref{sec:string operators} and show in Section \ref{sec:anyon representations} that these string operators yield the GNS representations $\pi_e^X$ of the anyon states $\omega_e^X$ when composed with the vacuum representation of the model. We use these string operators to show that the $\pi_e^X$ satisfy the superselection criterion. Finally, in Section \ref{sec:completeness} we show that any irreducible anyon representation is equivalent to one of the $\pi_e^X$, concluding that the irreducible anyon sectors of the model are in one-to-one correspondence with equivalence classes of simple objects of $\ZC$. Appendices \ref{app:proof of characterization of skein modules} and \ref{app:proof of skein subspace isomorphism} contain proofs of certain basic properties of skein modules and skein subspaces.

\subsection*{Acknowledgements}
We thank Corey Jones and David Penneys for useful conversations, and Daniel Wallick for pointing out an error in our original version of what is now Lemma \ref{lem:restriction yields maximally mixed boundary conditions}. B. K. was supported by the
Villum Foundation through the QMATH Center of Excellence (Grant No. 10059) and the Villum Young Investigator (Grant No. 25452) programs.

\setcounter{tocdepth}{2}
\tableofcontents

\section{The Levin-Wen model and its anyon sectors} \label{sec:model}

\subsection{Local degrees of freedom} \label{subsec:local degrees of freedom}

Fix a unitary fusion category (UFC) $\caC$ with a representative set of simple objects $\Irr \caC$. (See Section \ref{subsec:UFC} below for details.) To each site $v \in \Z^2 \subset \R^2$ of the square lattice we associate a local degree of freedom
$$ \caH_v = \bigoplus_{a, b, c, d \in \Irr \caC} \, \caC(a \otimes b \rightarrow c \otimes d). $$
An element $\phi \in \caH_v$ in the subspace $\caC(a \otimes b \rightarrow c \otimes d)$ will be represented graphically by
\begin{equation}
    \adjincludegraphics[valign=c, height = 1.0cm]{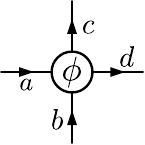} \quad,
\end{equation}
where the morphism $\phi$ is associated to the site $v \in \Z^2$, and the diagram is read from bottom left to top right.

Following \cite{kong2014universal, christian2023lattice, green2024enriched}, the degrees of freedom $\caH_v$ are equipped with the \emph{skein inner product} given for $\phi, \psi \in \caC(a \otimes b \rightarrow c \otimes d)$ by
\begin{equation}
    \langle \phi, \psi \rangle := \frac{\tr \{ \phi^{\dag} \circ \psi \}}{\sqrt{d_a d_b d_c d_d}} = \frac{1}{\sqrt{d_a d_b d_c d_d}} \,\,\, \adjincludegraphics[valign=c, height = 1.0cm]{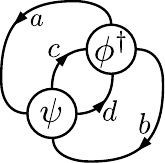} \quad.
\end{equation}

For any $v \in \Z^2$ we put $\caA_v = \End(\caH_v)$, and for finite $V \subset\subset \Z^2$ we put $\caH_V = \bigotimes_{v \in V} \caH_v$ and $\caA_V = \bigotimes_{v \in V} \caA_v \simeq \End( \caH_V )$. If $V \subset W$ is an inclusion of finite subsets of $\Z^2$ then there is a natural embedding $\caA_V \hookrightarrow \caA_W$ by tensoring with the identity of $\caA_{W \setminus V}$. For any, possibly infinite, $X \subset \Z^2$ these inclusions make $\{ \caA_V \}_{V \subset\subset X}$ into a directed system of matrix algebras. We denote its direct limit by $\caA_X^{\loc}$ and put $\caA_X = \overline{\caA_X^{\loc}}^{\norm{\cdot}}$. The *-algebra $\caA_X^{\loc}$ is called the algebra of local observables on $X$, and the $\rm C^*$-algebra $\caA_X$ is called the quasi-local algebra on $X$. We also write $\caA = \caA_{\Z^2}$ and $\caA^{\loc} = \caA_{\Z^2}^{\loc}$, called the quasi-local algebra and the local algebra respectively.

For any $S \subset \R^2$ let $\overline S = S \cap \Z^2$ and write $\caA_S := \caA_{\overline S}$. If $\overline S$ is finite, we also write $\caH_S := \caH_{\overline S}$.

\subsection{The Levin-Wen Hamiltonian}

Let $\bse_1 = (1, 0)$ and $\bse_2 = (0, 1)$ be unit vectors in $\Z^2 \subset \R^2$. We denote by $ \vec \caE = \{ (v_0, v_1) \in \Z^2 \times \Z^2 \, : \, \dist(v_0, v_1) = 1 \} $ the set of oriented edges of $\Z^2$ and by $ \caE = \{ (v_0, v_1) \in \Z^2 \times \Z^2 \, : \, v_1 - v_0 \in \{ \bse_1, \bse_2 \} \}$ the set of edges of $\Z^2$ that are oriented to the top right. We identify $\caE$ with the set of unoriented edges of $\Z^2$.
For $e = (v_0, v_1) \in \vec \caE$ we write $\bar e = (v_1, v_0)$ for its orientation reversal. We also write $\partial_{\ii} e = v_0$ and $\partial_{\f} e = v_1$.

For each $e = (v_0, v_1) \in \vec \caE$ we define a projector $A_e$ acting on $\caH_{v_0} \otimes \caH_{v_1}$ by
\begin{equation} \label{eq:A_e defined}
    A_e \,\,\, \adjincludegraphics[valign=c, height = 0.8cm]{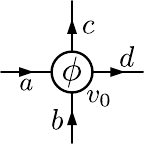} \,\, \otimes \adjincludegraphics[valign=c, height = 0.8cm]{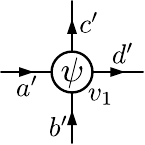} \,\,\, = \,\,\, \delta_{d \, a'} \,\, \adjincludegraphics[valign=c, height = 0.8cm]{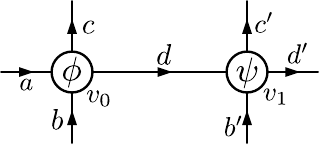} \quad.
\end{equation}
Clearly $A_e = A_{\bar e}$. We say $A_e$ enforces the \emph{string-net constraint} at edge $e$.

Denote by $\caF$ the set of faces of $\Z^2$. For any face $f \in \caF$ we let $\caH_f = \bigotimes_{v \in f} \caH_v$ and define an orthogonal projector $B_f$ on $\caH_f$ which annihilates the orthogonal complement of $\prod_{e \in f} A_e \caH_f$ and acts on $\prod_{e \in f} A_e \caH_f$ by inserting the regular element of $\caC$ and using local relations to bring the diagram back into product form: 
\begin{equation} \label{eq:B_f defined}
    B_f \,\, \adjincludegraphics[valign=c, height = 1.0cm]{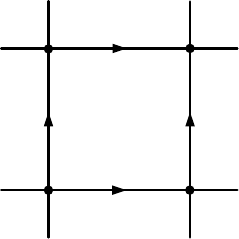} \,\,\, = \,\,\, \frac{1}{\caD^2} \, \sum_{a \in \Irr \caC} \, d_a \,\, \adjincludegraphics[valign=c, height = 1.0cm]{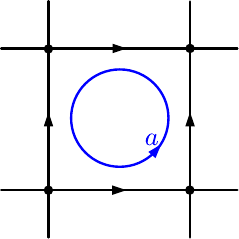} \quad.
\end{equation}
Here $d_a$ is the quantum dimension of the object $a$, and $\caD^2 = \sum_a \, d_a^2$. A precise definition will be given in Section \ref{subsec:Bf defined} below.

It was shown that $B_f$ is an orthogonal projector in \cite[page~10]{green2024enriched}. Moreover, all $\{B_f\}_{f \in \caF}$ commute with each other, as has been shown in \cite[~Proposition 5.14]{zhang2016temperley}. This fact will follow from Lemma \ref{lem:commutativity lemma} below.

The Levin-Wen Hamiltonian is the formal commuting projector Hamiltonian
$$ H_{LW} = - \sum_{f \in \caF}  B_f. $$
A state $\omega : \caA \rightarrow \C$ is a frustration free ground state of $H_{LW}$ if $\omega(B_f) = 1$ for all $f \in \caF$.
The following Proposition has been proved in \cite{jones2023localtopologicalorderboundary} (see section 2.3 and Theorem 4.8 of that paper):
\begin{proposition} \label{prop:unique ffgs}
	The Levin-Wen Hamiltonian has a unique frustration free ground state which we denote by $\omega^{\I}$. This frustration free ground state is pure.
\end{proposition}
We present an independent proof at the end of section \ref{subsec:anyons on punctures} below. We let $(\pi^{\I}, \caH, \Omega)$ be the GNS triple of $\omega^{\I}$. Since $\omega^{\I}$ is pure, the \emph{vacuum representation} $\pi^{\I}$ is irreducible.

\subsection{Classification of irreducible anyon sectors}

The cone with apex at $a \in \R^2$, axis $\hat v \in \R^2$ of unit length, and opening angle $\theta \in (0, 2\pi)$ is the open subset of $\R^2$ given by
\begin{equation*}
 	\Lambda_{a, \hat v, \theta} := \{ x \in \R^2 \, : \, (x - a) \cdot \hat v > \norm{x-a} \cos (\theta/2)   \}. \quad\quad\quad\quad \adjincludegraphics[width=2.5cm,valign=c]{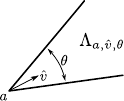}
\end{equation*}
Any subset $\Lambda \subset \R^2$ of this form will be called a cone.

\begin{definition} \label{def:anyonsector}
	A representation $\pi: \caA \rightarrow \mathcal{B}(\caH)$ satisfies the superselection criterion w.r.t. the vacuum representation $\pi^{\I}$ if for any cone $\Lambda$, there is a unitary equivalence
	$$ \pi|_{\Lambda} \simeq_{u.e.} \pi^{\I}|_{\Lambda} $$
	of representations of $\caA_{\Lambda}$. We will call such a representation an \emph{anyon representation}. A unitary equivalence class of anyon representations is called an \emph{anyon sector}. An anyon sector is called irreducible if any and therefore all of its representative representations is irreducible.
\end{definition}

\begin{theorem} \label{thm:classification of anyon sectors}
	The irreducible anyon sectors with respect to the vacuum representation of the Levin-Wen model over $\caC$ are in one-to-one correspondence with equivalence classes of simple objects of the unitary Drinfeld center $\ZC$.
\end{theorem}

The proof appears in Section \ref{subsec:proof of main theorem}.

\section{Tube algebras and skein modules} \label{sec:tube algebras and skein modules}

\subsection{Unitary fusion category and the unitary Drinfeld center} \label{subsec:UFC}

Throughout this paper we consider a fixed unitary fusion category (UFC) $\caC$ with its canonical spherical structure. See \cite{etingof2015tensor, turaev2017monoidal, penneysUFCnotes, penneys2018unitary} for general introductions. We write $\caC(x \rightarrow y)$ for the vector space of morphisms from $x \in \Ob \caC$ to $y \in \Ob \caC$. Each $\caC(x \rightarrow x)$ is equipped with a spherical trace $\tr$ and the quantum dimension of a non-zero object $x$ is $d_x := \tr \id_x > 0$. We choose a representative set of simple objects $\Irr \caC$ which contains the tensor unit $\I$. We put $\caD := \sqrt{ \sum_{a \in \Irr \caC} \, d_a^2 }$, so the global quantum dimension of $\caC$ is $\caD^2$.

The dual of an object $x \in \Ob \caC$ is denoted $x^*$ and its evaluation and coevaluation morphisms are $\ev_x : x^* \otimes x \rightarrow \I$ and $\coev_x : \I \rightarrow x \otimes x^*$. For each $a \in \Irr \caC$ there is a unique $\bar a \in \Irr \caC$ such that $\bar a$ is isomorphic to $a^*$.

The unitary structure consists of dagger maps $\dag: \caC(x \rightarrow y) \rightarrow \caC(y \rightarrow x)$ written as $ f \mapsto f^{\dag}$. Each morphism space $\caC(x \rightarrow y)$ is a Hilbert space with the trace inner product $(f, g)_{\tr} := \tr \lbrace f^{\dag} \circ g \rbrace$.

We freely use the graphical calculus to represent and manipulate morphisms of $\caC$. See \cite[Chapter~I.2]{turaev2017monoidal} for a good introduction. As usual, we suppress associators, unitors, and pivotal isomorphisms.

The \emph{unitary Drinfeld center}\footnote{The unitary Drinfeld center is often denoted by $Z^{\dag}(\caC)$, but since we work entirely in the unitary setting, we reserve the notation $\ZC$ for the unitary Drinfeld center.} of $\caC$ is the category $\ZC$ whose objects are pairs $(X, \sigma)$ of $X \in \Ob \caC$ and a unitary \emph{half-braiding} $\sigma : X \otimes - \implies - \otimes X$ whose component morphisms and their inverses we represent by
\begin{equation} \label{eq:half-braidings}
    \sigma_x = \adjincludegraphics[valign=c, height=1.0cm]{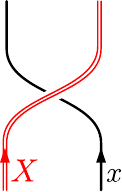}, \quad\quad \sigma_x^{-1} = \adjincludegraphics[valign=c, height=1.0cm]{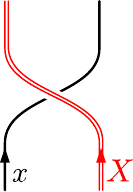}.
\end{equation}
The morphisms $\ZC( (X, \sigma) \rightarrow (X' , \sigma') )$ are precisely those morphisms $f \in \caC(X \rightarrow X')$ for which
\begin{equation}
    \adjincludegraphics[valign=c, height=1.5cm]{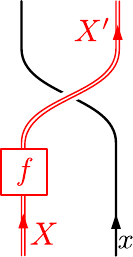} \quad  =  \quad \adjincludegraphics[valign=c, height=1.5cm]{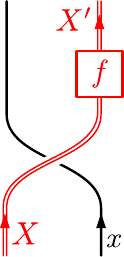}
\end{equation}
for all $x \in \Ob \caC$.

The unitary Drinfeld center of $\caC$ is a unitary modular tensor category whose pivotal and dagger structures coincide with that of $\caC$ \cite{muger2003subfactorsquantumdouble}. (See \cite[Proposition~2.3]{henriques2015categorifiedtracemoduletensor} for the statement about the pivotal structure in the case where $\caC$ is not strict pivotal.) In particular, we can choose a finite set of representative simple objects $\Irr \ZC$ containing the tensor unit $\I$, and for each $X \in \Irr \ZC$ there is a unique $\bar X \in \Irr \ZC$ which is isomorphic to $X^*$.

\subsection{String diagrams and skein modules} \label{subsec:string diagrams and skein modules}

We introduce skein modules on decorated surfaces. See \cite[Appendix A]{koenig2010quantum}, \cite[Section~2]{kirillov2011string}, \cite{walker2021universal}, \cite{walker2006tqft} for similar setups. We work in the category of piecewise linear manifolds throughout. A \emph{decorated 1-manifold} is a compact oriented 1-manifold (possibly with boundary) $\caN$ together with a \emph{decoration}, which consists of a finite collection of \emph{signed marked points} $m_{\caN} \subset \caN$. A \emph{decorated surface} is an oriented surface $\Sigma$ together with a decoration of its boundary $\partial \Sigma$. We will often make reference to the topology of the underlying surface of $\Sigma$ and for example say that $\Sigma$ is homeomorphic to a sphere with $m$ holes removed.

A \emph{string diagram} on $\Sigma$ is an embedded graph $\Gamma$ in $\Sigma$ whose edges are labelled by objects of $\caC$ and whose internal vertices (\ie not the marked boundary points) are labelled by morphisms of $\caC$ in the Hom space determined by the labels of the attaching edges by the usual rules of the graphical calculus. Using the forgetful functor $\ZC \to \caC$, string diagrams may also be labelled by objects and morphisms of $\ZC$.
The graph meets the boundary of $\Sigma$ transversally at the marked boundary points, with edges oriented into the boundary where they meet positive boundary points, and oriented away from the boundary where they meet negative boundary points. The labels of the edges attaching to $\partial \Sigma$ constitute the \emph{boundary condition} $b : m_{\partial \Sigma} \rightarrow \Ob \caC$ of the string diagram. We say the boundary condition is \emph{simple} if all these labels belong to $\Irr \caC$. We write $\Bd(\Sigma)$ for the set of connected boundary components of $\Sigma$, seen as decorated 1-manifolds. Given a boundary condition $b$ and a decorated submanifold $\caN \subset \partial \Sigma$, we write $b_{\caN}$ for the restriction of $b$ to the marked points in $\caN$. Given a string diagram $x$ on $\Sigma$, we also write $x_{\caN}$ for the boundary condition on $\caN$ induced by $x$.

Given a decorated 1-manifold $\caN$ we let $\hat \caN$ be the decorated 1-manifold obtained from $\caN$ by reversing orientation and flipping the signs of all marked points. Given a decorated surface $\Sigma$ we let $\hat \Sigma$ be the decorated surface obtained from $\Sigma$ by reversing the orientation and flipping the signs of all marked boundary points. Then $\partial \hat \Sigma = \widehat{\partial \Sigma}$ as decorated 1-manifolds. Given a string diagram $x$ on $\Sigma$ we obtain a string diagram $\hat x$ on $\hat \Sigma$ by reversing the orientations of all edges and replacing the label of each internal vertex by its dagger.

Let $\scrS(\Sigma;b)$ denote the set of all string diagrams on $\Sigma$ with boundary condition $b$. We define $A(\Sigma; b) := \C[ \scrS(\Sigma; b)] / \sim$, the space of all finite formal $\C$-linear combinations of string diagrams on $\Sigma$ with boundary condition $b$, modded out by the equivalence relation $\sim$ generated by
\begin{itemize}
    \item isotopy of string diagrams in $\Sigma$ keeping $\partial \Sigma$ fixed.
    \item local relations of the graphical calculus applied inside contractible disks in $\Sigma$, to string diagrams whose edges cross the boundary of this disk transversally.
\end{itemize}
We call $A(\Sigma) := \bigoplus_{b \in \caB(\Sigma)} A(\Sigma; b)$ the \emph{skein module} on $\Sigma$. Here, $\caB(\Sigma)$ is the finite set of all simple boundary conditions for string diagrams on $\Sigma$.

Given a string diagram $x$ on a decorated surface $\Sigma$, we write $[x]_{\Sigma} \in A(\Sigma)$ for its equivalence class in the skein module. When the decorated surface is clear from context, we often simply write $[x]$.

\begin{convention} \label{conv:coloured vertices}
    For any $x_1, \cdots, x_n \in \Ob \caC$, a pair of coloured vertices,
    \begin{equation*}
        \adjincludegraphics[valign=c, width = 2.0cm]{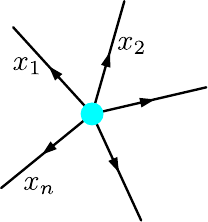} \otimes \adjincludegraphics[valign=c, width = 2.0cm]{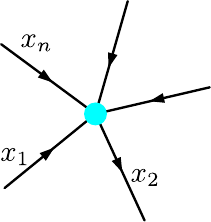} = \sum_{i} \adjincludegraphics[valign=c, width = 2.0cm]{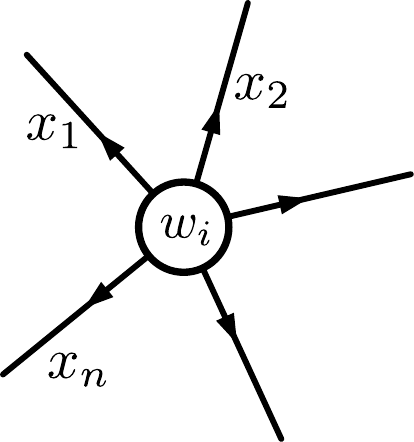} \otimes \adjincludegraphics[valign=c, width = 2.0cm]{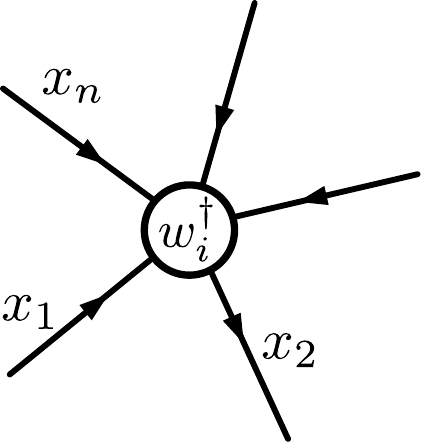},
    \end{equation*}
    will stand for a sum over an orthonormal basis $\{ w_i \}$ of $\caC(\I \rightarrow x_1 \otimes \cdots \otimes x_n)$ w.r.t. the trace inner product, and its dagger. This object does not depend on the choice of basis $\{ w_i \}$.
\end{convention}

With this convention, we have the following well known local relation of the graphical calculus (\cite[Lemma~1.1]{kirillov2010turaev}):
\begin{equation} \label{eq:decomposition into simples}
    \sum_{a \in \Irr \caC} \, d_a \, \adjincludegraphics[valign=c, width = 1.5cm]{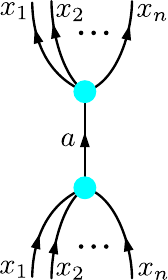} = \adjincludegraphics[valign=c, width = 1.5cm]{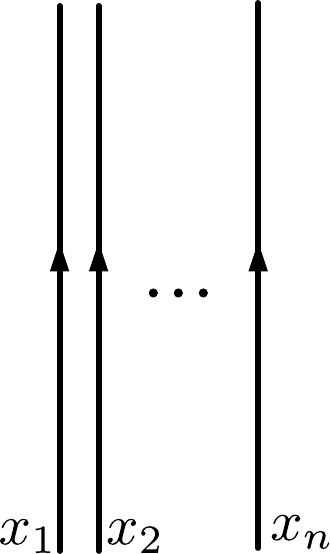} \quad.
\end{equation}

\begin{convention} \label{conv:dotted line}
    We allow string diagrams to have edges coloured by a dotted line, defined by
    \begin{equation} \label{eq:dotted line defined}
        \adjincludegraphics[valign=c, height = 0.8cm]{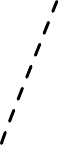} = \frac{1}{\caD^{2}} \sum_{a \in \Irr \caC} d_a \,\,\, \adjincludegraphics[valign=c, height = 0.8cm]{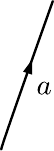}\quad.
    \end{equation}
    This dotted line satisfies the following local relations for any $x \in \Ob \caC$ and any $X \in \Ob \ZC$ (\cite[Corollary~3.5]{kirillov2011string}, \cite[Lemma~2.2]{kirillov2010turaev}):
    \begin{equation} \label{eq:normalization and absorption properties}
        \adjincludegraphics[valign=c, width = 1.5cm]{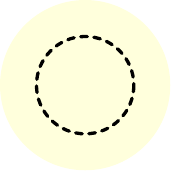} = \adjincludegraphics[valign=c, width = 1.5cm]{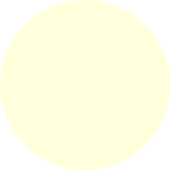}, \quad 
        \adjincludegraphics[valign=c, width = 1.5cm]{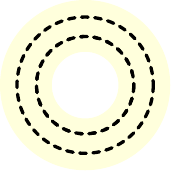} = \adjincludegraphics[valign=c, width = 1.5cm]{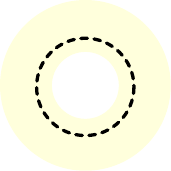}, \quad \adjincludegraphics[valign=c, width = 1.5cm]{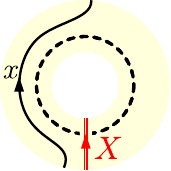} = \adjincludegraphics[valign=c, width = 1.5cm]{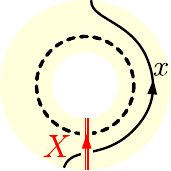}
    \end{equation}
    on embedded disks and annuli in any decorated surface. These are called the normalization, projector, and cloaking properties of the dotted line.
\end{convention}

\subsection{Gluing} \label{subsec:gluing}

Let $\Sigma$ be a decorated surface and let $\caN, \caM \subset \partial \Sigma$ be two disjoint decorated submanifolds of the boundary of $\Sigma$. We say $\caN$ can be glued to $\caM$  if $\caN \simeq \hat \caM$ as decorated 1-manifolds. If this is the case then there is an orientation reversing homeomorphism $\psi : \caN \rightarrow \caM$ which maps marked boundary points to marked boundary points. We let $\Sigma_{\psi}$ be the extended surface obtained from $\Sigma$ by identifying $\caN$ and $\caM$ under the map $\psi$. The resulting decorated surface $\Sigma_{\psi}$ depends only on the isotopy class of $\psi$ (\cite[Lemma~4.1.1]{bakalov2001lectures}).

Let the identification $\psi$ be fixed and write $\Sigma_{\gl}$ for the glued surface. Given a string diagram $x$ on $\Sigma$ we say the boundary conditions $x_{\caM}$ and $x_{\caN}$ match if all corresponding pairs of marked boundary points $(m, \psi(m)) \in m_{\caN} \times m_{\caM}$ are labelled by the same object. Suppose this is the case for $x$. Let us write $x_{\gl}$ for the string diagram obtained by interpreting $x$ as a string diagram on $\Sigma_{\gl}$:
\begin{equation*}
    \adjincludegraphics[valign=c, height=1.0cm]{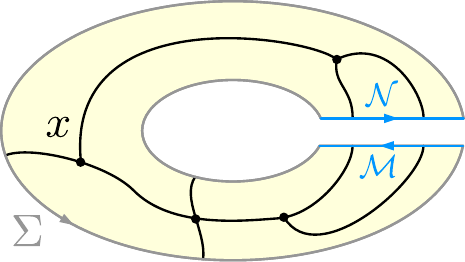} \, \xmapsto{\gl} \, \adjincludegraphics[valign=c, height=1.0cm]{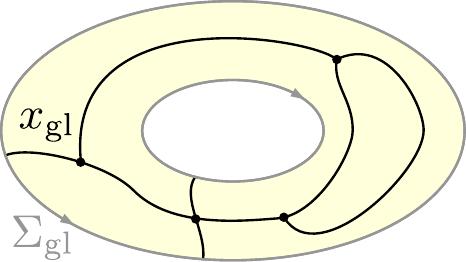} \quad.
\end{equation*}
Setting $[x_{\gl}] = 0$ if the boundary conditions do not match, the gluing of string diagrams descends to a well-defined gluing map $\gl : A(\Sigma) \rightarrow A(\Sigma_{\gl})$ by setting $ \gl([x]) := [x_{\gl}] $ and extending to $A(\Sigma)$ by linearity.

\subsection{Cylinder algebras and their actions on skein modules} \label{subsec:Tube algebras and their actions}

Let $\caN$ be a decorated 1-manifold, then $\caN$ is a disjoint union of decorated circles and decorated intervals. By taking the product with the closed unit interval $I$ we get a decorated surface $\caN \times I$ which is a union of decorated cylinders and disks. We let $\partial_b (\caN \times I) = (\hat \caN \times \{0\})$ be the \emph{bottom} and $\partial_t (\caN \times I) = ( \caN \times \{1\} )$ the \emph{top} of the decorated surface $\caN \times I$.

The skein module $A(\caN) := A(\caN \times I)$ has the structure of a $\rm C^*$-algebra with multiplication of two elements $[x]$ and $[y]$ defined by gluing the bottom of $[x]$ to the top of $[y]$, and with *-operation given by $[x]^* = [f(\hat x)]$ where $f : \caN \times I \rightarrow \caN \times I : (\theta, r) \mapsto (\theta, 1-r)$. That is, the string diagram is flipped upside down, orientations of all strands are reversed, and all morphisms are replaced by their daggers. We call $A(\caN)$ the \emph{cylinder algebra} on $\caN$. If $\caS$ is a decorated circle then $\Tube_{\caS} = A(\caS)$ is called the \emph{Tube algebra} on $\caS$ \cite{izumi2000, izumi2001examples, muger2003subfactorsquantumdouble}.

If $\caI$ is a decorated \emph{interval} then its marked points $m_{\caI} = (m_1, \cdots, m_n)$ have a linear order following the opposite orientation of $\caI$. For any labelling $\underline a : m_{\caI} \rightarrow \Irr \caC$ we write $\otimes \underline a = a(m_1)^{\sigma_1} \otimes \cdots \otimes a(m_n)^{\sigma_n}$ where $\sigma_i \in \{ +, - \}$ is the sign of the marked point $m_i$, and $a^+ = a$ while $a^- = a^*$ for any $a \in \Ob \caC$. Let $\chi = \bigoplus_{a \in \Irr \caC}a$. We further define the object
\begin{equation}
    \chi^{\otimes\caI} := \bigotimes_{i = 1}^{n} \chi^{\sigma_i} \simeq \bigoplus_{\underline a : m_{\caI} \rightarrow \Irr \caC} \otimes \underline a.
\end{equation}
There is an obvious identification $A(\caI) \simeq \caC( \chi^{\otimes\caI} \rightarrow \chi^{\otimes\caI} )$.

If $\caN' \subset \caN$, then $\caN' \times I \subset \caN \times I$, and we define an embedding $\iota : A(\caN') \hookrightarrow A(\caN)$ by
\begin{equation} \label{eq:cyclinder inclusion}
    \iota \left( \adjincludegraphics[valign=c, height = 1.0cm]{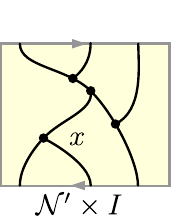} \right) = \adjincludegraphics[valign=c, height = 1.0cm]{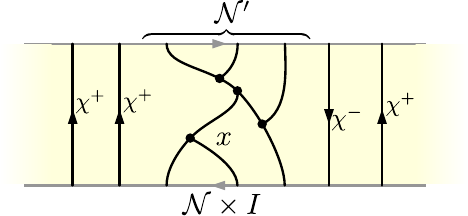}\quad.
\end{equation}

Let $\Sigma$ be a decorated surface with compact boundary and let $\caN \subset \partial \Sigma$ be a decorated submanifold of the boundary of $\Sigma$. Then there is a left action $\triangleright_{\caN}$ of $A(\caN)$ on $A(\Sigma)$ given by gluing the bottom of $[x] \in A(\caN)$ onto $\caN \subset \partial \Sigma$. Similarly, there is a left action $\triangleright^{\op}_{\caN}$ of $A(\hat \caN)^{\op}$ on $A(\Sigma)$ given by gluing the top of $[x] \in A(\hat \caN)$ onto $\caN \subset \partial \Sigma$. There is an isomorphism of $\rm C^*$-algebras $A(\hat \caN) \simeq A(\caN)^{\op}$ induced by flipping string diagrams upside down. We denote the image of $a \in A(\caN)$ under this isomorphism by $a^{\op}$. Then we have $a \, \triangleright^{\op}_{\caN} = a^{\op} \, \triangleright_{\hat \caN}$.

If $\Sigma$ has compact boundary, $\partial \Sigma$ is the disjoint union of connected boundary components $\caS \in \Bd(\Sigma)$, each of which is a decorated circle. We have
$$ A(\partial \Sigma) = \bigotimes_{\caS \in \Bd(\Sigma)} \, \Tube_{\caS}. $$

\subsection{Matrix units for \texorpdfstring{$\Tube_{\caS}$}{TubeS}} \label{ssubsec:matrix units for Tube_S}

Let $\caN$ be a decorated 1-manifold and pick a \emph{fiducial point} on $\caN$ which is distinct from all the marked points. A decorated 1-manifold equipped with a fiducial point is called an \emph{extended 1-manifold}. A decorated surface of which each connected boundary component is equipped with a fiducial point is call an \emph{extended surface}. Fiducial points of extended 1-manifolds will always be depicted by an X-mark in figures.

Let $\caS$ be an extended circle and assume that $\caS$ has $n$ marked points $m_{\caS} \subset \caS$. The fiducial point induces an enumeration of the marked points $m_{\caS} = (m_1, \cdots, m_n)$ starting at the fiducial point and following the \emph{opposite} orientation of $\caS$. Let $\sigma_i \in \{+, -\}$ denote the sign of the marked point $m_i$ for $i = 1, \cdots, n$.

For a boundary condition $\underline{a} : m_{\caS} \rightarrow \Irr \caC$ we write $\otimes \underline a = a(m_1)^{\sigma_1} \otimes \cdots \otimes a(m_n)^{\sigma_n}$. We also write $d_{\underline a} := d_{\otimes \underline a}$ for the quantum dimension of the boundary condition $\underline{a}$, and set $\chi^{\otimes\caS} = \bigoplus_{\underline{a} : m_{\caS} \rightarrow \Irr \caC} \otimes \underline{a}$.

For every $X \in \Irr \ZC$ and every boundary condition $\underline{a}$ we fix an orthonormal basis $\{ w^{X \, \underline a}_i \}_i$ of $\caC(X \rightarrow \otimes \underline a)$ with respect to the trace inner product. We will represent these morphisms graphically as follows:
$$ w^{X \underline{a}}_{i} = \adjincludegraphics[valign=c, height = 1.2cm]{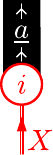}, \quad\quad (w^{X \underline{a}}_{i})^{\dag} =  \adjincludegraphics[valign=c, height = 1.2cm]{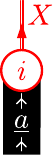}, \quad\quad \text{where} \quad \quad  \adjincludegraphics[valign=c, height = 1.0cm]{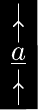} = \adjincludegraphics[valign=c, height = 1.0cm]{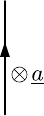} = \adjincludegraphics[valign=c, height = 1.0cm]{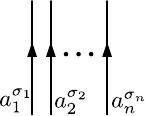}\quad.
$$
These morphisms satisfy the following useful identity:
\begin{lemma} \label{lem:useful identity}
    We have
    \begin{equation} \label{eq:useful identity}
        \adjincludegraphics[valign=c, height = 1.5cm]{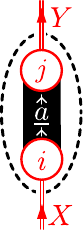} = \delta_{X, Y} \, \delta_{i, j} \frac{\id_X}{d_X}.
    \end{equation}
\end{lemma}

\begin{proof}
    The cloaking property of the dotted line makes the left hand side of the claimed equality an element of $\ZC( X \rightarrow Y )$. Since $X$ and $Y$ are irreducible it follows that the morphism represented by the left hand side equals $\delta_{X, Y} \, \lambda \,  \id_X$ for some $\lambda \in \C$. Assuming $X = Y$ and taking the trace yields
    $$ \lambda d_X = \adjincludegraphics[valign=c, height = 1.5cm]{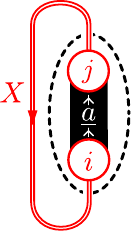} = \adjincludegraphics[valign=c, height = 1.5cm]{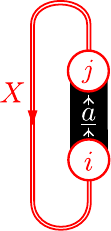} = \delta_{i, j}.  $$
    Here we used sphericity of $\caC$ in the second step to unwrap the dotted line.
\end{proof}

\begin{proposition} \label{prop:matrix units for Tube_n}
    The $\Tube_{\caS}$-elements
    \begin{equation*}
        E_{ \underline{b}, j;\underline{a}, i}^X := d_X \, \adjincludegraphics[valign=c, width = 2.0cm]{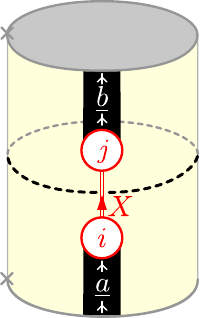}
    \end{equation*}
    form a complete set of matrix units for $\Tube_{\caS}$. We have
    \begin{equation} \label{eq:Tube_n matrix unit properties}
        E_{\underline{c}, k;\underline{b}', j'}^Y E_{\underline{b}, j;\underline{a}, i}^X = \delta_{X Y} \delta_{\underline{b} \, \underline{b}'} \delta_{j j'} \, E_{\underline{c}, k; \underline{a}, i}^X, \quad \text{and} \quad \sum_{X, \underline{a}, i} E^X_{\underline{a}, i; \underline{a}, i} = \id_{\Tube_{\caS}}.
    \end{equation}
    In particular, the minimal central projections of $\Tube_{\caS}$ are given by
    \begin{equation}\label{eq:def of P^X}
        P^X = \sum_{\underline{a}, i} E_{\underline{a}, i; \underline{a}, i}^X
    \end{equation}
    for each $X \in \Irr \ZC$.
\end{proposition}

We will say that a projector $p \in \Tube_{\caS}$ is of \emph{type} $X \in \Irr Z^{\dag}(
\caC)$ if it is dominated by the central projector $P^X$.

\begin{proofof}[Proposition \ref{prop:matrix units for Tube_n}]
    Using the cloaking property of one dotted loop to pull the other dotted loop to the front of the cylinder and using the relation Eq. \eqref{eq:useful identity} we find
    \begin{align*}
        &E_{\underline{c}, k; \underline{b}', j'}^Y E_{\underline{b}, j; \underline{a}, i}^X = \delta_{\underline{b} \, \underline{b}'} \, d_X d_Y \, \adjincludegraphics[valign=c, width = 1.7cm]{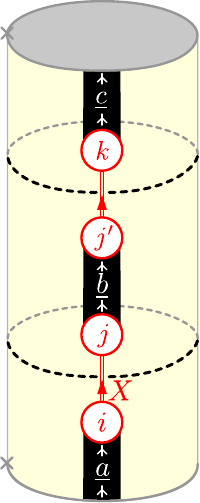} = \delta_{\underline{b} \, \underline{b}'} \, d_X d_Y \, \adjincludegraphics[valign=c, width = 1.7cm]{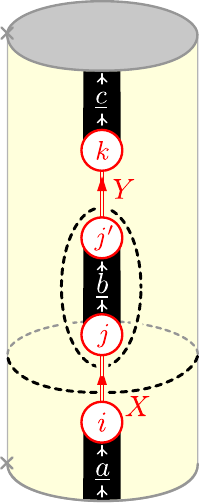} \\
        &= \delta_{X, Y} \delta_{\underline{b} \, \underline{b}'} \delta_{j \, j'} \, d_X \, \adjincludegraphics[valign=c, width = 1.7cm]{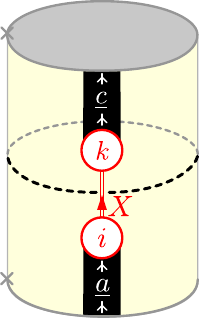} = \delta_{X Y} \delta_{\underline{b} \, \underline{b}'} \delta_{j \, j'} \, E_{\underline{c}, k; \underline{a}, i}^X \,\, ,
    \end{align*}
    so the $E_{\underline{b}, j ;\underline{a}, i}^X$ do indeed form a system of matrix units. Linear independence of the $E_{\underline{b}, j;\underline{a}, i}^X$ also follows from this.

    That this set of matrix units is complete follows from the identification of $\Tube_{\caS}$ with $$\bigoplus_{a \in \Irr \caC} \caC( a \otimes \chi^{\otimes\caS} \rightarrow \chi^{\otimes \caS} \otimes a  )
    \simeq \bigoplus_{X\in \Irr \ZC } \caC(\chi^{\otimes \caS} \to X^*) \otimes \caC(X \to \chi^{\otimes S}),
    $$ 
    where this isomorphism is provided in Lemma 7.4 of \cite{kirillov2010turaev}.
\end{proofof}

\subsection{The TQFT inner product} \label{subsec:TQFT inner product}

The skein modules $A(\Sigma)$ are the state spaces of the Turaev-Viro-Barrett-Westbury TQFT \cite{turaev1992state, barrett1996invariants}. Since $\caC$ is unitary one can use the TQFT partition function $Z$ to equip these skein modules with an inner product as follows  (see \cite{walker2021universal, walker2006tqft}).

Recall that the partition function $Z$ assigns a complex number $Z(M; x)$ to any 3-manifold $M$ with a string diagram $x \in \scrS( \partial M )$ on its boundary. In fact, $Z(M; x)$ depends only on the equivalence class of $x$ in $A(\partial M)$ so we have a linear map $Z_M : A(\partial M) \rightarrow \C$ for each 3-manifold $M$. Since any 3-manifold admits a handle decomposition, the partition function is completely determined by the following normalisation, gluing, and product properties:
\begin{itemize}
	\item For the 3-ball $B$ and a string diagram $x$ on the 2-sphere $\partial B$ we have 
            $$Z_B(x) = \ev(x) \, ,$$
		where $\ev(x)$ interprets $x$ as a diagram in the plane and uses the graphical calculus of $\caC$ to evaluate that diagram to a number (\cite[Theorem~2.4]{kirillov2010turaev}, \cite{barrett1996invariants}).

	\item If $M_{\gl}$ is obtained from a 3-manifold $M$ with boundary $\partial M = N \cup \overline N \cup L$ by gluing $N$ to $\overline N$ then \cite[~Section 4]{walker2021universal}
		\begin{equation} \label{eq:gluing formula}
			Z_{M_{\gl}}(x_{\gl}) = \sum_{e} \, \frac{Z_{M}(x \cup e \cup \hat e)}{(e, e)} \, ,
		\end{equation}
		where the sum ranges over an orthogonal basis of $A(N ; b)$, with $b$ the boundary condition induced by $x$.

    \item If $M = M_1 \sqcup M_2$ is a disjoint union and $x = x_1 \sqcup x_2$ is a string diagram on $\partial M$ with $x_1$ on $\partial M_1$ and $x_2$ on $\partial M_2$ then
    $$ Z_M(x) = Z_{M_1}(x_1) Z_{M_2}(x_2). $$
\end{itemize}

We can use the partition function to define an inner product on the skein module $A(\Sigma)$ as follows. Whenever $\Sigma$ is a surface, $\Sigma \times I$ is defined to be the pinched product of $\Sigma$ and $I = [0, 1]$ (\ie the Cartesian product is pinched by retracting $\partial \Sigma \times I \simeq \partial \Sigma$ so that $\partial( \Sigma \times I ) = \Sigma \cup \hat \Sigma$). Given $[x] \in A(\Sigma; b)$ and $[y] \in A(\Sigma; b')$ we put
$$ ( [x], [y] )_{A(\Sigma)} := \delta_{b \, b'} \, Z_{\Sigma \times I}( \hat x \cup y ),  $$
where $\hat x \cup y$ is the string diagram on $\partial( \Sigma \times I)$ consisting of string diagram $\hat x$ on $\hat \Sigma$ and string diagram $y$ on $\Sigma$. 
This yields a well-defined inner product on $A(\Sigma)$ which we will call the TQFT inner product.

It is straightforward to verify that the Tube actions on the various boundary components of $\Sigma$ are *-actions w.r.t. the TQFT inner product. This gives $A(\Sigma)$ the structure of a unitary $A(\partial \Sigma)$-module.

\subsection{Characterisation of skein modules on punctured spheres} \label{subsec:skein modules on punctured spheres}

We characterise the skein module $A(\Sigma)$ where $\Sigma$ is an extended surface homeomorphic to the sphere with $m$ holes cut out.
The characterization will depend on a choice of \emph{anchor} for the extended surface $\Sigma$. (see \cite{henriques2016planar} for a similar notion). An anchor $\anchor$ for $\Sigma$ is an embedded graph in $\Sigma$ consisting of one vertex called the anchor point, and for each boundary component $\caS \in \Bd(\Sigma)$ an edge running from the anchor point to that boundary component, attaching transversally at an attachment point on $\caS$ which is distinct from the fiducial point on that boundary component. The edges of the graph are moreover linearly ordered going clockwise around the anchor point. This induces an enumeration $\{ \caS^{\anchorino}_{\kappa} \}_{\kappa = 1}^m$ of $\Bd(\Sigma)$ which we call the enumeration induced by $\anchor$. Two anchors are equivalent if one can be deformed into the other by isotopy in $\Sigma$ which keeps the fiducial points fixed on all boundary components.

The characterization of $A(\Sigma)$ is given in terms of morphism spaces of $\caC$ and $\ZC$, all of which we regard as Hilbert spaces equipped with the trace inner product.
For any extended circle $\caS$ and for any $X \in \Ob \ZC$ the space $\caC( X \rightarrow \chi^{\otimes\caS} )$ is a unitary left $\Tube_{\caS}$-module with action given by
\begin{equation} \label{eq:Tube_n_module}
    \adjincludegraphics[valign=c, width = 2.0cm]{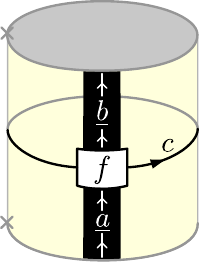} \,\,\, \triangleright \,\,\,  \adjincludegraphics[valign=c, width = 0.7cm]{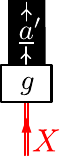} = \delta_{\underline{a}, \underline{a}'} \times \,\adjincludegraphics[valign=c, width = 2.5cm]{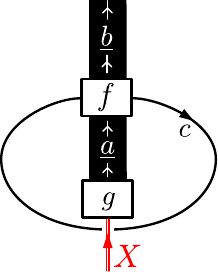} \quad.
\end{equation}
If $X \in \Irr \ZC$, then the module $\caC( X \rightarrow \chi^{\otimes\caS} )$ is irreducible.

We will show that $A(\Sigma)$ is isomorphic as a unitary $A(\partial \Sigma)$-module to
\begin{align} \label{eq:definition of algebraic module}
    \begin{split}
    \caC_{\anchorino}(\Sigma) :=  \bigoplus_{X_1, \ldots, X_m \in \Irr ( \ZC )}  \ZC( &\I \rightarrow X_1 \otimes \cdots \otimes X_m) \\ & \otimes \caC( X_1 \rightarrow \chi^{\otimes\caS^{\anchorino}_1} ) \otimes \cdots \otimes \caC(X_m \rightarrow \chi^{\otimes\caS^{\anchorino}_m}).
    \end{split}
\end{align}
For each $\caS^{\anchorino}_{\kappa}$ there is a unitary $\Tube_{\caS^{\anchorino}_{\kappa}}$ action on $\caC_{\anchorino}(\Sigma)$ given by Eq. \eqref{eq:Tube_n_module} which we denote by $\triangleright_{\kappa}$. Recalling that $A(\partial\Sigma) = \bigotimes_{\caS \in \Bd(\Sigma)} \, \Tube_{\caS}$, we see that these actions give $\caC_{\anchor}(\Sigma)$ the structure of a unitary $A(\partial \Sigma)$-module.

Given an anchor $\anchor$ for $\Sigma$ we define a linear map $\Phi_{\Sigma}^{\anchorino} : \caC_{\anchorino}(\Sigma) \rightarrow A(\Sigma)$ by
\begin{equation} \label{eq:Phi_anchor_defined}
    \Phi_{\Sigma}^{\anchorino} : \al \otimes w_1 \otimes \cdots \otimes w_m \mapsto \left( \prod_{\kappa=1}^m d_{X_{\kappa}} \right)^{1/2} \, \caD^{m-1} \,\,\, \adjincludegraphics[valign=c, width = 6.0cm]{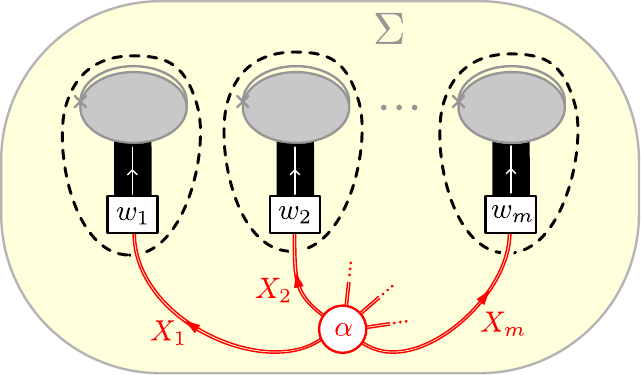},
\end{equation}
where in the string diagram on the right hand side, the morphism $\al$ sits at the anchor point of $\anchor$, and the strand labelled $X_{\kappa}$ lies along the $\kappa^{\mathrm{th}}$ edge of $\anchor$ until it resolves into the morphisms $w_{\kappa}$ near the boundary component $\caS^{\anchorino}_{\kappa}$. If two anchors $\anchor \sim \anchor'$ are equivalent, then $\Phi^{\anchorino}_{\Sigma} = \Phi_{\Sigma}^{\anchorino'}$.

\begin{proposition} \label{prop:characterisation of skein modules}
    Let $\Sigma$ be an extended surface homeomorphic to the sphere with $m$ holes cut out, and let $\anchor$ be an anchor for $\Sigma$. The map
    \begin{equation*}
    		\Phi_{\Sigma}^{\anchorino} :  \caC_{\anchorino}(\Sigma) \rightarrow A(\Sigma)
    \end{equation*}
    is a unitary isomorphism of  $A(\partial \Sigma)$-modules.
\end{proposition}

The proof appears in Appendix \ref{app:proof of characterization of skein modules}.

\begin{remark}
    There is a well-known description of Tube algebras in terms of $\chi$ and the adjoint of the forgetful functor $F: \ZC \to \caC$
    (see \cite{muger2003subfactorsquantumdouble}, and \cite{kawagoe2024levin,kawagoe2020microscopic} for a nice overview in the context of Levin-Wen models). 
    Following \cite{henriques2015categorifiedtracemoduletensor}, we denote the adjunction $F \dashv \Tr$.
    Then
    \begin{equation*}
        \Tube_{\caS} \simeq \End_{\ZC}(\Tr \chi^{\otimes\caS}),
    \end{equation*}
    with the isomorphism provided by the adjunction once the vector space of the Tube algebra is realised as
    $\Tube_{\caS} \simeq \bigoplus_{a\in \Irr \caC} \caC(a \otimes \chi^{\otimes\caS} \otimes a^* \to \chi^{\otimes\caS})$.
    On  $\Tube_{\caS}$-representations the adjunction gives the isomorphisms
    \begin{equation*}
        \caC(X \to \chi^{\otimes\caS}) \simeq \ZC (X \to \Tr \chi^{\otimes\caS}),
    \end{equation*}
    where the right hand side is a natural representation of $\End_{\ZC}(\Tr \chi^{\otimes\caS})$ given by composition of morphisms.    
    In this description, \Cref{prop:characterisation of skein modules} can be recast as 
    \begin{equation}\label{eq:m-fold Tr}
        A(\Sigma) \simeq
        \ZC(\1 \to \Tr \chi^{\otimes\caS_1} \otimes \cdots \otimes \Tr \chi^{\otimes\caS_m}).
    \end{equation}
\end{remark}

For later use, we investigate how the map $\Phi_{\Sigma}^{\anchorino}$ changes when we modify the anchor using Dehn twists around the holes of $\Sigma$.

\begin{lemma} \label{lem:anchors related by Dehn twists}
    Let $\Sigma$ be a decorated surface homeomorphic to the sphere with $m$ holes cut out. Let $\anchor$ and $\anchor'$ be anchors for $\Sigma$ so that $\anchor'$ is obtained from $\anchor$ by a Dehn twist around boundary component $\caS_{\kappa}^{\anchorino}$. Then
    \begin{equation}
        \Phi_{\Sigma}^{\anchorino'} \big( \al \otimes w_1 \otimes \cdots \otimes w_m  \big) = \theta_{X_{\kappa}} \times \Phi_{\Sigma}^{\anchorino} \big( \al \otimes w_1 \otimes \cdots \otimes w_m  \big)
    \end{equation}
    for any $\al \in \ZC(\I \rightarrow X_1 \otimes \cdots \otimes X_m)$ and any $w_{\kappa} \in \caC( X_{\kappa} \rightarrow \chi^{\otimes \caS^{\anchorino}_{\kappa}} )$ for $\kappa = 1, \cdots, m$.

    Here $\theta_{X_{\kappa}} \in U(1)$ is the twist of $X_{\kappa}\in \Irr \ZC$, defined for any $X \in\Ob \ZC$ by
    \begin{equation} \label{eq:twists defined}
        \adjincludegraphics[valign=c, height = 0.8cm]{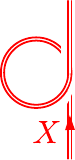} = \theta_X \times \adjincludegraphics[valign=c, height = 0.8cm]{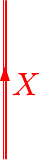}.
    \end{equation}
\end{lemma}

\begin{proof}
    Writing $K = \left( \prod_{\kappa=1}^m d_{X_{\kappa}} \right)^{1/2} \, \caD^{m-1}$ we find
    \begin{align*}
        \Phi_{\Sigma}^{\anchorino'} \big( \al \otimes w_1 \otimes &\cdots \otimes w_m  \big) = K \times \adjincludegraphics[valign=c, height=1.0cm]{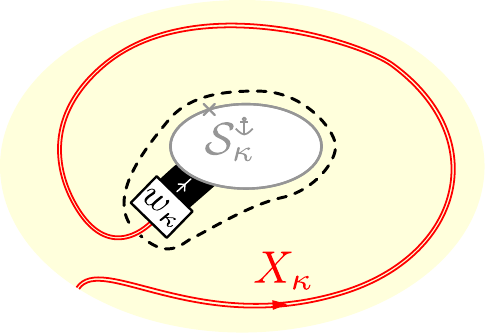} = K \times \adjincludegraphics[valign=c, height=1.0cm]{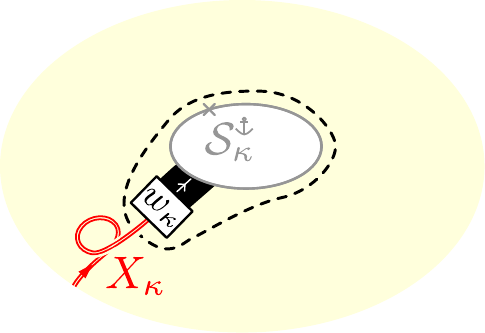} \\
        &= K \times \theta_{X_{\kappa}} \times \adjincludegraphics[valign=c, height=1.0cm]{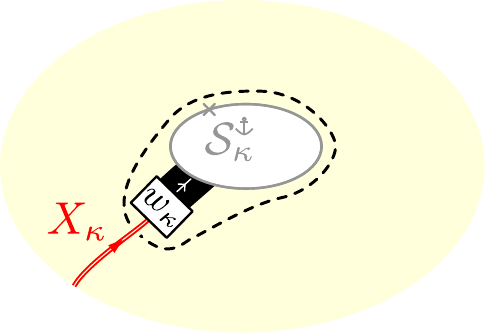} = \theta_{X_{\kappa}} \times \Phi_{\Sigma}^{\anchorino} \big( \al \otimes w_1 \otimes \cdots \otimes w_m  \big),
    \end{align*}
    where we used the cloaking property of the dotted line in the second step.
\end{proof}

\section{Skein subspaces} \label{sec:skein subspaces}

\subsection{Regions and their string-net subspaces} \label{subsec:regions and string-net subspaces}

Let $C^{\Z^2}$ be the cell complex with vertices $C_0^{\Z^2} = \Z^2$, all edges $\caE$ of $\Z^2$ thought of as open intervals in $\R^2$ connecting neighbouring vertices, and all faces $\caF$ of $\Z^2$ thought of as open unit squares in $\R^2$.

A \emph{region} is a subcomplex $C$ of $C^{\Z^2}$. In particular, if $C$ contains a face $f$ then it also contains all edges on the boundary of that face, and if $C$ contains an edge $e$ then it also contains the vertices at the endpoints of that edge. We write $C_0, C_1$ for the 0-skeleton and 1-skeleton of $C$ respectively,  considered as subcomplexes of $C$. We also write $V_C$, $E_C$, and $F_C$ for the set of vertices, edges, and faces of $C$ respectively. Also, $\vec E_C = \{ e \in \vec \caE \,: \, e \in E_C\, \text{or} \, \bar e \in E_C \}$ is the set of oriented edges of $C$, and $\caA_C = \caA_{C_0}$ is the algebra of observables supported on $C$.
The \emph{oriented boundary} of $C$ is
$$ \vec \partial C := \{ e \in \vec \caE \, : \, \partial_{\f} e \in V_C \, \text{and} \, \partial_{\ii} e \not\in V_C \}, $$
and we write
$$ \vec \caE_C := \vec \partial C \cup \vec E_C = \{ e \in \vec \caE \, : \, \partial_{\f} e \in V_C \}. $$

A \emph{labelling} of $C$ is an assignment $l : \vec \caE_C \rightarrow \Irr \caC$ of simple objects from $\Irr \caC$ to oriented edges of $C$. An \emph{internal labelling} of $C$ is an assignment $l : \vec E_C \rightarrow \Irr \caC$ and a \emph{boundary labelling} of $C$ is an assignment $l : \vec \partial C \rightarrow \Irr \caC$. We denote by $\caL(C), \caL^{\inn}(C)$ and $\caL^{\partial}(C)$ respectively the sets of labellings, internal labellings, and boundary labellings of $C$.

An (internal) labelling $l$ of $C$ is called a (internal) \emph{string-net labelling} if $l(e) = l(\bar e)$ for all $e \in \vec E_C$. We denote by $\caL_{SN}(G)$ the set of all string-net labellings of $C$ and by $\caL_{SN}^{\inn}(C)$ the set of all internal string-net labellings of $C$.

We can identify any vertex $v$ with the subcomplex $C^v=C^v_0 = \{v\}$. Then $\vec \partial v = \{ e_1, e_2, e_3, e_4 \}$ where $e_1, e_2, e_3, e_4$ are the four oriented edges pointing into $v$. Given a boundary labelling $l : \vec \partial v \rightarrow \Irr \caC$ we let $\caH_v(l) \subset \caH_v$ be the subspace of $\caH_v$ corresponding to the direct summand $\caC( l(e_1) \otimes l(e_2) \rightarrow l(e_3) \otimes l(e_4) )$. Then $\caH_v = \bigoplus_{l \in \caL^{\partial}(v)} \caH_v(l)$.

Given a finite region $C$ we write
$$ \caH_C := \caH_{V_C} = \bigoplus_{l \in \caL(C)} \caH_C(l), \quad \text{where} \quad \caH_C(l) := \bigotimes_{v \in V_C} \, \caH_v(l|_{\vec \partial v}). $$
The \emph{string-net subspace} associated to $C$ is
$$ H_{C_1} := \bigoplus_{l \in \caL_{SN}(C)} \caH_C(l) = \left( \prod_{e \in E_C} A_e \right) \caH_C, $$
\ie the string-net constraints are imposed along all edges of $C$. As the notation reflects, the string-net subspace of $C$ depends only on the 1-skeleton $C_1$.

We further decompose the string-net subspace according to boundary labellings. For $b \in \caL^{\partial}(C)$ we denote by $\caL^b(C)$ the set of labellings $l$ of $C$ such that $l|_{\vec \partial C} = b$. Similarly, $\caL^b_{SN}(C)$ is the set of string-net labellings $l$ of $C$ such that $l|_{\vec \partial C} = b$. Then we have
$$ H_{C_1} = \bigoplus_{b \in \caL^{\partial}(C)} H_{C_1}^b \quad \text{with} \quad H_{C_1}^b := \bigoplus_{l \in \caL^b_{SN}(C)} \caH_C(l). $$

\subsection{Extended surfaces assigned to regions}

To any region $C$ we assign an extended surface $\Sigma_C$ as follows. The surface $\Sigma_C$ is the 1/3-fattening of $C \subset \R^2$. That is, $\Sigma_C := \{ v \in \R^2 \, : \, \dist(v, C) \leq 1/3 \}$ where $\dist(v, C) = \inf_{w \in C} \, \norm{w - v}_{\infty}$ and we regard $C$ as a closed subset of the plane. The marked boundary points are the points of intersection of $\partial \Sigma_C$ with edges of $\Z^2$. Such a marked boundary point has \emph{positive} sign if the corresponding edge $e \in \caE$ (which is oriented towards the top left) points out of the surface $\Sigma_C$, and it has \emph{negative} sign if the corresponding edge $e \in \caE$ points into the surface. Finally, each connected boundary component $\caN \in \Bd(\Sigma_C)$ is the union of some horizontal and vertical line segments. If $\caN$ is compact then we place the fiducial point at the upper extremity of top leftmost vertical segment. If $\caN$ is not compact we pick an arbitrary fiducial point.
When $C$ is finite and connected there is a distinct \emph{outer} boundary component of $\Sigma_C$. All other boundary components are referred to as \emph{inner}. 

Note that $\Sigma_C$ and $\Sigma_{C_1}$ have the same marked boundary points. Indeed, the surface $\Sigma_{C_1}$ has all the same boundary components as $\Sigma_C$, and in addition a boundary component for each face of $C$, but there are no marked boundary points on the boundary components of $\Sigma_{C_1}$ corresponding to these faces.

Boundary conditions for string diagrams on $\Sigma_C$ or $\Sigma_{C_1}$ are equivalent to boundary labellings of the region $C$ itself, and we will identify these concepts.

\subsection{Isomorphism of string-net subspaces with skein modules} \label{subsec:isomorphism of string-net with skein}

Let $C$ be a finite region. There is a natural way to regard product vectors in the string-net subspace $H_{C_1}$ as string diagrams on $\Sigma_{C_1}$. For a given string-net labelling $l \in \caL_{SN}(C)$, consider unit vectors $f_v \in \caH_v(l)$ for all $v\in C_0$. Recall that $f_v$ is a morphism in $\caC$. Denoting $f=(f_v)_{v\in C_0}$, we define  $\phi_{f} = \bigotimes_{v \in C_0} f_v \in \caH_C(l)$. By construction, $\phi_f$ belongs to the string-net subspace $H_{C_1}$, and $H_{C_1}$ is spanned by such product vectors. 
We define the corresponding string diagram $x_f$ on $\Sigma_{C_1}$ whose graph is the intersection of the graph of $\Z^2$ with the surface $\Sigma_{C_1}$ (with edges directed towards the top right), which has edges labelled according to $l$, and vertices $v \in C_0$ labelled by the morphisms $f_v$. 
We define $\pi_{C_1} : H_{C_1} \rightarrow A(\Sigma_{C_1})$ by
\begin{equation} \label{eq:pi_C defined}
    \pi_{C_1}( \phi_f ) = [ x_f ]_{\Sigma_{C_1}},
\end{equation} 
where we interpret the string diagram $x_f$ as a string diagram on $\Sigma_{C_1}$.

The following lemma extends Lemma 5.3 of \cite{kirillov2011string} to include surfaces with boundary. Their proof generalizes mutatis mutandis. 
\begin{lemma}[\cite{kirillov2011string}]
\label{lem:string-net isomorphism}
    For any finite region $C$, the map $\pi_{C_1} : H_{C_1} \rightarrow A(\Sigma_{C_1})$ is an isomorphism of vector spaces.
\end{lemma}

\begin{convention} \label{conv:graphical representation}
    We will freely use the isomorphism $\pi_{C_1}$ to represent states in $H_{C_1}$ by string diagrams on $\Sigma_{C_1}$, where we can use the graphical calculus. Whenever a string diagram on a surface $\Sigma_{C_1}$ is interpreted as an element of $H_{C_1}$, this will implicitly be done using $\pi_{C_1}$.

    If $x$ is a string diagram on $\Sigma_{C_1}$, then we write $\phi_x := \pi_{C_1}^{-1}( [x] )$ for the corresponding vector in the string-net subspace $H_{C_1}$.
\end{convention}

Let us also introduce the map $\sigma_{C_1} : H_{C_1} \rightarrow A(\Sigma_{C_1})$ defined on  product vectors by
\begin{equation} \label{eq:sigma_C defined}
    \sigma_{C_1}( \phi_f ) := \, d_{\partial l}^{-1/4} \, [ x_f ]_{\Sigma_{C_1}},
\end{equation}
and extended linearly to all of $H_{C_1}$. Here, $d_{\partial l} = \prod_{e \in \vec \partial C} d_{l(e)}$ is the product of the quantum dimensions of the objects assigned by $l$ to the boundary of $C$. Note that by construction we have $\sigma_{C}( H_{C_1}^b ) \subset A(\Sigma_{C_1}; b)$ for any boundary condition $b$.

The reason for modifying the map $\pi_{C_1}$ with weights as in Eq. \eqref{eq:sigma_C defined} is to obtain an isomorphism of Hilbert spaces, see \Cref{lem:TQFT vs skein on string-net space}. 

\subsection{Actions of Tube algebras} \label{subsec:Tube actions on collars}

Let $C$ be a region such that $\Sigma_C$ has compact boundary. To each connected boundary component $\caS \in \Bd(\Sigma_C)$ we associate a \emph{collar region} $C^{\caS}$, defined to be the smallest subregion of $C$ such that $\caS \in \Bd( \Sigma_{C^{\caS}})$. Note that if $C$ contains a face $f$ and $C'$ is obtained from $C$ by removing $f$, then $\Sigma_{C'}$ has a connected boundary component at $f$ which we denote by $\caS_f$. We denote the corresponding collar region by $C^f := C^{\caS_f}$.

More generally, if $\caN \subset \caS \in \Bd(\Sigma_C)$ then we let $C^{\caN}$ be the smallest subregion of $C$ so that $\caN \subset \partial \Sigma_{C^{\caN}}$. The region $C^{\caN}$ does not depend on $\caS$ in the sense that if $\caN \subset \caS' \inn \Bd(\Sigma_{C'})$ for some other region $C'$, then $C^{\caN}$ is also the smallest subregion of $C'$ such that $\caN \subset \partial \Sigma_{C^{\caN}}$.

For each $\caS \in \Bd(\Sigma_C)$ we define a representation $\frt_{\caS}$ of $\Tube_{\caS}$ on $\caH_{C^{\caS}}$ as follows. First, for any $a \in \Tube_{\caS}$, the operator $\frt_{\caS}(a)$ acts as zero on the orthogonal complement of the string-net subspace $H_{C^{\caS}}$. That is, any $\frt_{\caS}(a)$ enforces string-net constraints on the collar region $C^{\caS}$.

On the string-net subspace, the action is given by
\begin{equation} \label{eq:frt_caS defined}
    \frt_{\caS}(a) \, \phi := \sigma_{C^{\caS}}^{-1} \big( a \triangleright_{\caS} \sigma_{C^{\caS}}( \phi )  \big)
\end{equation}
for any $\phi \in H_{C^{\caS}}$.
It follows from Lemma \ref{lem:TQFT vs skein on string-net space} that each $\frt_{\caS}$ is a *-representation of $\Tube_{\caS}$. Let us also define a *-representation of $\Tube_{\hat \caS}^{\op}$ on $\caH_{C^{\caS}}$ by 
\begin{equation}
    \frt^{\op}_{\caS}(a) := \frt_{\caS}(a^{\op})
\end{equation}
for any $a \in \Tube_{\hat \caS}$.

Unpacking definitions and making use of Convention \ref{conv:graphical representation}, the action of $\frt_{\caS}([y])$ on $H_{C^{\caS}}$ is represented graphically by
\begin{equation} \label{eq:graphical frt action}
    \frt_{\caS} \left( \adjincludegraphics[valign=c, width = 3.0cm]{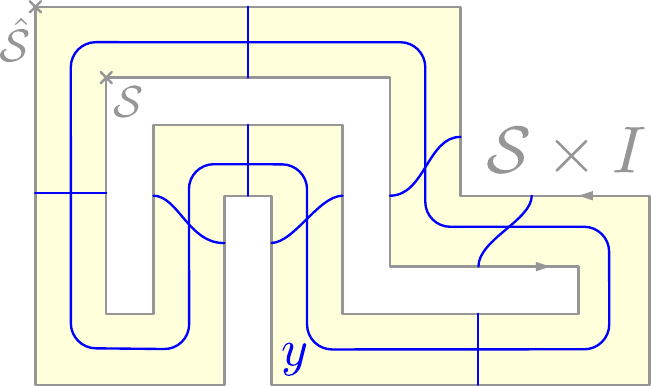}  \right) \,\,\adjincludegraphics[valign=c, width = 3.0cm]{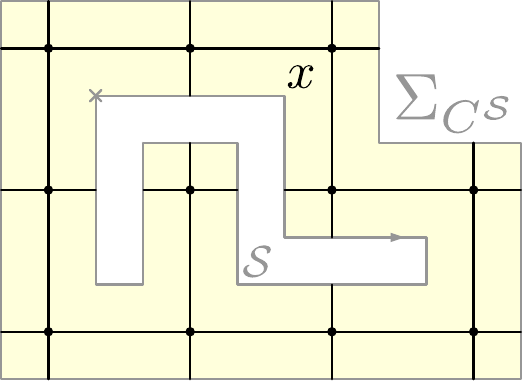}  = \delta_{x_{\caS}, y_{\hat \caS}} \, \left( \frac{d_{y_{\caS}}}{d_{x_{\caS}}} \right)^{1/4} \, \adjincludegraphics[valign=c, width = 3.0cm]{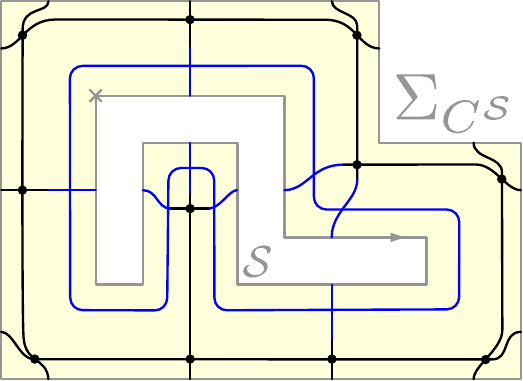} \,\, .
\end{equation}
Here we represented $\caS \times I$ as a subset of the plane with the bottom boundary component identified with the outer boundary in the figure.

More generally, if $\caN \subset \caS \in \Bd(\Sigma_C)$ then we define a *-representation $\frt_{\caN}$ of $A(\caN)$ on $\caH_{C^{\caN}}$ by letting $\frt_{\caN}(a)$ annihilate the orthogonal complement of the string-net subspace $H_{C^{\caN}}$, and letting $\frt_{\caN}(a)$ act on this subspace by
\begin{equation} \label{eq:frt_caN defined)}
    \frt_{\caN}(a) \phi := \sigma_{C^{\caN}}^{-1} \big(  a \triangleright_{\caN} \sigma_{C^{\caN}}( \phi )  \big)
\end{equation}
for any $\phi \in H_{C^{\caN}}$. We also define the $A(\hat \caN)^{\op}$-action $\frt_{\caN}^{\op}(a) := \frt_{\caN}(a^{\op})$. It follows immediately from the definitions that
\begin{lemma} \label{lem:frt inclusions}
    If $\caN \subset \caS \in \Bd(\Sigma_C)$ for some finite region $C$ then
    $$ \frt_{\caN}(a) \frt_{\caS}(\id) = \frt_{\caS}(\id) \frt_{\caN}(a) = \frt_{\caS}(\iota(a)), $$
    where $\iota : A(\caN) \rightarrow A(\caS)$ is the inclusion \eqref{eq:cyclinder inclusion}.
\end{lemma}

\subsection{Definition of \texorpdfstring{$B_f$}{Bf} projectors} \label{subsec:Bf defined}

For a boundary component $\caS_f$ that corresponds to a face $f \in \caF$ we write $\frt_f := \frt_{\caS_f}$ for the corresponding $\Tube_{\caS_f}$-representation on $\caH_{C^f}$. We define the orthogonal projector
\begin{equation} \label{eq:Bf defined}
    B_f = \frt_f( P^{\I} ), \quad \text{where} \quad  P^{\I} = \adjincludegraphics[valign=c, width = 0.8cm]{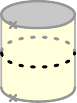} \in \Tube_{\caS_f}.
\end{equation}
That is, $B_f$ enforces string-net constraints and inserts a dotted loop around the face $f$.
Using Convention \ref{conv:graphical representation}, we can also describe the action of $B_f$ on $H_{C^f}$ graphically as
\begin{equation} \label{eq:B_f graphical definition}
        \adjincludegraphics[valign=c, width = 2.0cm]{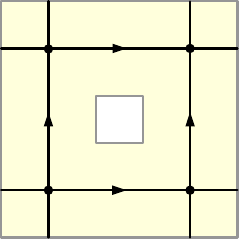} \,\,\, \xmapsto{B_f} \,\,\, \adjincludegraphics[valign=c, width = 2.0cm]{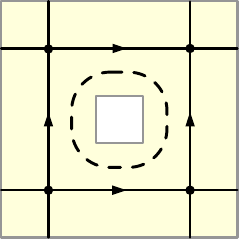} \,\, .
\end{equation}

\subsection{Skein subspaces} \label{subsec:skein subspaces defined}

Let $C$ be a finite region. We define its associated \emph{skein subspace} by
\begin{equation} \label{eq:skein subspace defined}
    H_C := B_C H_{C_1},
\end{equation}
where $B_C := \prod_{f \in F_C} B_f$. The skein subspace $H_C$ is a Hilbert space equipped with the inner product inherited as a subspace $\caH_{C}$.

The inclusion $\Sigma_{C_1}\subset \Sigma_C$ induces a map
$$\iota_C : A(\Sigma_{C_1}) \rightarrow A(\Sigma_C) : [y]_{\Sigma_{C_1}}  \mapsto [y]_{\Sigma_C}$$ 
which is well-defined since local relations in $\Sigma_{C_1}$ are also local relations in $\Sigma_C$, and isotopy fixing the boundary of $\Sigma_{C_1}$ extends to isotopy fixing the boundary of $\Sigma_C$. The map $\iota_C$ is surjective because any string diagram on $\Sigma_C$ can be deformed by isotopy to be contained in $\Sigma_{C_1}\subset \Sigma_C$.

Recalling definitions \eqref{eq:pi_C defined}\eqref{eq:sigma_C defined}, define 
\begin{align*}
    \pi_C = \iota_C \circ \pi_{C_1} \,\qquad \qquad &: H_{C_1} \to A(\Sigma_C), \\
    \sigma_C =   \caD^{-|F_C|} \times \iota_C \circ \sigma_{C_1}  &: H_{C_1} \to A(\Sigma_C).
\end{align*}
The normalization of $\sigma_C$ is so that its restriction to $H_C$ yields an isomorphism of Hilbert spaces.

\begin{proposition} \label{prop:isomorphism of skein subspace and skein module}
    The map $\sigma_C|_{H_C} : H_C \rightarrow A(\Sigma_C)$ is a unitary isomorphism. Moreover, for any connected boundary component $\caS \in \Bd(\Sigma_C)$ of $\Sigma_C$, we have 
    $$ a \triangleright_{\caS} \sigma_C(\psi) = \sigma_C \big( \frt_{\caS}(a) \psi \big)$$
    for all $\psi \in H_C$ and all $a \in \Tube_{\caS}$. In particular, the skein subspace $H_C$ equipped with the $\Tube$-actions given by the representations $\{  \frt_{\caS }\}$ is a unitary $A(\partial \Sigma_C)$-module, and $\sigma_C|_{H_C}$ is an intertwiner of unitary $A(\partial \Sigma_C)$-modules.
\end{proposition}
The proof appears in Appendix \ref{app:proof of skein subspace isomorphism}. As a corollary, we obtain the \emph{commutativity lemma}:
\begin{lemma} \label{lem:commutativity lemma}
    Let $C$ be a finite region and let $\caS, \caS' \in \Bd(\Sigma_C)$ be distinct boundary components. Then we have
    $$ \frt_{\caS}(a) \frt_{\caS'}(b) = \frt_{\caS'}(b) \frt_{\caS}(a), $$
    for all $a \in \Tube_{\caS}$ and all $b \in \Tube_{\caS'}$.
    In particular, all $B_f$ projectors commute with one another.
\end{lemma}

\begin{proof}
    Let $C^{\caS}$ and $C^{\caS'}$ be the collar regions of the boundary components $\caS$ and $\caS'$ and put $\widetilde C = C^{\caS} \cup C^{\caS'}$. Then $\caS$ and $\caS'$ are also boundary components of $\Sigma_{\widetilde C}$ with the same corresponding $\Tube$-representations $\frt_{\caS}$ and $\frt_{\caS'}$. We regard $\frt_{\caS}(a)$ and $\frt_{\caS'}(b)$ as operators on $\caH_{\widetilde C}$. On the subspace $H_{\widetilde C} \subset \caH_{\widetilde C}$, commutativity of  $\frt_{\caS}(a)$ and $\frt_{\caS'}(b)$ now follows 
    from Proposition \ref{prop:isomorphism of skein subspace and skein module} and the fact that $\Tube$-actions corresponding to different boundary components on the skein module $A(\Sigma_{\widetilde C})$ all commute with each other. On the orthogonal complement of $H_{\widetilde C}$ both $\frt_{\caS}(a) \frt_{\caS'}(b)$ and $\frt_{\caS'}(b) \frt_{\caS}(a)$ vanish because at least one string-net constraint on one of the collar regions $C^{\caS}$ or $C^{\caS'}$ is violated.

    The claim about the $B_f$'s now follows by recalling that $B_f = \frt_f(P^{\I})$ and considering, for any two faces $f_1$ and $f_2$, the region $\widetilde C=C^{f_1}\cup C^{f_2}$.
\end{proof}

\subsection{Boundary conditions and density matrices supported on skein subspaces} \label{subsec:boundary conditions}

Let $H_C$ be the skein subspace associated to a finite region $C$. A \emph{boundary condition} for $H_C$ is an assignment $\underline p = \{ p_{\caS} \}_{\caS \in \Bd(\sigma_C)}$ of a projector $p_{\caS} \in \Tube_{\caS}$ to each connected boundary component $\caS \in \Bd
(\Sigma_C)$. We define
\begin{equation}\label{eq:def H_C(p)}
    H_C( \underline p ) := \left( \prod_{\caS \in \Bd(\Sigma_C)}^m \frt_{\caS}(p_{\caS})  \right) \, H_C.
\end{equation}
Given an enumeration $\{ \caS_{\kappa} \}_{\kappa = 1}^m$ of $\Bd(\Sigma_C)$ we will also write this as
$$ H_C(p_{\caS_1}, \cdots, p_{\caS_m}) = H_C( \underline p ). $$
We further define the following collections of density matrices:
\begin{align*}
    \caD_C :=& \, \{ \text{density matrices } \rho \in \caB(\caH_C) \, \text{supported on } H_C \}, \\
    \caD_C( \underline p ) :=& \, \{ \text{density matrices } \rho \in \caB(\caH_C) \, \text{supported on } H_C( \underline p ) \}, \\
    =& \, \{ \rho \in \caD_C \, : \,   \Tr \lbrace \rho \,  \frt_{\caS}(p_{\caS}) \rbrace = 1 \, \text{for all } \, \caS \in \Bd(\Sigma_C) \}.
\end{align*}
Given an enumeration $\{ \caS_{\kappa} \}_{\kappa = 1}^m$ of $\Bd(\Sigma_C)$ we will also write this as
$$ \caD_C(p_{\caS_1}, \cdots, p_{\caS_m}) = \caD_C( \underline p ). $$

Restrictions of density matrices in $\caD_C$ to smaller regions yield density matrices that satisfy so-called \emph{maximally mixed boundary conditions} (\cf Lemma \ref{lem:restriction yields maximally mixed boundary conditions}), which will play an important role in our analysis.
\begin{definition} \label{def:maximally mixed boundary conditions}
    A density matrix $\rho \in \caD_C$ is said to satisfy maximally mixed boundary conditions $\star$ at boundary component $\caS \in \Bd(\Sigma_C)$ if for any $a = [y] \in \Tube_{\caS}$ we have
    \begin{equation} \label{eq:maximally mixed boundary conditions}\tag{$\star$}
        d_{y_{\caS}} \, \frt_{\caS}(a) \rho = d_{y_{\hat \caS}} \, \rho \frt_{\caS}(a).
    \end{equation} 
    (Recall that $y_{\caS}$ is the boundary condition induced by $y$ on the top of the cylinder $\caS \times I$, and $y_{\hat \caS}$ is the boundary condition induced at the bottom.)
    If in addition 
    \begin{equation} \label{eq:maximally mixed boundary conditions of type X} \tag{$\star^X$}
        \Tr \lbrace \rho \frt_{\caS}(P^X) \rbrace = 1,
    \end{equation}
    then we say $\rho$ satisfies the maximally mixed boundary condition $\star^X$ of type $X$ at boundary component $\caS$.
\end{definition}

We allow a boundary condition $\underline p = \{ p_{\caS} \}_{\caS}$ to have components $p_{\caS} = \star$ or $p_{\caS} = \star^X$ for some $X \in \Irr \ZC$, and define $\caD_C( \underline p )$ to consist of those density matrices $\rho \in \caD_C$ such that for each $\caS \in \Bd(\Sigma_C)$ we have $\Tr \lbrace \rho p_{\caS} \rbrace = 1$ in case $p_{\caS}$ is a projector, and such that $\rho$ satisfies maximally mixed boundary conditions at $\caS$ in case $p_{\caS} = \star$, or the maximally mixed boundary condition of type $X$ at $\caS$ if $p_{\caS} = \star^{X}$.

\subsection{Characterization of skein subspaces} \label{subsec:characterization of skein subspaces}

Let $C$ be a finite connected region. Then there is $m \in \N_0$ so that $\Sigma_C$ is homeomorphic to the sphere with $m+1$ holes cut out. An anchor $\anchor$ for $\Sigma_C$ is said to be an anchor for $C$ if the induced enumeration $\{ \caS_{\kappa}^{\anchorino}\}_{\kappa = 0}^m$ is such that $\caS^{\anchorino}_0$ is the outer boundary of $\sigma_C$, the outer boundaries of the surfaces $\Sigma_C$ play a distinguished role and we find it convenient to replace $\caC_{\anchorino}(\Sigma_C)$ by
\begin{align} \label{eq:great interface space}
    \begin{split}
    \caC^*_{\anchorino}(\Sigma_C) :=  \bigoplus_{X_0, \cdots, X_m \in \Irr ( \ZC )}  \ZC( &X_0^* \rightarrow X_1 \otimes \cdots \otimes X_m) \\ & \otimes \caC( X_0 \rightarrow \chi^{ \otimes \caS^{\anchorino}_0} ) \otimes \cdots \otimes \caC(X_m \rightarrow \chi^{\otimes\caS^{\anchorino}_m}),
    \end{split}
\end{align}
which is naturally isomorphic to $\caC_{\anchorino}(\Sigma_C)$ by composing elements of the factor $\ZC( X_0^* \rightarrow X_1 \otimes \cdots \otimes X_m)$ with a duality morphism.
Composing with the isomorphism of \Cref{prop:characterisation of skein modules}, this defines a map $\Phi_{\Sigma_C}^{\anchorino,*} : \caC^*_{\anchorino}(\Sigma_C) \rightarrow A(\Sigma_C)$, with
\begin{equation} \label{eq:modified Phi}
    \Phi_{\Sigma_C}^{\anchorino, *} \big( \al \otimes w_0 \otimes \cdots \otimes w_m \big) := 
    \Phi_{\Sigma_C}^{\anchorino} \big( (( \id_{X_0} \otimes \alpha ) \circ \coev_{X_0}) \, \otimes w_0 \otimes \cdots \otimes w_m   \big)
\end{equation}
for any $\al \in \ZC( X_0^* \rightarrow X_1 \otimes \cdots \otimes X_m )$, and any boundary conditions $w_{\kappa} \in \caC(X_{\kappa} \rightarrow \chi^{\otimes\caS^{\anchorino}_{\kappa}})$ for $\kappa = 0, \cdots, m$.
Combined with the isomorphism of \Cref{prop:isomorphism of skein subspace and skein module}, $H_C \simeq A(\Sigma_C)$, this gives:
\begin{proposition} \label{prop:the great interface}
    Let $C$ be a finite connected region and let $\anchor$ be an anchor for $C$. Then the map
    \begin{equation} \label{eq:great interface isomorphism}
        \Psi_C^{\anchorino} := \sigma_C^{-1} \circ \Phi_C^{\anchorino, *} : \caC^*_{\anchorino}(\Sigma_C) \rightarrow H_C
    \end{equation}
    is a unitary isomorphism of  $A(\partial \Sigma_C)$-modules.
\end{proposition}

\begin{lemma} \label{lem:maximally mixed boundary conditions decouple from the bulk}
    Let $\caS$ be an extended circle and let $\frt : \Tube_{\caS} \to B(K)$ be an irreducible representation on a Hilbert space $K$.

    Let $\rho_{K}$ be the density matrix on $K$ proportional to $ \sum_p \, d_p \frt(p)$,
    where the sum ranges over minimal projections $p = E_{\underline a, i}^X$ of $\Tube_{\caS}$ according to Eq. \eqref{eq:Tube_n matrix unit properties}, and $d_p := d_{\underline a}$.

    Let $L$ be another finite dimensional Hilbert space and let $\rho$ be a density matrix on $L \otimes K$ which satisfies maximally mixed boundary conditions $\eqref{eq:maximally mixed boundary conditions}$ w.r.t. the representation $a \mapsto \I \otimes \frt(a)$. Then $\rho = \rho_L \otimes \rho_K$ for some density matrix $\rho_L$ on $L$.
\end{lemma}

\begin{proof}
    Since $K$ is irreducible, the minimal projections $p = E_{\underline a, i}^X$ are represented with rank at most one. As $\rho$ satisfies the condition \eqref{eq:maximally mixed boundary conditions}, $\rho$ commutes with $\I \otimes \frt(p)$, so 
    $$
        \rho = \sum_p \, A_{p} \otimes \frt(p),
    $$
    for some operators $A_p$ acting on $L$.

    Assume $\frt(p), \frt(q) \neq 0$. Since the representation is irreducible, there is a non-zero $a \in \Tube_{\caS}$ with $a = q a p$. Applying \eqref{eq:maximally mixed boundary conditions} then yields
    $$ d_{q} \, A_{p} \otimes \frt(a) = d_{q} \, (\I \otimes \frt(a)) \rho = d_{p} \rho (\I \otimes \frt(a)) = d_{p} \, A_{q} \otimes \frt(a), $$
    hence $d_{q} A_{p} = d_{p} A_{q}$ for all such boundary conditions $p, q$. Since quantum dimensions are non-zero, it follows that there is $A \in \End(L)$ with $A_{p} = d_{p} A$ for all $p$ for which $\frt(p) \neq 0$. We conclude that
    $$ \rho = \sum_p \, d_p A \otimes \frt(p) = \rho_L \otimes \rho_K, $$
    with $\rho_L = \Tr \big(\sum_p d_p \frt(p) \big) A$.
\end{proof}

We can now characterise the spaces $\caD_C( \underline p )$ as follows:

\begin{lemma} \label{lem:characterization of some caDs}
	Let $C$ be a finite region and let $\underline p = \{ p_{\caS} \}_{\caS \in \Bd(\Sigma_C)}$ with each $p_{\caS}$ either a minimal projector of type $X_{\caS} \in \Irr \ZC$, or a maximally mixed boundary condition $\star^{X_{\caS}}$ of type $X_{\caS}$. Let $\anchor$ be an anchor for $C$. Then $\caD_C( \underline p )$ is isomorphic to the space of density matrices on $\ZC\Big( X^*_{\caS^{\anchorino}_0} \rightarrow \bigotimes_{\kappa = 1}^{m} \, X_{\caS^{\anchorino}_{\kappa}} \Big)$.
\end{lemma}

\begin{proof}
    If each $p_{\caS}$ is a minimal projector of type $X_{\caS}$, then Proposition \ref{prop:the great interface} yields an isomorphism, 
    $$ H_C(\underline p) \stackrel{\Psi_C^{\anchorino}}{\simeq} \ZC\Big( X_{\caS^{\anchorino}_0} \rightarrow \bigotimes_{\kappa = 1}^{m} \, X_{\caS^{\anchorino}_{\kappa}} \Big),$$
    and the claim follows.

	Let us now consider the case with one maximally mixed boundary condition $p_{\caS_*} = \star^{X_{\caS_*}}$ and all other $p_{\caS}$ given by minimal projections of types $X_{\caS}$. 
    Define $\underline{\tilde p}$ from $p$ by substituting $\star^{X_{\caS_\star}}$ with $P^X_{\caS_\star}$.
    Proposition \ref{prop:the great interface} yields an isomorphism of $\Tube$-modules
	\begin{align*} 
	   H_C(\underline{\tilde p}) \stackrel{\Psi_C^{\anchorino}}{\simeq} \ZC&\Big( X^*_{\caS^{\anchorino}_0} \rightarrow \bigotimes_{\kappa = 1}^{m} \, X_{\caS^{\anchorino}_{\kappa}} \Big) \otimes \caC( X_{\caS_*} \rightarrow \chi^{\otimes\caS_*} ).
	\end{align*}
    Density matrices $\rho \in \caD_C(\underline{\tilde p})$ are supported on $H_C(\underline{\tilde p})$ and therefore in one-to-one correspondence with density matrices $\rho' = \big( \Psi_C^{\anchorino}|_{H_C(\underline{p})} \big)^{-1} \,\rho \, \Psi_C^{\anchorino}|_{H_C(\underline{p})}$ acting on $\ZC\big( X^*_{\caS^{\anchorino}_0} \rightarrow \bigotimes_{\kappa = 1}^{m} \, X_{\caS^{\anchorino}_{\kappa}} \big) \otimes \caC( X_{\caS_*} \rightarrow \chi^{\otimes\caS_*} )$.
    
    The maximally mixed boundary condition on $\rho$ is transferred by this isomorphism to the maximally mixed boundary condition on $\rho'$. Since the $\Tube_{\caS_*}$ representation on $\caC(X_{\caS_*} \to \chi^{\otimes \caS_*})$ is irreducible, it follows from Lemma \ref{lem:maximally mixed boundary conditions decouple from the bulk} that $\rho' = \sigma \otimes \rho_{\caS_*}$ where $\rho_{\caS_*}$ is the maximally mixed density matrix on the irreducible representation $\caC(X_{\caS_*} \to \chi^{\otimes \caS_*})$ introduced in the statement of Lemma \ref{lem:maximally mixed boundary conditions decouple from the bulk}, and $\sigma$ is a density matrix on  $\ZC \big( X^*_{\caS^{\anchorino}_0} \rightarrow \bigotimes_{\kappa = 1}^{m} \, X_{\caS^{\anchorino}_{\kappa}} \big)$. Conversely, any density matrix of this form satisfies \eqref{eq:maximally mixed boundary conditions} and is therefore mapped by $\Psi_C^{\anchorino}$ to an element of $\caD_C(\underline p)$. This shows that the space of density matrices $\caD_C(\underline{p})$ is isomorphic to the space of density matrices on $\ZC( X^*_{\caS^{\anchorino}_0} \rightarrow \bigotimes_{\kappa = 1}^{m} \, X_{\caS^{\anchorino}_{\kappa}} )$, as required.
    
	The case with multiple maximally mixed boundary conditions is obtained similarly.
\end{proof}

\begin{remark} \label{rem:particles and fusion in entanglement bootstrap program}
    The spaces $\caD_C(\underline p)$, with each $p_{\caS} = \star$ given by the maximally mixed boundary condition, are precisely the information convex sets of the entanglement bootstrap program \cite{shi2020fusion}. Indeed, the information convex sets of \cite{shi2020fusion} consist of density matrices supported on some region $C$, obtained by tracing out the boundary degrees of freedom of density matrices in $\caD_{\widetilde C}$, where $\widetilde C$ is a slight `fattening' of $C$. It follows from Lemma \ref{lem:restriction yields maximally mixed boundary conditions} below that such density matrices satisfy maximally mixed boundary conditions on all boundary components.
    
    By applying Lemma \ref{lem:characterization of some caDs} to annuli and twice punctured disks, we therefore obtain the simple anyon types and the fusion rules of the Levin-Wen model in the framework of the entanglement bootstrap program.
\end{remark}

\subsection{Restrictions yield maximally mixed boundary conditions}

Maximally mixed boundary conditions arise when density matrices supported on skein subspaces of some region are restricted to certain smaller regions \cite{levin2006detectingTO}.
\begin{lemma} \label{lem:restriction yields maximally mixed boundary conditions}
    Let $C' \subset C$ be finite regions such that $\Sigma_C \setminus \Sigma_{C'}$ is an annulus. Then there are unique boundary components $\caS \in \Bd(\Sigma_C)\setminus \Bd(\Sigma_{C'})$, and $\caS' \in \Bd(\Sigma_{C'})\setminus \Bd(\Sigma_{C})$.
    Suppose $\rho \in \caD_C$ satisfies boundary conditions of type $X \in \Irr \ZC$ at boundary component $\caS$. Let $\rho' = \Tr_{C_0 \setminus C'_0} \lbrace \rho \rbrace$ be the restriction of $\rho$ to the region $C'$. Then $\rho' \in \caD_{C'}$ satisfies maximally mixed boundary conditions of type $X$ at boundary component $\caS'$.
\end{lemma}

\begin{proof}
    By hypothesis $\Sigma_C \setminus \Sigma_{C'}$ has two connected boundary components $\caS$ and $\hat \caS'$ corresponding to those of $\Sigma_C$ and $\Sigma_{C'}$ as stated. 
    Let $C''$ be the  region consisting of all faces, edges, and vertices contained in the annulus $\Sigma_C \setminus \Sigma_{C'}$.
    
    Consider a product state $\phi = \phi_{x'} \otimes \phi_{x''} \in H_{C'_1} \otimes H_{C''_1}$ corresponding to string diagrams $x'$ on $\Sigma_{C'_1}$ and $ x''$ on $\Sigma_{C''_1}$, and $a\in \Tube_{\caS'}$ corresponding to a string diagram $y$ on $\caS' \times I$. In the case of an inner boundary, we may illustrate the situation as follows (recall that we use Convention \ref{conv:graphical representation} to interpret string diagrams as vectors in $H_{C_1}$):
    \begin{equation*}
	   \phi := \adjincludegraphics[valign=c, width = 4.0cm]{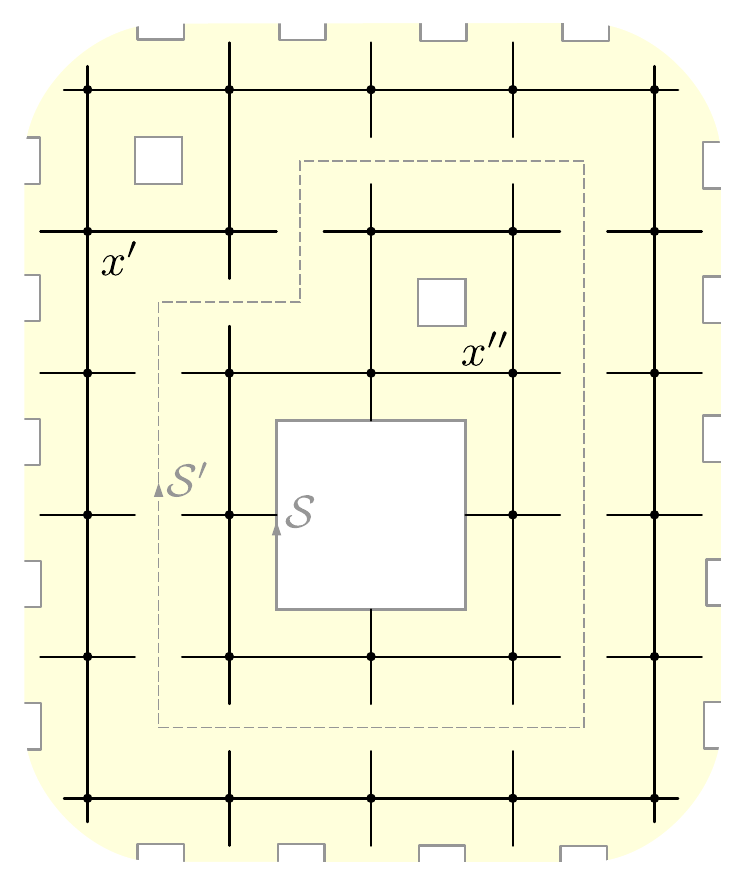}, \quad a :=  \adjincludegraphics[valign=c, width = 2.5cm]{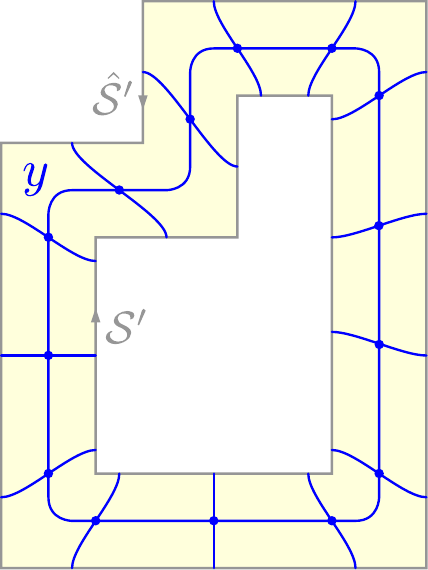}.
    \end{equation*}
    Recalling (Eq. \eqref{eq:graphical frt action}), we find 
    \begin{align*}
        B_C \frt_{\caS'}(a) \phi &= \left( \frac{d_{y_{\caS'}}}{d_{x'_{\caS'}}} \right)^{1/4} \times \, \adjincludegraphics[valign=c, width = 4.0cm]{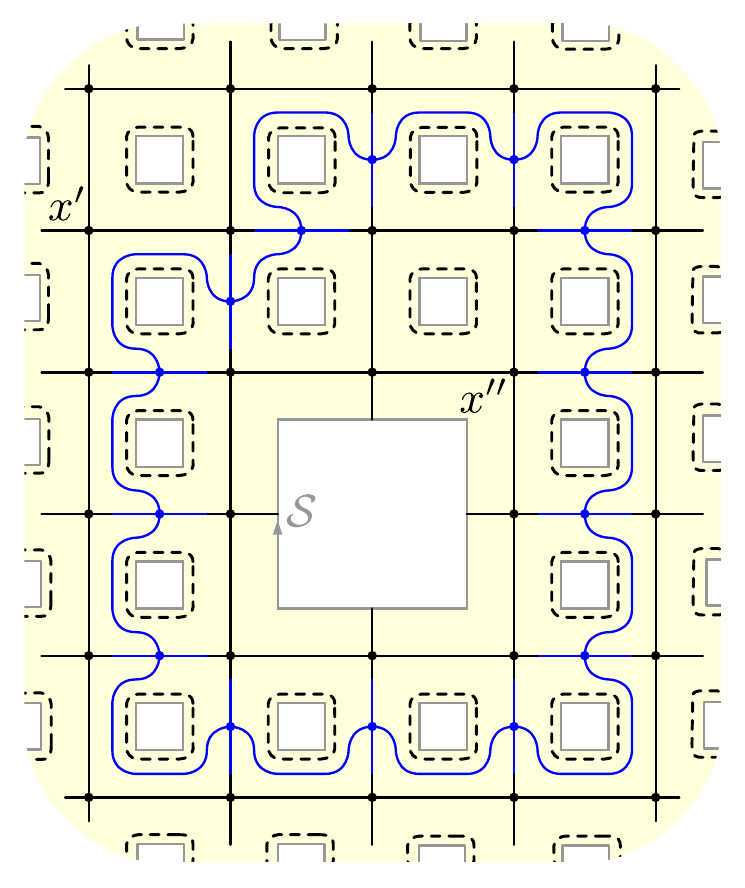}, \\
        B_C \frt^{\op}_{\hat \caS'}(a) \phi &= \left( \frac{d_{y_{\hat \caS'}}}{d_{x''_{\caS'}}} \right)^{1/4} \times \, \adjincludegraphics[valign=c, width = 4.0cm]{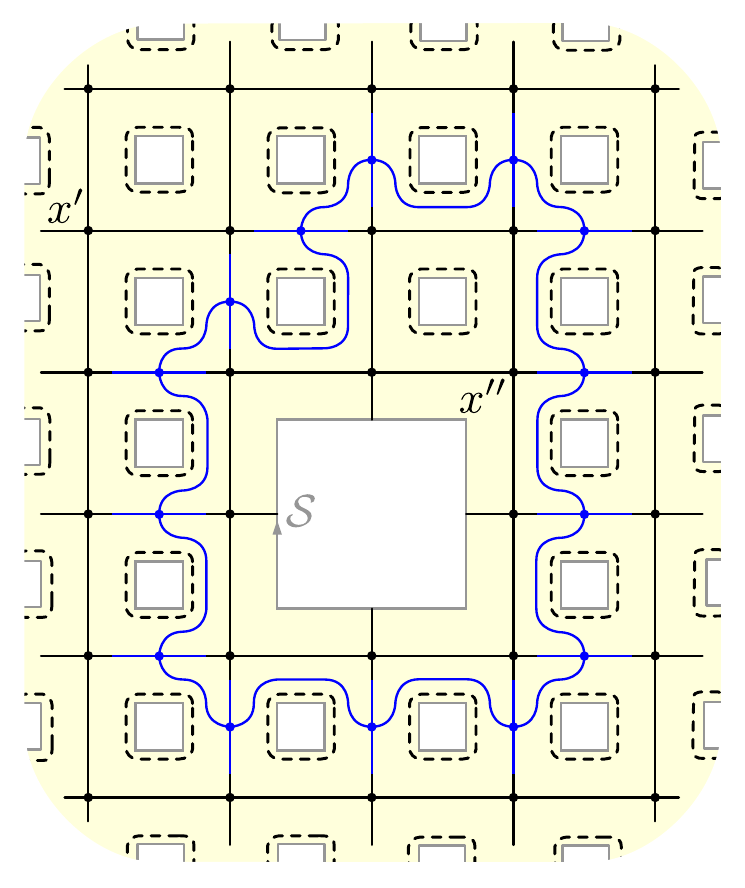}.
    \end{align*}
    Note that both these vectors are zero if any of the boundary conditions of the string diagrams $x', x''$ and $y$ do not match where they are glued together. 
    By using the cloaking property of the dotted loops within each face that sits on $\caS'$,     we find that
    \begin{equation*}
        d_{y_{\hat \caS'}}^{1/2} \, B_C \frt_{\caS'}(a) \phi = d_{y_{\caS'}}^{1/2} \, B_C {\frt}^{\op}_{\hat \caS'}(a) \phi.
    \end{equation*}
    By linearity, this equality extends to any $\phi \in H_{C'_1} \otimes H_{C''_1}$. In fact, this equality holds for any $\phi \in \caH_C$ because if any string-net constraint in $C'$ or in $C''$ is violated, then $\phi$ is annihilated either by a $\Tube$-action or by $B_C$. We conclude that
    \begin{equation}\label{eq:Tube agrees with op on B_C}
        d_{y_{\hat \caS'}}^{1/2} \, B_C \frt_{\caS'}(a) = d_{y_{\caS'}}^{1/2} \, B_C \frt^{\op}_{\hat \caS'}(a).
    \end{equation} 
    Repeating the argument with $a^*$ and taking adjoints, we find that
    $  d_{y_{\caS'}}^{1/2} \, \frt_{\caS'}(a) B_C =  d_{y_{\hat \caS'}}^{1/2} \, \frt^{\op}_{\hat \caS'}(a) B_C.$
    Using these relations, the fact that $\rho = \rho B_C = B_C \rho$, and the fact that $\frt_{\caS'}(a)$ is supported on $C'$ while $\frt^{\op}_{\hat \caS'}(a)$ is supported on $C''$, we find that
    \begin{align*}
        d_{y_{\hat \caS'}} \, \rho' \, \frt_{\caS'}(a) &= d_{\hat y_{\caS'}} \, \Tr_{\caH_{C''}} \lbrace  \rho \frt_{\caS'}(a)  \rbrace = d_{\hat y_{\caS'}} \, \Tr_{\caH_{C''}} \lbrace \rho B_C \frt_{\caS'}(a) \rbrace \\
        &= (d_{y_{\caS'}} d_{y_{\hat \caS'}})^{1/2} \, \Tr_{\caH_{C''}} \lbrace \rho B_C \frt^{\op}_{\hat \caS'}(a) \rbrace \\
        &= (d_{y_{\caS'}} d_{y_{\hat \caS'}})^{1/2} \, \Tr_{\caH_{C''}} \lbrace \frt^{\op}_{\hat \caS'}(a) B_C \rho \rbrace \\
        &= d_{y_{\caS'}} \, \rho' \, \Tr_{\caH_{C''}} \lbrace \frt_{\caS'}(a) \rho \rbrace = d_{y_{\caS'}} \, \frt_{\caS'}(a) \rho'.
    \end{align*}
    Since the string diagram $y$ defining  $a \in \Tube_{\caS'}$ was arbitrary, we conclude that $\rho'$ satisfies maximally mixed boundary conditions \eqref{eq:maximally mixed boundary conditions} at $\caS'$.

    Let us finally show that $\rho'$ satisfies boundary conditions of type $X$ at the boundary component $\caS'$. By assumption, we have $\Tr_{\caH_C} \lbrace \rho \frt_{\caS}(P^X) \rbrace = 1$. Writing $\rho'' = \Tr_{\caH_{C'}} \lbrace \rho \rbrace$ for the restriction of $\rho$ to the annular region $C''$, we immediately obtain that $\Tr_{\caH_{C''}} \lbrace \rho'' \, \frt_{\caS}(P^X) \rbrace = 1$. It follows that $\rho''$ belongs to $\caD_{C''}( P^X, \id)$, where the first slot corresponds to the boundary component $\caS$ and the second slot corresponds to the boundary component $\hat \caS'$.
    By Proposition \ref{prop:the great interface}, in fact $\rho'' \in \caD_{C''}(P^X, P^{\bar X})$, meaning that $\Tr_{\caH_{C''}} \lbrace \rho'' \, \frt_{\hat \caS'}(P^{\bar X}) \rbrace = 1$. 
    Expressing $P^X$ as a sum of minimal projections as in Eq. \eqref{eq:def of P^X}, we obtain $B_C \frt_{\caS'}(P^X) =  B_C \frt^{\op}_{\hat \caS'}(P^X)$ from Eq. \eqref{eq:Tube agrees with op on B_C}. Using also $(P^X)^{\op} = P^{\bar X}$, we get
        \begin{align*}
        \Tr_{\caH_{C'}} \lbrace \rho' \, \frt_{\caS'}(P^{X}) \rbrace &=  
        \Tr_{\caH_C} \lbrace \rho \, B_C \, \frt_{\caS'}(P^{X}) \rbrace 
        = \Tr_{\caH_C} \lbrace \rho \, B_C \, \frt_{\hat \caS'}^{\op}(P^{X}) \rbrace \\
        &= \Tr_{\caH_C} \lbrace \rho \, B_C \, \frt_{\hat \caS'}(P^{\bar X}) \rbrace 
        = \Tr_{\caH_{C''}} \lbrace \rho'' \, \frt_{\hat \caS'}(P^{\bar X}) \rbrace = 1
    \end{align*}
    as desired.
\end{proof}

\subsection{Gluing along intervals} \label{ssubsec:gluing skein subspaces}

A dual path $I = \{f_i\}_{i = 1}^l$ is a sequence of faces such that $f_i$ neighbours $f_{i+1}$ for all $i = 1, \cdots, l-1$. This dual path is self-avoiding if all the $f_i$ are distinct. We write $\partial_{\ii} I = f_1$ and $\partial_{\f} I = f_l$. We also write $I_{\inn} = \{ f_i \}_{i = 2}^{l-1}$ for the dual path obtained from $I$ by removing the first and last face. ($I_{\inn}$ can be empty). We write $\caE_I \subset \caE$ for the set of edges which lie between successive faces of $I$. 

Let $D$ be a finite connected region. We say $D$ can be cut along a self-avoiding dual path $I$ if all faces of $I_{\inn}$ and all edges in $\caE_I$ belong to $D$, and the initial and final faces $\partial_{\ii} I$ and $\partial_\f I$ of $I$ are distinct and sit on the outer boundary of $D$. Write $D \setminus I$ for the region obtained from $D$ by removing the faces in $I_{\inn}$ as well as all the edges between successive faces of $I$. Then $D \setminus I$ is the union of two uniquely determined disjoint finite connected regions $C$ and $C'$.  Moreover, each inner boundary component of $\Sigma_D$ is either an inner boundary component of $C$ or an inner boundary component of $C'$. We write $D = C \sqcup_I C'$ and say $D$ is obtained by gluing $C$ and $C'$ along $I$.

The subregions $C$ and $C'$ have associated surfaces $\Sigma_C$ and $\Sigma_{C'}$ whose outer boundary components both contain a boundary interval (see Section \ref{subsec:gluing}) $\caI$ and $\caI'$ respectively, determined by $I$ as follows. Let $m_{C, I}$ be the set of marked boundary points of $\Sigma_C$ which lie on one of the edges of $\caE_I$. Define $m_{C', I}$ similarly. Then we may define $\caI$ as the smallest closed subinterval of $\partial \Sigma_C$ that contains $B_{1/3}(m)$ for each $m \in m_{C, I}$, where $B_r(x) \subset \R^2$ is the closed ball or radius $r$ centred on $x \in \R^2$. The boundary interval $\caI'$ is defined similarly. Moreover, $\caI$ can be glued to $\caI'$ to obtain a new surface homeomorphic to $\Sigma_D$ as extended surfaces. We will identify the glued surface with $\Sigma_D$. Denote by $\gl : A(\Sigma_C \sqcup \Sigma_{C'}) \simeq A(\Sigma_C) \otimes A(\Sigma_{C'}) \rightarrow A(\Sigma_D)$ the associated gluing map as described in Section \ref{subsec:gluing}:
\begin{equation*}
    \adjincludegraphics[valign=c, height = 2.0cm]{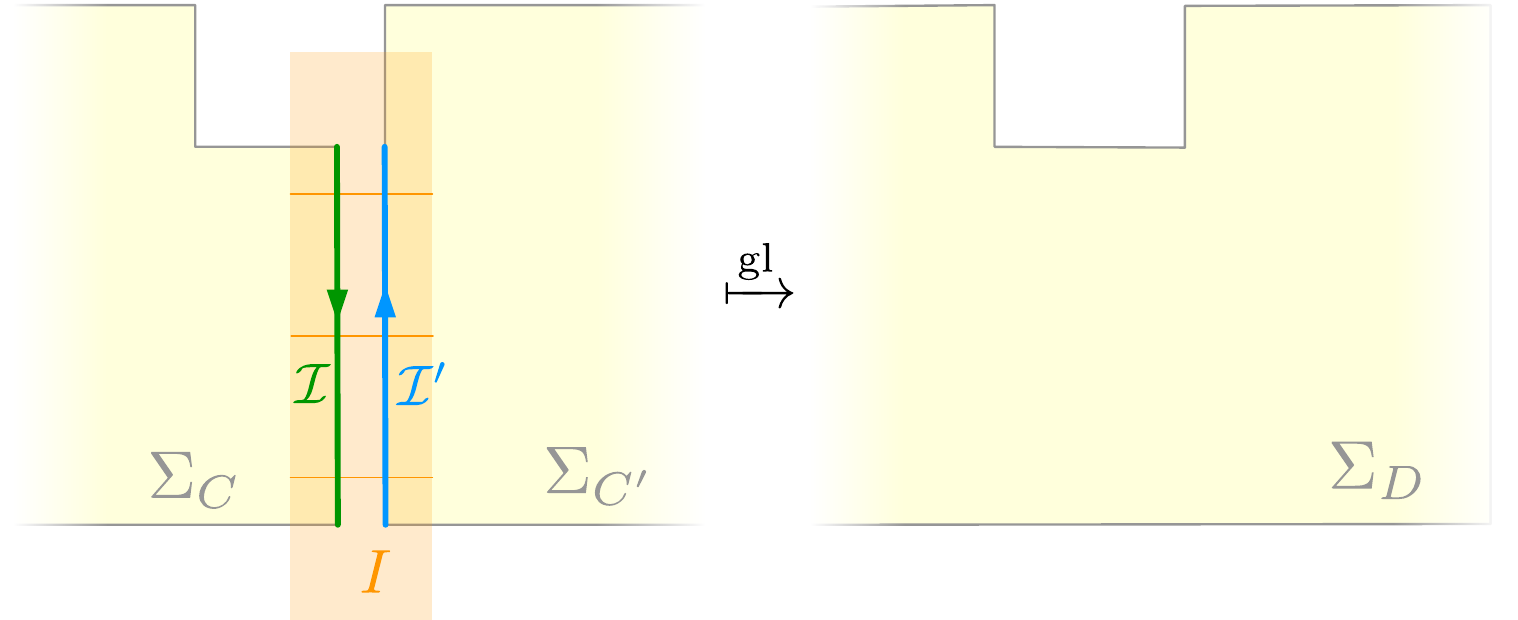} \quad.
\end{equation*}

Let $B_I = \prod_{f \in I_{\inn}} B_f$.
\begin{lemma} \label{lem:gluing skein subspaces}
    Let $D = C \sqcup_I C'$ be as above. For any $\phi \in H_{C_1}^{b}$ and $\phi' \in H_{C'_1}^{b'}$ we have
    \begin{equation}\label{eq:gluing skein modules}
        \big( \gl \circ (\sigma_C \otimes \sigma_{C'}) \big)( \phi \otimes \phi' ) = \caD^{\abs{I_{\inn}}} \, d_{b_{\caI}}^{-1/2} \, \sigma_D\big( B_I \, \phi \otimes \phi' \big).
    \end{equation}
\end{lemma}

\begin{proof}
    By construction of the gluing map $\gl$ and \Cref{lem:sigma_C absorbs B_C} we have 
    $$\gl \circ (\pi_C \otimes \pi_{C'}) = \pi_D = \pi_D \circ B_D.$$ 
    To finish the proof it therefore suffices to verify that normalization factors relating the $\sigma$'s to the $\pi$'s depending on boundary conditions agree.
    For $\phi \in H^b_{C_1}$ and $\phi' \in H^{b'}_{C'_1}$ we find
    $$ \sigma_C( \phi ) = \caD^{-\abs{F_C}} \, d_b^{1/4} \, \pi_C(\phi), \quad  \sigma_{C'}( \phi' ) = \caD^{-\abs{F_{C'}}} \, d_{b'}^{1/4} \, \pi_{C'}(\phi),$$
    and
    $$ (\sigma_D \circ B_D)( \phi \otimes \phi' ) = \delta_{b_{\caI}, b'_{\caI'}} \, \caD^{-\abs{F_D}} \, \left(\frac{d_{b} d_{b'}}{ d_{b_{\caI}}^2 }\right)^{1/4} \, (\pi_D \circ B_D)( \phi \otimes \phi' ). $$
    The claim follows by noting that $\abs{F_D} - \abs{F_C} - \abs{F_{C'}} = \abs{I_{\inn}}$ and that $\gl$ also imposes compatible boundary conditions $\delta_{b_{\caI}, b'_{\caI'}}$.
\end{proof}

\section{Anyon States} \label{sec:anyon states}

\subsection{Unique anyon states from local constraints}

Recall that to any region $C$ we assign the quasi-local algebra $\caA_C = \caA_{C_0}$ of observables of the Levin-Wen spin system which are supported on the set of vertices $C_0$ of $C$, as described in Section \ref{subsec:local degrees of freedom}.

Let $C$ be an infinite region such that $\Sigma_C$ is homeomorphic to $\R^2$ with an open disk removed. The associated surface $\Sigma_C$ has a single connected boundary component $\caS$. We denote the corresponding $\Tube_{\caS}$ action by $\frt = \frt_{\caS}$ for the remainder of this section. By definition, operators $\frt(a)$ for $a \in \Tube_{\caS}$ act on the degrees of freedom at the vertices of the collar region $C_{\caS} \subset C$ (\cf Section \ref{subsec:Tube actions on collars}), so these operators belong to $\caA_C$.

Note that Lemma \ref{lem:commutativity lemma} implies that $\frt(a)$ commutes with $B_f$ for all $a \in \Tube_{\caS}$ and all $f \in F_C$.

\begin{definition} \label{def:state spaces}
    Let $\caS_C$ be the set of states $\omega$ on $\caA_C$ for which $\omega(B_f) = 1$ for all $f \in F_C$.

    For any projector $p \in \Tube_{\caS}$ we let $\caS_C^p$ be the set of states $\omega$ on $\caA_C$ such that $\omega \in \caS_C$ and $\omega(\frt(p)) = 1$.

    We further define $\caS_C^{\star}$ to be the set of states $\omega \in \caS_C$ for which $\omega = \omega \circ \Ad[ \frt(u) ]$ for all unitaries $u \in \Tube_{\caS}$.
    
    For any $X \in \Irr \ZC$ we let $\caS_C^{\star^X} =\caS_C^{\star}\cap \caS^{P^X}_C$. We say that $\omega\in \caS_C^{\star^X}$ satisfies maximally mixed boundary conditions of type $X$.
\end{definition}

\begin{proposition} \label{prop:unique anyon states}
    If $p \in \Tube_\caS$ is a minimal projector, then $\caS_C^p$ consists of a single pure state. Similarly, the state space $\caS_C^{\star^X}$ consists of a single state (which however need not be pure).
\end{proposition}

\begin{proof}
    We first prove the claim about $\caS_C^p$. For any $R > 0$, denote by $C(R)$ the subregion of $C$ consisting of all vertices, edges and faces of $C$ that are contained in $B_R = \{ v \in \R^2 \,: \, \norm{v}_{\infty} \leq R \}$.

    There is an $R_0$ such that for all $R \geq R_0$ large enough, the associated surface $\Sigma_{C(R)}$ is homeomorphic to an annulus. The inner boundary component of $\Sigma_{C(R)}$ is the same as the unique boundary component of $\Sigma_C$, and so comes with the same $\Tube_{\caS}$ representation $\frt$. Let $\caS_R$ be the outer boundary component of $\Sigma_{C(R)}$ and denote by $\frt_R = \frt_{\caS_R}$ the $\Tube_{\caS_R}$ action associated to that boundary component.

    Since $p$ is minimal, it has definite type $X$ for some $X\in \Irr \ZC$.
    By \Cref{lem:characterization of some caDs} there is a unique density matrix in $\caD_{C(R)}(p, \star^{\bar X})$ defining a state $\psi_R$ on $\caA_{C(R)}$ for all $R\ge R_0$.
    By \Cref{lem:restriction yields maximally mixed boundary conditions}, $\psi_{R'}|_{C(R)}=\psi_R$ whenever $R'>R$, so it is clear that 
    $\psi(O) = \lim_{R \uparrow \infty} \psi_R(O)$ defines a state on $\caA_C^{\loc}$ which extends uniquely to a state $\psi : \caA_C \rightarrow \C$.
    By construction $\psi_R(\frt(p)) = 1$ for all $R>R_0$, and $\psi_R(B_f) = 1$ for $f \in F_C$ whenever $R$ is large enough, so $\psi \in \caS_C^p$.

    To see that $\psi$ is the only state in $\caS_C^p$ it suffices to note that the restriction of any state in $\caS_C^p$ to the region $C(R)$ for $R \geq R_0$ corresponds to the unique density matrix in $\caD_{C(R)}(p, \star^X)$.

    To see that $\psi$ is pure, suppose $\psi = \lambda \psi_1 + (1- \lambda) \psi_2$ for some $\lambda \in (0, 1)$. For any projector $q$, if $\psi(q) = 1$ then $\lambda \psi_1(q) + (1-\lambda) \psi_2(q) = 1$, which can only hold if $\psi_1(q) = \psi_2(q) = 1.$ Applying this to $q = \frt(p)$ and $q = B_f$ it follows that $\psi_1$ and $\psi_2$ both belong to $\caS_C^p$ as well. By the uniqueness we conclude that $\psi_1 = \psi_2 = \psi$, so $\psi$ is pure.

    The claim about $\caS_C^{\star^X}$ is shown in exactly the same way, using that the spaces $\caD_{C(R)}(\star^X, \star^{\bar X})$ all contain a unique density matrix by Lemma \ref{lem:characterization of some caDs}. (Purity of the unique state in $\caS_C^{\star^X}$ does not follow in this case because, unlike $\caS_C^p$, the space $\caS_C^{\star^X}$ is not defined by commuting projector constraints).
\end{proof}

\subsection{Pure anyons at punctures} \label{subsec:anyons on punctures}

Consider the infinite region $C^{(e)}$ obtained from $C^{\Z^2}$ by removing the edge $e$ and the two neighbouring faces. This region satisfies the  requirements of Proposition \ref{prop:unique anyon states}. In particular, the associated surface $\Sigma_{C^{(e)}}$ has a single connected boundary component $\caS_e$, referred to as the puncture at $e$, with an associated collar region $C^{\caS_e}$ and $\Tube_{\caS_e}$-representation denoted by $\frt_e = \frt_{\caS_e}$. It follows from Proposition \ref{prop:unique anyon states} that for any minimal projector $p \in \Tube_{\caS_e}$ we have a unique pure state $\omega_e^{p} \in \caS_{C^{(e)}}^p$. Note that $\caA_{C^{(e)}} = \caA$, so $\omega_e^p$ is a state on the full quasi-local algebra.

Let us finally argue how we get the existence, uniqueness, and purity of the frustration free ground state $\omega^{\I}$ (Proposition \ref{prop:unique ffgs}) as a corollary of Proposition \ref{prop:unique anyon states}. Let
\begin{equation} \label{eq:p^one defined}
    p^{\I} := v v^*, \quad \text{where} \quad v := \frac{1}{\caD} \sum_{a \in \Irr \caC} d_a^{1/2} \, \adjincludegraphics[valign=c, width = 2.0cm]{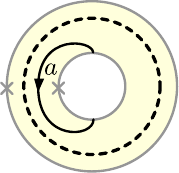} \in \Tube_{\caS_e},
\end{equation}
where we present $\caS_e \times I$ as an annulus whose outer boundary is identified with the bottom of $\caS_e \times I$.

\begin{lemma} \label{lem:p^one is minimal}
    The element $p^{\I} \in \Tube_{\caS_e}$ is a minimal projector of type $\I$.
\end{lemma}

\begin{proof}
    We compute
    \begin{equation}
        v^* v = \frac{1}{\caD^2} \sum_{a} \, d_a \, \adjincludegraphics[valign=c, width = 2.0cm]{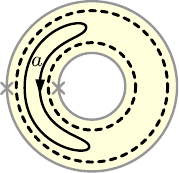} = \adjincludegraphics[valign=c, width = 2.0cm]{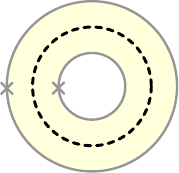} \,\, ,
    \end{equation}
    which is a minimal projector of type $\I$ by Proposition \ref{prop:matrix units for Tube_n}.  This implies that $v$ is a partial isometry, and that $p^{\I} = v v^*$ is unitarily equivalent to $v^* v$, and is therefore also a minimal projector of type $\I$.
\end{proof}

\begin{lemma} \label{lem:p^one enforces ground state constraints}
    Let $f_1$ and $f_2$ be the faces neighbouring a given edge $e$. Then $\frt_e(p^{\I}) = B_{f_1} B_{f_2}$.
\end{lemma}

\begin{proof}
    Write $V_e = \frt_e(v)$, then (using Convention \ref{conv:graphical representation})
    \begin{align*}
        \adjincludegraphics[valign=c, width = 3.0cm]{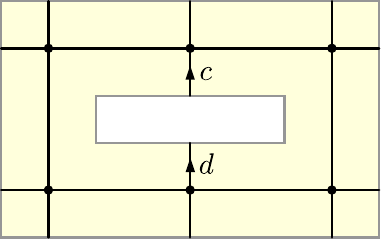} \xmapsto{V_e^*} \, \frac{\delta_{c, d}}{\caD} \,\,\, &\adjincludegraphics[valign=c, width = 3.0cm]{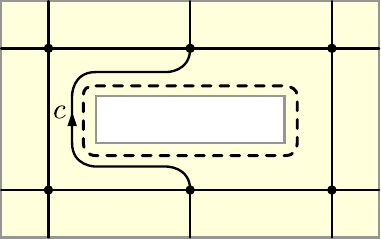} \, \xmapsto{V_e} \frac{\delta_{c, d}}{\caD^2} \sum_a d_a \,\,\, \adjincludegraphics[valign=c, width = 3.0cm]{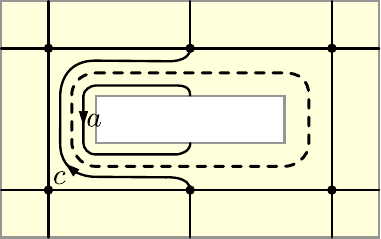} \\
        = \frac{\delta_{c, d}}{\caD^2} \sum_a d_a \,\,\,  &\adjincludegraphics[valign=c, width = 3.0cm]{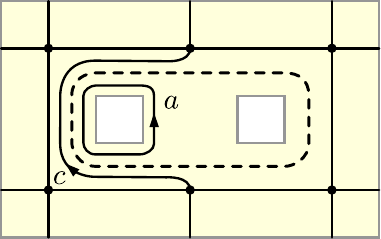} = \delta_{c, d} \,\,\, \adjincludegraphics[valign=c, width = 3.0cm]{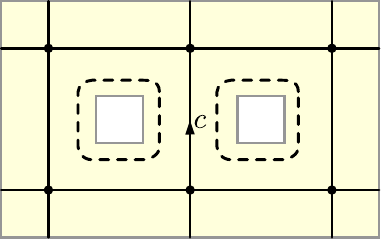}  \\ 
        = B_{f_1} B_{f_2} \,\,\, &\adjincludegraphics[valign=c, width = 3.0cm]{healing1.pdf},
    \end{align*}
    where in applying $V_e^*$ we noted that the factor $d_c^{1/2}$ coming from the definition of $v$ cancels the factor $d_c^{-1/2}$ due to the dependence of $\frt_e$ on boundary conditions (see Eq. \eqref{eq:graphical frt action}), and in applying $V_e$ we get one factor of $d_a^{1/2}$ from the definition of $v$ and other factor of $d_a^{1/2}$ from the boundary conditions. This shows that $\frt_e(p^{\I}) = V_e V_e^* = B_{f_1} B_{f_2}$ on the string-net subspace. If any of the string net constraints is violated, left and right hand sides of the claimed equality both evaluate to zero. This proves the Lemma.
\end{proof}

\begin{proofof}[\Cref{prop:unique ffgs}]
    It follows from Lemma \ref{lem:p^one enforces ground state constraints} that the space $\caS_{C^{(e)}}^{p^{\I}}$ consists of those states $\omega$ on $\caA$ for which
    $$ 1 = \omega(B_f) = \omega(B_{f_1} B_{f_2}) $$
    for all faces $f$ belonging to $C^{(e)}$, and where $f_1$ and $f_2$ are the two faces neighbouring the edge $e$. Using the Cauchy-Schwarz inequality we find
    $$ \abs{ \omega( B_{f_1} - B_{f_1} B_{f_2} ) }^2 = \abs{ \omega( B_{f_1} (\I - B_{f_1} B_{f_2}) ) }^2 \leq \omega( B_{f_1} ) \omega( \I - B_{f_1} B_{f_2} ) = 0, $$
    where we used that $B_{f_1}$ and $B_{f_2}$ are commuting projectors (Lemma \ref{lem:commutativity lemma}). It follows that $\omega(B_{f_1}) = \omega(B_{f_1} B_{f_2}) = 1$ and therefore also  $\omega(B_{f_2}) = 1$. We conclude that $\caS_{C^{(e)}}^{p^{\I}}$ consists precisely of the frustration free ground states of the Levin-Wen Hamiltonian. By Proposition \ref{prop:unique anyon states} and Lemma \ref{lem:p^one is minimal} we find that $\caS_{C^{(e)}}^{p^{\I}}$ contains a single pure state $\omega_e^{p^{\I}}$, which is therefore the unique frustration free ground state of the Levin-Wen Hamiltonian.
\end{proofof}

\subsection{Restrictions of infinite volume states}
Restrictions lead to maximally mixed boundary conditions also for infinite volume states. 
We formalise this in the present setting. 

\begin{lemma} \label{lem:restrictions of infinite volume states}
    Let $C$ be an infinite region so that $\Sigma_C$ is homeomorphic to $\R^2$ with an open disk removed. Let $C' \subset C$ be such that $\Sigma_C \setminus \Sigma_{C'}$ is homeomorphic to an annulus.
    If $p \in \Tube_{\partial C}$ is a projection of type $X \in \Irr \ZC$ and
    $\omega \in \caS_{C}^{p}$, then $\omega|_{\caA_{C'}} \in \caS_{C'}^{\star^X}$.  
\end{lemma}

\begin{proof}
    It suffices to show that $\omega|_{\caA_{C'}}$ satisfies maximally mixed boundary conditions at the unique $\caS' \in \Bd(\Sigma_{C'})$. 
    This follows by applying \Cref{lem:restriction yields maximally mixed boundary conditions} to the density matrix corresponding to the restriction of $\omega$ to a finite region $\widetilde C\subset C$, chosen such that $\caS \in \Bd(\Sigma_{\widetilde C})$ and such that if $\widetilde C' = \widetilde C \cap C'$ then $\Sigma_{\widetilde C} \setminus \Sigma_{\widetilde C'}$ is an annulus and  $\caS' \in \Bd(\Sigma_{\widetilde C'})$.
\end{proof}

\section{String Operators} \label{sec:string operators}

\subsection{Drinfeld insertions} \label{subsec:Drinfeld insertions}

For each $X \in \Irr \ZC$ and each extended circle $\caS$ we fix a unit vector $w^{X}_{\caS} \in \caC(X \rightarrow \chi^{\otimes\caS})$ with respect to the trace inner product. We write $p_{\caS}^{X} \in \Tube_{\caS}$  for the corresponding minimal projector (see Proposition \ref{prop:matrix units for Tube_n}). 
In the case where $X=\I$ and $\caS=\caS_e$, we specify
\begin{equation} \label{eq:ground state boundary condition at puncture}
    w^{\I}_{\caS_e} = \frac{1}{\caD} \, \sum_{a \in \Irr \caC} d_a^{1/2} \,\,\, \adjincludegraphics[valign=c, width = 1.0cm]{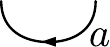}
\end{equation}
so that the corresponding minimal projector $p^{\I}_{\caS_e} = p^{\I}$ is the projector introduced in Eq. \eqref{eq:p^one defined}. (Note that all boundary components $\caS_e$ for $e \in \caE$ are isomorphic as extended circles). Since the extended circle $\caS$ will always be clear from context, it will always be dropped from the notation, so that $w^X$ stands for $w^X_{\caS}$ and $p^X$ stands for $p^X_{\caS}$.

\begin{convention} \label{conv:Drinfeld strands attaching to boundary components}
    Having fixed the boundary conditions $w^X$ for $X \in \Irr \ZC$ we introduce the following graphical convention for string diagrams:
    \begin{equation}
        \adjincludegraphics[valign=c, height = 1.2cm]{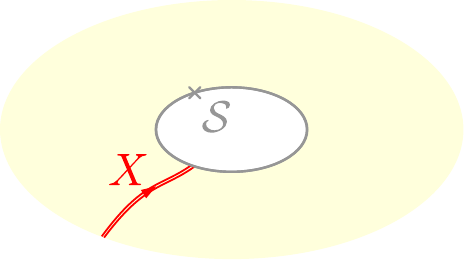} = \adjincludegraphics[valign=c, height = 1.2cm]{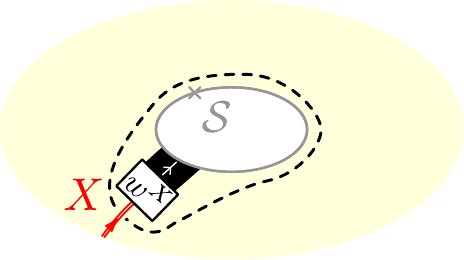}.
    \end{equation}
    Given an arbitrary boundary condition $w \in \caC(X \rightarrow \chi^{\otimes\caS})$ we also depict
    \begin{equation}
        \adjincludegraphics[valign=c, height = 1.2cm]{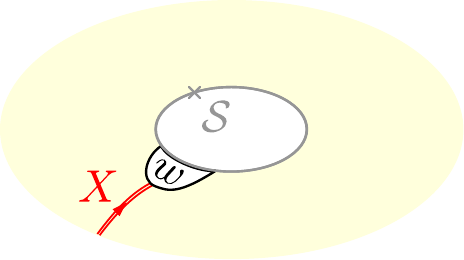} = \adjincludegraphics[valign=c, height = 1.2cm]{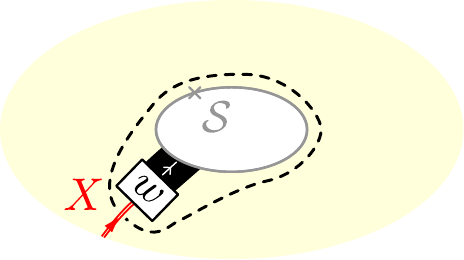}.
    \end{equation}
\end{convention}

Let $C$ be a finite connected region and let $\anchor$ be an anchor for $C$. Proposition \ref{prop:the great interface} implies that vectors of the the form
\begin{align*}
    \Psi_C^{\anchorino} \big( \al \otimes w_0 \otimes w^{X_1} &\otimes \cdots \otimes w^{X_m} \big) \\ &= \left( \prod_{\kappa=0}^m d_{X_{\kappa}} \right)^{1/2} \, \caD^{m} \times \sigma_C^{-1} \left( \adjincludegraphics[valign=c, height = 2.0cm]{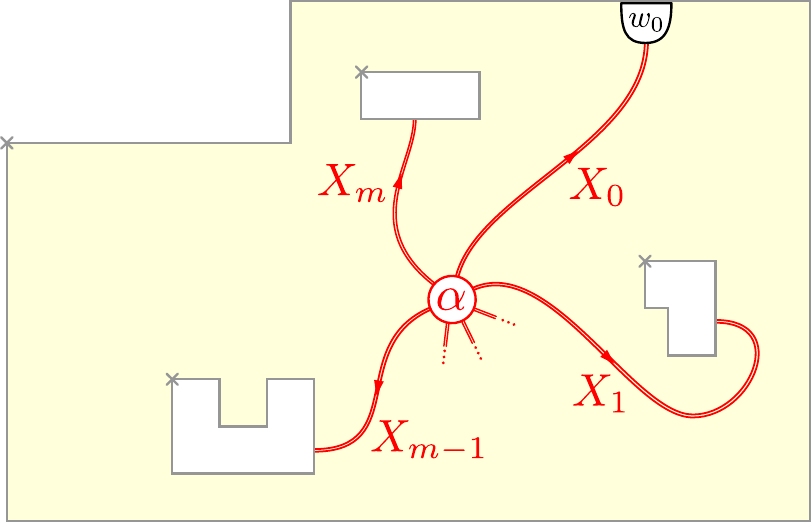} \right),
\end{align*}
for $X_0 \in \Irr \ZC$, $w_0 \in \caC(X_0 \rightarrow \chi^{\otimes\caS_0^{\anchorino}})$, and $\al \in \ZC( X_0^* \rightarrow X_1 \otimes \cdots \otimes X_m )$ span the subspace 
\begin{equation} \label{eq:skein subspace with free outer boundary}
    H_C^{(X_1, \cdots, X_m)} := H_C( \id, p^{X_1}, \cdots, p^{X_m} )
\end{equation}
defined in Eq. \eqref{eq:def H_C(p)}. Here we use the enumeration of connected boundary components induced by $\anchor$.

For any $\beta \in \ZC( X_1 \otimes \cdots \otimes X_m \rightarrow Y_1 \otimes \cdots \otimes Y_m )$ we define the Drinfeld insertion $\Dr_C^{\anchorino}[\beta] \in \caA_C$ which acts on $H_C^{(X_1, \cdots, X_m)}$ as
\begin{align} \label{eq:Drinfeld insertion defined}
\begin{split}
    \Dr_C^{\anchorino}[\beta] : \Psi_C^{\anchorino} \big( \al \otimes w_0 &  \otimes w^{X_1} \otimes \cdots \otimes w^{X_m} \big) \\  & \mapsto  \Psi_C^{\anchorino} \big( (\beta \circ \al) \otimes w_0 \otimes w^{Y_1} \otimes \cdots \otimes w^{Y_m} \big),
\end{split}
\end{align}
and annihilates the orthogonal complement of $H_C^{(X_1, \cdots, X_m)} \subset \caH_C$.
Graphically, the action of $\sigma_C \circ \Dr_C^{\anchorino}[\beta] \circ \sigma^{-1}_C : A(\Sigma_C) \rightarrow A(\Sigma_C)$ is given by
\begin{equation} \label{eq:Drinfeld insertion graphical}
    \adjincludegraphics[valign=c, height = 1.7cm]{graphical_vector_in_H_C.pdf} \,\,\, \mapsto \left( \prod_{\kappa = 1}^{m} \,\,\, \frac{ d_{Y_{\kappa}} }{ d_{X_{\kappa}} } \right)^{1/2} \,\, \adjincludegraphics[valign=c, height = 1.7cm]{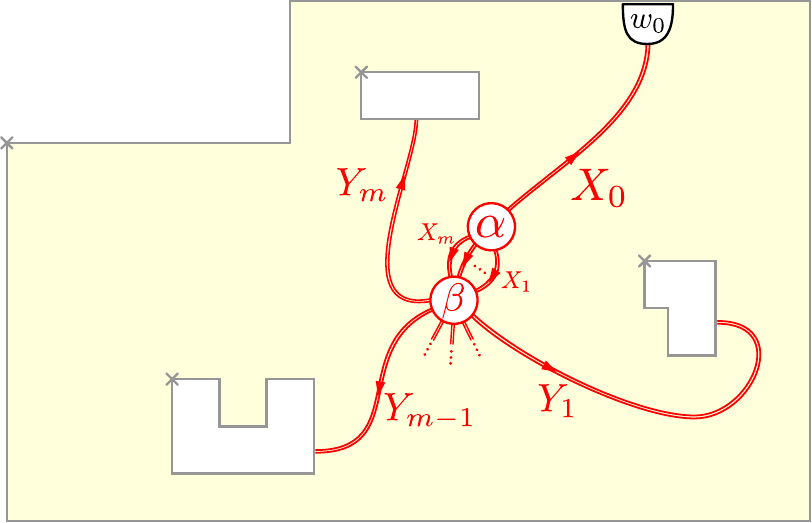}\,\,.
\end{equation}

\begin{lemma}\label{lem:Drinfeld insertion multiplicativity}
    Let $C$ be a finite connected region with $|\Bd(\Sigma_C)|=m+1$ and $\anchor$ an anchor for $C$. 
    If $\Theta = \bigoplus_{X\in \Irr(\ZC)} X$,
    then $\Dr_C^{\anchorino}$ defines a *-representation of $\End_{\ZC}(\Theta^{\otimes m})$ on 
    \begin{equation}
    \label{eq:domain of Dr-action}
        \bigoplus_{X_1,\ldots, X_m\in \Irr(\ZC)} H_C^{(X_1, \cdots,X_m)}.
    \end{equation}
\end{lemma}
\begin{proof}
    We have that 
    $$\End_{\ZC}(\Theta^m) \simeq \bigoplus_{\substack{X_1,\ldots, X_m\in \Irr(\ZC) \\ Y_1,\ldots, Y_m\in \Irr(\ZC)}} \ZC\left(X_1 \otimes \cdots \otimes X_m \to Y_1 \otimes \cdots \otimes Y_m \right).
    $$
    Let $\beta$ be a morphism belonging to one of these direct summands. 
    It is clear that $\Ran \, \Dr_C^{\anchorino}[\beta] \subset H_C^{(Y_1, \cdots, Y_m)}$, and using that $\Psi_C^{\anchorino}$ is an isomorphism of Hilbert spaces we find $(\Dr_C^{\anchorino}[\beta])^* = \Dr_C^{\anchorino}[ \beta^{\dag} ]$.
    Multiplicativity is immediate from the definition \eqref{eq:Drinfeld insertion defined}, and $\Dr_C^{\anchorino}[\id_{\Theta^{\otimes m}}]$ is the projection onto the subspace in \eqref{eq:domain of Dr-action}.
\end{proof}

We say two anchors are equivalent up to $\caS\in \Bd(\Sigma)$, if they induce the same enumeration on $\Bd(\Sigma)$ and if, after removing the edge connected to $\caS$ from either anchor, the remaining subgraphs are equivalent in the sense of equivalence of anchors. That is, one \emph{subanchor} may be transformed into the other by isotopy of $\Sigma$ which fixes all fiducial points.

\begin{lemma}
    If $C$ is a finite connected region and $\anchor,\anchor'$ are anchors on $\Sigma_C$, equivalent up to the outer boundary $\caS_0\in \Bd(\Sigma_C)$, 
    then $\Dr_C^{\anchorino} = \Dr_C^{\anchorino'}$.
\end{lemma}
\begin{proof}
    Let $\epsilon_{\anchorino}\subset \anchor, \epsilon_{\anchorino'}\subset \anchor'$ denote the edges connecting the anchor point to $\caS_0$ in either anchor.
    Since $\Psi_{C}^{\anchorino}$ depends only on the equivalence class of the anchor, the same is true for the representation $\Dr_C^{\anchorino}$.
    Therefore, without loss of generality, we may assume that $\anchor\setminus \epsilon_{\anchor}= \anchor'\setminus \epsilon_{\anchor'}$.
    Moreover, we may assume that 
    the edges $\epsilon_{\anchorino}, \epsilon_{\anchorino'}$ 
    have the same attaching point on $\caS_0$. 
    Cutting $\Sigma_C$ along the common subanchor will connect all the interior boundary components, producing a surface homeomorphic to an annulus. 
    Since $\anchor, \anchor'$ induce the same enumeration of $\Bd(\Sigma_C)$, the edges $\epsilon_{\anchorino}, \epsilon_{\anchorino'}$ have the same attaching point on the unique interior boundary of the cut surface.
    Since the circle $S^1$ is a deformation retract of the annulus, the 
    homotopy-classes of paths between two fixed points on either boundary component of the annulus is a $\pi_1(S^1)$-torsor.
    The generator is (the isotopy class of) a Dehn twist around $\caS_0$.
    It follows by \Cref{lem:anchors related by Dehn twists}, that there is a $z\in \Z$ such that 
    $$ \Psi_C^{\anchorino'} \big( \al \otimes w_0 \otimes w^{X_1} \otimes \cdots \otimes w^{X_m}\big) =  \theta_{X_0}^z \times \Psi_C^{\anchorino} \big( \al \otimes w_0 \otimes w^{X_1} \otimes \cdots \otimes w^{X_m}\big)$$
    for all $\al \in \ZC( X_0^* \rightarrow X_1 \otimes \cdots \otimes X_m )$ and all $w_0 \in \caC( X_0 \rightarrow \chi^{\otimes \caS_0^{\anchorino}} )$.

    That $\Dr_C^{\anchorino}[\beta] = \Dr_C^{\anchorino'}[\beta]$ for any $\beta \in \ZC( X_1 \otimes \cdots \otimes X_m \rightarrow Y_1 \otimes \cdots \otimes Y_m )$ now follows from the fact that these operators do not change the object $X_0$ assigned to the outer boundary component.
\end{proof}

\begin{lemma} \label{lem:string insertions commute with outer Tube actions}
    The actions $\Dr_C^{\anchorino}$ and $\frt_{\caS_0^{\anchorino}}$ commute for any finite connected region $C$ with anchor $\anchor$.    
    That is, $[ \Dr_C^{\anchorino}[\beta], \frt_{\caS^{\anchorino}_0}(a) ] = 0$ for all 
    $a \in \Tube_{\caS^{\anchorino}_0}$
    and $\beta\in \End_{\ZC}(\Theta^{\otimes m})$, where $|\Bd(\Sigma_C)|=m+1$. The same holds if $\frt_{\caS_0^{\anchorino}}$ is replaced by $\frt_{\caN}$ for a decorated submanifold $\caN\subset \caS_0^{\anchorino}$. 
\end{lemma}

\begin{proof}
    Let us write $\caS = \caS_0^{\anchorino}$. We first show $\big[\Dr_C^{\anchorino}[\beta], \frt_{\caS}(a) \big] = 0$ for any $a \in \Tube_{\caS}$. On the orthogonal complement of the skein subspace $H_C$ the commutator vanishes because $\Dr_C^{\anchorino}[\beta]$ annihilates $H_C^{\perp}$, and  $\frt_{\caS}(a) H_C^{\perp} \subset H_C^{\perp}$. It therefore remains to verify that the commutator vanishes on $H_C$. Since $\Psi_C^{\anchorino}$ is an isomorphism of unitary $\Tube$-modules it is sufficient to show that 
    $$ \big( \Psi_C^{\anchorino} \big)^{-1} \, \Dr_C^{\anchorino}[\beta] \, \Psi^{\anchorino}_C : \caC^*_{\anchorino}(\Sigma_C) \rightarrow \caC^*_{\anchorino}(\Sigma_C) $$
    commutes with the $\Tube$ action $\triangleright_0$ on $\caC^*_{\anchorino}$ corresponding under $\Psi_C^{\anchorino}$ to $\frt_{\caS}$. But $\triangleright_0$ acts only on the tensor factor $\caC(X_0 \rightarrow \chi^{\otimes\caS_0^{\anchorino}})$ of $\caC^*_{\anchorino}$, while $\big( \Psi_C^{\anchorino} \big)^{-1} \, \Dr_C^{\anchorino}[\beta] \, \Psi^{\anchorino}_C$ acts as identity on that factor by definition. We conclude that $\big[\Dr_C^{\anchorino}[\beta], \frt_{\caS}(a) \big] = 0$.
    
    Finally, let $\caN \subset \caS_0^{\anchorino}$ and $a\in A(\caN)$. As before, $[ \Dr_C^{\anchorino}[\beta], \frt_{\caN}(a) ]$ vanishes on $H_C^{\perp}$, and $\frt_{\caN}(a) =  \frt_{\caS_0^{\anchorino}}( \iota(a) )$ when acting on $H_C$ (see \Cref{lem:frt inclusions}) so the result for $\frt_\caN$ follows from that of $\frt_{\caS^{\anchorino}_0}$.
\end{proof}

\begin{lemma} \label{lem:string insertions preserve ground state constraints}
    Let $C$ be a finite connected region with anchor $\anchor$ and let $f$ be a face such that one of the following holds:
    \begin{itemize}
        \item none of the vertices of $f$ belong to $C$,
        \item $f$ belongs to $C$,
        \item $f$ sits on the outer boundary of $C$.
    \end{itemize}
    Then $B_f$ commutes with the action $\Dr_C^{\anchorino}$, i.e. $[ \Dr_C^{\anchorino}[\beta], B_f] = 0$.
\end{lemma}

\begin{proof}
    The first two cases are obvious.
    
    Write $\caS_0 = \caS^{\anchor}_0$ for the outer boundary. If $f$ sits on the outer boundary then the region $C^f$ decomposes into two non-empty regions $C^{\inn}$ and $C^{\out}$, where $C^{\inn}$ consists of the vertices and edges of $C^f$ that are entirely contained in $\Sigma_C$, and $C^{\out}$ consists of the vertices and edges that are entirely contained in $\Sigma^c_C$. Let $\caN_{\inn} := \partial \Sigma_{C^{\inn}} \cap \partial \Sigma_C$ be the decorated submanifold shared by $\partial \Sigma_{C^{\inn}}$ and $\partial \Sigma_C$. Let us specialise to the case where $C^{\inn}_0$ consists of the left two vertices of the face $f$ and the edge between them. By writing $B_f$ as
    \begin{equation} \label{eq:B_f in out decomposition}
        B_f \, \adjincludegraphics[valign=c, width = 2.0cm]{B_f_before.pdf} = \adjincludegraphics[valign=c, width = 2.0cm]{B_f_after.pdf} = \frac{1}{\caD^2} \, \sum_{a \in \Irr \caC} d_a \sum_{r, s \in \Irr \caC} \, d_r d_s \, \adjincludegraphics[valign=c, width = 2.0cm]{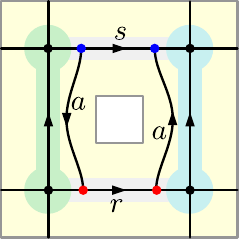}\,\,,
    \end{equation}
    where we indicated $C^{\inn}$ in green and $C^{\out}$ in blue, we see that $B_f$ belongs to the algebra generated by $\frt_{\caN_{\inn}}( A(\caN_{\inn}) )$ and $\frt_{\partial \Sigma_{C^{\out}}}( A(\partial \Sigma_{C^{\out}})$. Since $\Dr^{\anchorino}_C[\beta]$ commutes with $\frt_{\caN_{\inn}}( A(\caN_{\inn}) )$ by Lemma \ref{lem:string insertions commute with outer Tube actions}, and $\Dr^{\anchorino}_C[\beta]$ commutes with $\frt_{\partial \Sigma_{C^{\out}}}( \Tube_{\partial \Sigma_{C^{\out}}} )$ by disjoint supports, we find that $[ \Dr_C^{\anchorino}[\beta], B_f] = 0$.

    All other cases can be treated in the same way, by decomposing the insertion of an $a$-loop as above, or in one of the following ways:
    \begin{equation*}
        \sum_{r, s \in \Irr \caC} \, d_r d_s \, \adjincludegraphics[valign=c, width = 2.0cm]{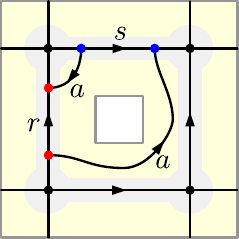} = \sum_{r, s, t \in \Irr \caC} \, d_r d_s d_t \, \adjincludegraphics[valign=c, width = 2.0cm]{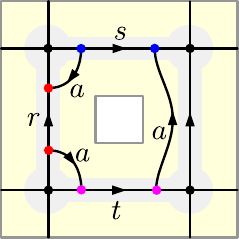} =  \sum_{r, s, t, u \in \Irr \caC} \, d_r d_s d_t d_u \, \adjincludegraphics[valign=c, width = 2.0cm]{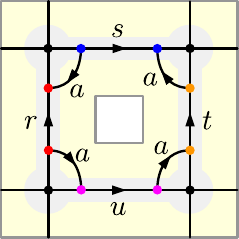}\,\,,
    \end{equation*}
    or in a rotated version of one of these.
\end{proof}

\begin{remark}
    \Cref{lem:string insertions preserve ground state constraints} also follows from \Cref{lem:string insertions commute with outer Tube actions}, \Cref{lem:gluing skein subspaces}, and the statement that $\gl$ descends to an isomorphism $A(\Sigma_{C}) \otimes_{A(\caI)} A(\Sigma_{C'}) \simeq A(\Sigma_{C} \sqcup_{\caI} \Sigma_{C'})$, where $\caI$ is a common gluing boundary (see \cite[Theorem~5.2.10]{walker2006tqft}).
\end{remark}

\subsection{Inclusion Lemma} \label{subsec:inclusion lemma}

Let $D = C \sqcup_I C'$ be a finite connected region obtained by gluing $C$ and $C'$ along a self-avoiding dual path $I$ as described in Section \ref{ssubsec:gluing skein subspaces}. Then all inner boundary components of $\Sigma_C$ are also inner boundary components of $\Sigma_D$. Let $\anchor_C$ be an anchor of $\Sigma_C$ whose attachment point to the outer boundary component $\caS^{\anchorino_C}_0$ of $\Sigma_C$ also lies on the outer boundary component of $\Sigma_D$. Let $\anchor_D$ be an anchor for $\Sigma_D$ such that $\caS^{\anchorino_D}_0$ is the outer boundary of $\Sigma_D$ and such that $\caS^{\anchorino_C}_{\kappa} = \caS^{\anchorino_D}_{\kappa + \lambda}$ for $\kappa = 1, \cdots, m$ and some fixed \emph{offset} $\lambda \in \{ 0, 1, \cdots, n-m \}$. If the graph of $\anchor_D$ moreover contains the graph of $\anchor_C$ as a subgraph, then we say $\anchor_D$ \emph{extends} $\anchor_C$ with offset $\lambda$. More generally, any anchor $\anchor'_D \sim \anchor_D$ is also said to extend $\anchor_C$ with offset $\lambda$.

Write $P^{(X_1, \cdots, X_n)}_D$ for the orthogonal projector onto the skein subspace $H_D^{(X_1, \cdots, X_n)} \subset \caH_D$ defined in Eq. \eqref{eq:skein subspace with free outer boundary}.
\begin{lemma} \label{lem:inclusion lemma}
    Let $D$ be obtained by gluing $C$ and $C'$ along a self-avoiding dual path $I$, and let $\anchor_D$ extend $\anchor_C$ with offset $\lambda$. For any $X_1, \cdots, X_n, Y_1, \cdots, Y_n \in \Irr \ZC$ and any $\beta \in \ZC \big( X_{\lambda+1} \otimes \cdots \otimes X_{\lambda+m} \rightarrow Y_1 \otimes \cdots \otimes Y_m \big)$ we have
    $$  \Dr^{\anchorino_C}_C[\beta] P^{(X_1, \cdots, X_n)}_D = \Dr_{D}^{\anchorino_D}[\id_{X_1 \otimes \cdots \otimes X_{\lambda}} \otimes \beta \otimes \id_{X_{\lambda + m + 1} \otimes \cdots \otimes  X_n}].  $$
\end{lemma}

\begin{proof}
    Let
    $$ \psi  = \Psi_D^{\anchorino_D} \big(  \al \otimes w_0 \otimes w^{X_1} \otimes \cdots \otimes w^{X_n} \big) $$
    for arbitrary $X_0 \in \Irr \ZC$, $\al \in \ZC( X_0^* \rightarrow X_1 \otimes \cdots \otimes X_m)$, and $w_0 \in \caC( X_0 \rightarrow \chi^{\otimes\caS_0})$. 
    By Proposition \ref{prop:the great interface}, the range of $P_D^{(X_1, \cdots, X_n)}$ is spanned by vectors of this form.
    In order to compute the action of $\Dr_C^{\anchorino_C}[\beta]$ on $\psi$ we will write $\psi$ as a linear combination of pure tensors according to the decomposition $\caH_D = \caH_C \otimes \caH_{C'}$. By replacing $\anchor_D$ with an equivalent anchor, we may without loss of generality assume that the anchor point of $\anchor_D$ lies in $\Sigma_C$ and that its attachment point to the exterior boundary component also belongs to the exterior boundary of $\Sigma_{C'}$. 
    Schematically,
    $$
        \sigma_D(\psi) = \adjincludegraphics[valign=c, height = 2.5cm]{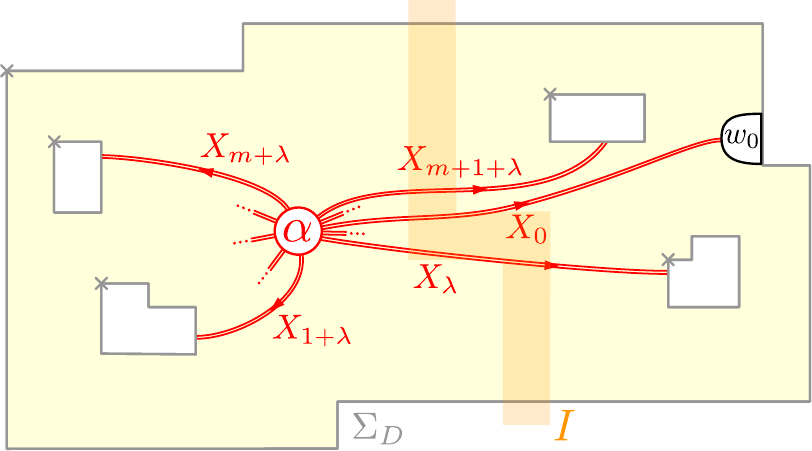} \,\, , 
    $$
    where the Drinfeld center strands lie on the anchor $\anchor_D$. By unpacking Convention \ref{conv:Drinfeld strands attaching to boundary components} for the exterior boundary condition $w_0$ and noting that the corresponding dotted line can be contracted to a point, we obtain
    $$
        \sigma_D(\psi) = \adjincludegraphics[valign=c, height = 2.5cm]{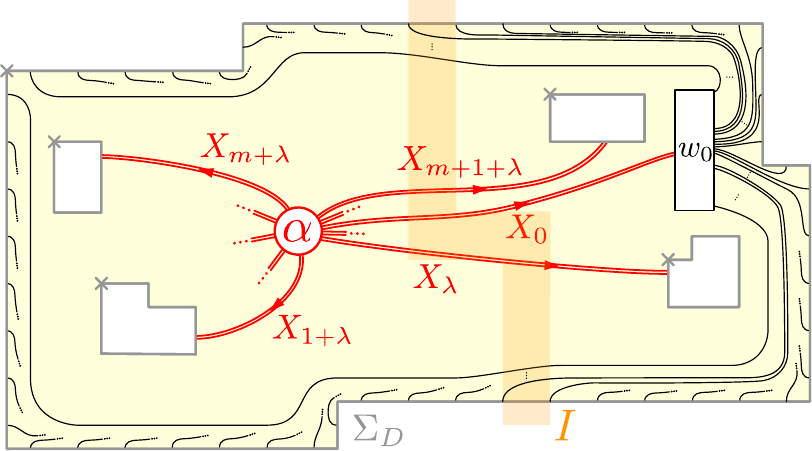} \,.
    $$
    This illustrates that by isotopy we may assume that anchor lines connected to inner boundary components of $C'$ as well as all the edges attaching to $\caS_0^{\anchorino_C}$ intersect the cut $I$ exactly once, and no other edges of the string diagram intersect $I$. Since $I$ is contractible, we may apply a decomposition \eqref{eq:decomposition into simples} to get a sum over string diagrams which have a single strand, labelled by a simple object $a \in \Irr \caC$, crossing $I$ along some edge in $\caE_I$. (Note that the assumptions on $I$ made in Section \ref{ssubsec:gluing skein subspaces} guarantee that $\caE_I$ is not empty.) 
    With this arrangement of the string diagrams we find that each summand factorises under gluing. That is, the restrictions to $\Sigma_C, \Sigma_{C'}$ are valid string diagrams. Schematically,
    \begin{align*}
        \sigma_D(\psi) &= \sum_{a \in \Irr \caC} \, d_a \, \adjincludegraphics[valign=c, height = 2.5cm]{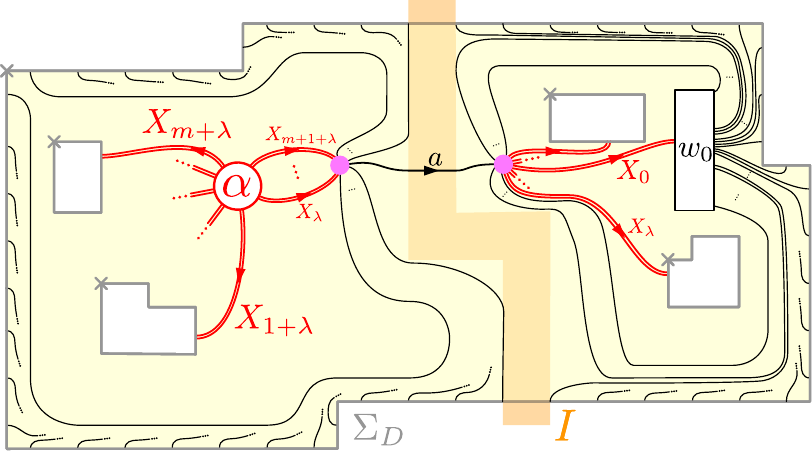} \\
        &= \sum_{a \in \Irr \caC} \, d_a  \, \gl \left( \adjincludegraphics[valign=c, height = 2.5cm]{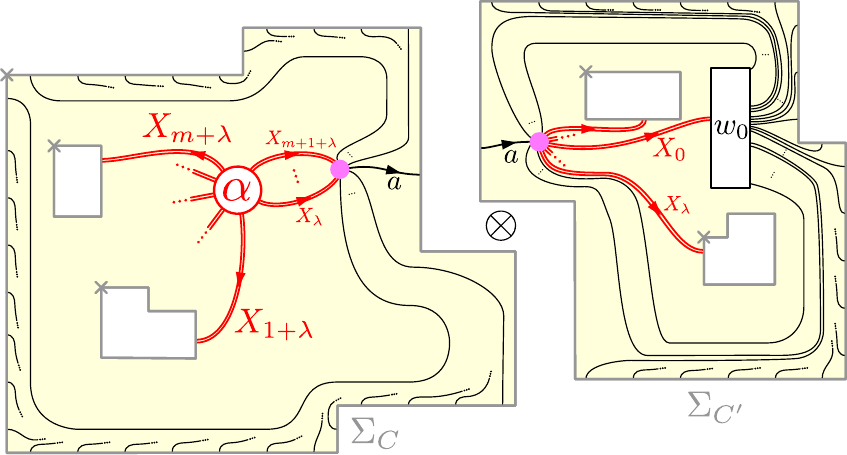} \right) \,\, ,
    \end{align*}
    where $\gl : A(\Sigma_C) \otimes A(\Sigma_{C'}) \to A(\Sigma_D)$ is the gluing map of Section \ref{subsec:gluing}. 
    Using Lemma \ref{lem:gluing skein subspaces} we express $\psi$ as the image under $B_I$ of a sum of product vectors in $\caH_{C}\otimes \caH_{C'}$:
    \begin{align*}
        \psi &= \caD^{\abs{I_{\inn}}} \, B_I \times \sum_{a} \, d_a^{1/2} \\
        &\, \sigma_C^{-1} \left( \adjincludegraphics[valign=c, height=2.5cm]{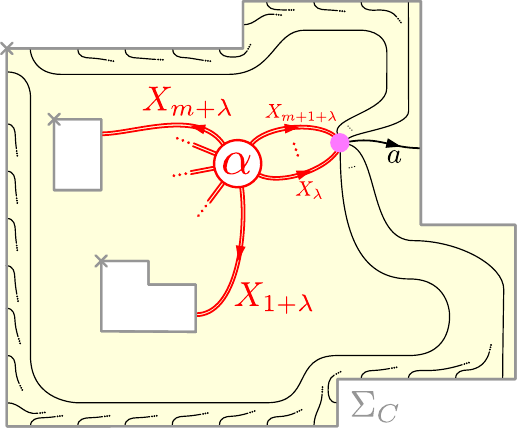} \right) \otimes \sigma_{C'}^{-1} \left( \adjincludegraphics[valign=c, height=2.5cm]{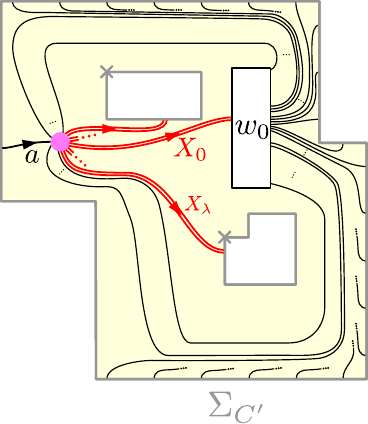} \right).
    \end{align*}
    Lemma \ref{lem:string insertions preserve ground state constraints} implies that $\Dr_C^{\anchorino_C}[\beta]$ commutes with $B_I$, so we can evaluate $\Dr_C^{\anchorino_C}[\beta]$ directly on the first factors of each summand. Here we can reinstate a dotted line and push it towards the outer boundary so, possibly after decomposing the identity on $X_{m+1+\lambda}\otimes \cdots \otimes X_{\lambda}$ in $\ZC$, these string diagrams are of the form \eqref{eq:Drinfeld insertion graphical}. That is, the action is given by \eqref{eq:Drinfeld insertion defined}, and we obtain
    \begin{align*}
        \Dr_C^{\anchorino_C}[\beta] \psi &= \caD^{\abs{I_{\inn}}} \, B_I \times  \left( \frac{ \prod_{\kappa=1}^m d_{Y_{\kappa}} }{ \prod_{\kappa = 1}^m \, d_{X_{\kappa + \lambda}} } \right)^{1/2} \sum_{a} \, d_a^{1/2} \\
        &\times \, \sigma_C^{-1} \left( \adjincludegraphics[valign=c, height=2.5cm]{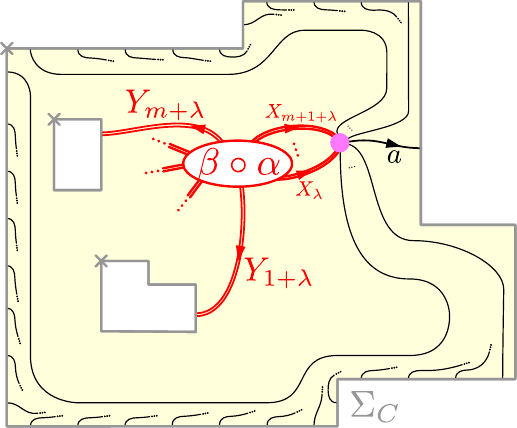} \right) \otimes \sigma_{C'}^{-1} \left( \adjincludegraphics[valign=c, height=2.5cm]{inclusion_4_right.pdf} \right) \,\, .
    \end{align*}
    The action of $\Dr_C^{\anchorino_C}[\beta]$ results in a string diagram for every summand with the morphism $\beta \circ \alpha$ at the anchor point, while the boundary condition at $S_0^{\anchorino_C}$ is left unchanged. The tensor product decomposition can now be undone, resulting in a string diagram on $\Sigma_D$ which is again of the form \eqref{eq:Drinfeld insertion graphical}, it is the image under $\Psi_D^{\anchorino_D}$ of a product vector in $\caC_{\anchorino}^*(\Sigma_D)$.
    We find $\Dr_C^{\anchorino_C}[\beta] \psi = \Dr_D^{\anchorino_D}[\id_{X_1 \otimes \cdots \otimes X_{\lambda}} \otimes \beta \otimes \id_{X_{\lambda + m + 1} \otimes \cdots \otimes  X_n}] \psi$ by comparing boundary conditions and the morphism at the anchor point. That is,
    \begin{align*}
        \Dr_C^{\anchorino_C}[\beta] \psi &= \left( \frac{ \prod_{\kappa=1}^m d_{Y_{\kappa}} }{ \prod_{\kappa = 1}^m d_{X_{\kappa + \lambda}} } \right)^{1/2} \sigma_D^{-1} \left( \adjincludegraphics[valign=c, height=2.5cm]{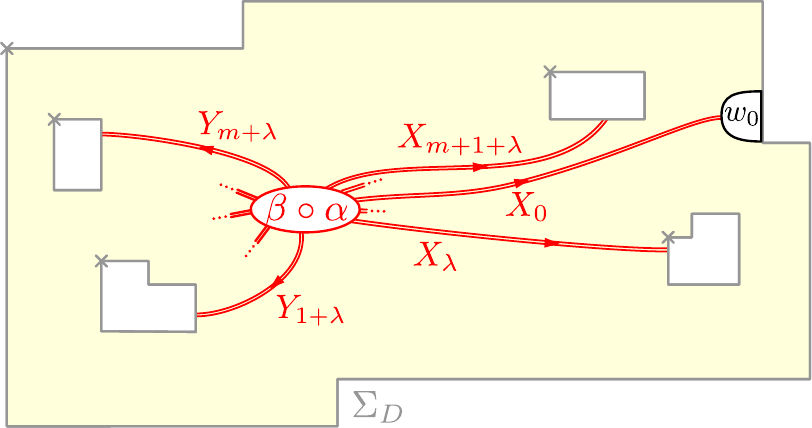} \right) \\
        &= \Dr_D^{\anchorino_D}[\id_{X_1 \otimes \cdots \otimes X_{\lambda}} \otimes \beta \otimes \id_{X_{\lambda + m + 1} \otimes \cdots \otimes  X_n}] \psi.
    \end{align*}
    Since the operator $\Dr_D^{\anchorino_D}[\id_{X_1 \otimes \cdots \otimes X_{\lambda}} \otimes \beta \otimes \id_{X_{\lambda + m + 1} \otimes \cdots \otimes  X_n}]$ annihilates the orthogonal complement of $P_D^{(X_1, \cdots, X_n)}$, this proves the lemma.
\end{proof}

\subsection{Links} \label{subsec:links}

A \emph{link} is a finite dual path $L = (f_1, \cdots, f_l)$ consisting of at least five faces, such that
\begin{itemize}
    \item $f_i$ and $f_j$ share an edge if and only if $j = i \pm 1$ (or $i = j$). In particular, $L$ is self-avoiding.
    \item the first four faces of $L$ lie on a straight line, as do the last four faces of $L$.
\end{itemize}

Let $e_{\ii} =: \partial_{\ii} L$ denote the edge between faces $f_1$ and $f_2$ and $e_{\f} =: \partial_{\f} L$ the edge between the last two faces $f_{l-1}$ and $f_l$. The puncture at $e_{\ii}$, resp. $e_{\f}$, is called the \emph{initial}, resp. \emph{final}, puncture of $L$.
To a link $L$ we associate a region $C^L$ consisting of all vertices and edges belonging to the faces constituting the link except  for $e_{\ii}$ and $ e_{\f}$, and all the faces $f_3, \cdots, f_{l-2}$, which we call the \emph{bulk faces} of the link. The associated surface $\Sigma_L := \Sigma_{C^L}$ is a twice punctured disk, with punctures at the edges $e_{\ii}$ and $e_{\f}$. The assumptions on $L$ guarantee that $C^L$ can be cut along some path in between the two punctures, for example cutting transversally across the third face.

Let $C^{L, r} \subset C^L$ be the subregion consisting of all vertices and edges on the right side of $C^L$ with respect to the direction of $L$. We call the surface $\Sigma_{L, r} := \Sigma_{C^{L, r}}$ the \emph{right strip} of $L$.

We fix an anchor $\anchor_L$ for $C^L$ whose equivalence class up to the outer boundary component $\caS_0^{\anchorino_L}$ is uniquely determined by the condition that the underlying graph of $\anchor_L$ lies in the right strip $\Sigma_{L,r}$. In particular, $\caS_{1}^{\anchorino_L}=\caS_{e_{\ii}}$ is the initial puncture, and $\caS^{\anchorino_L}_2=\caS_{e_{\f}}$ is the final puncture of $L$.
To see that this condition indeed determines the equivalence class up to $\caS_0^{\anchorino_L}$ uniquely, note first that $\Sigma_{L,r}$ is a disk, so any two anchors with the same attachment points in $\Sigma_{L,r}$ are equivalent by isotopy in $\Sigma_{L,r}$ keeping $\partial \Sigma_{L, r}$ fixed. The boundary components $\caS_1^{\anchorino_L}$ and $\caS_2^{\anchorino_L}$ are rectangles with fiducial points sitting on a corner, and $\partial \Sigma_{L,r}$ contains exactly one straight line of each of these boundary components. It follows that different choices of attachment point to these boundary components are also related to each other by isotopy of $\Sigma_L$ keeping the fiducial points fixed. Different choices of attachment point to the outer boundary are related by isotopy of $\Sigma_L$, where we need not keep the fiducial point on $\caS_0^{\anchorino_L}$ fixed.
\begin{figure}[h]
    \centering
    \includegraphics[width=0.5\textwidth]{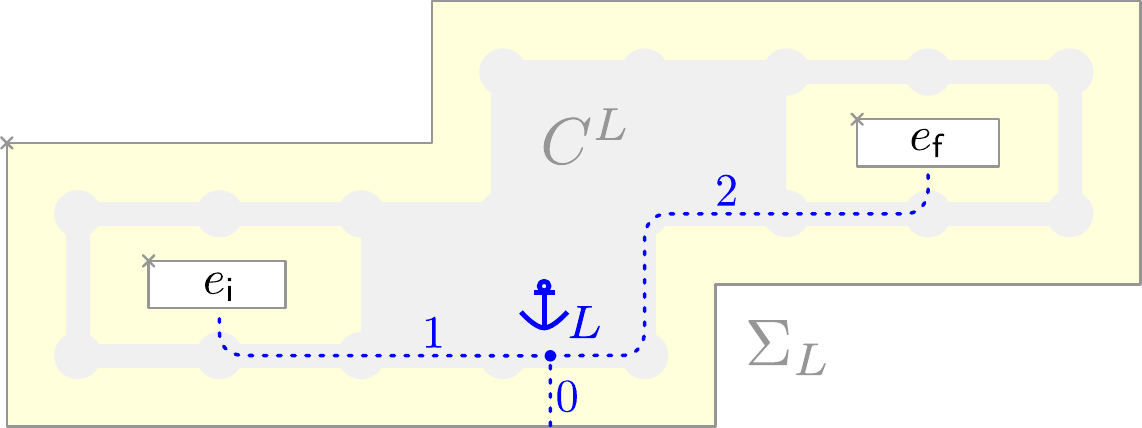}
    \caption{The anatomy of a link $L$. The region $C^L$ is indicated in light grey, the surface $\Sigma_L$ in light yellow, and the anchor $\anchor_L$ in blue.}
    \label{fig:link}
\end{figure}

We write $\supp(L) = V_{C^L}$ for the \emph{support} of the link $L$. We will often write $L$ instead of $C^L$ in notations depending on the region $C^L$. For example, $\caA_L := \caA_{C^L}, \caH_L := \caH_{C^L}$, $H_L := H_{C^L}$, $\sigma_L := \sigma_{C^L}$, $\Psi_L^{\anchorino} := \Psi_{C^L}^{\anchorino}$, etc. Consistent with this notation, we have $F_L = F_{C^L} = \{ f_3, \cdots, f_{l-2} \}$. 
We also suppress the standardised choice of anchor $\anchor_L$, writing $\Dr_L[\beta] = \Dr_{C^L}^{\anchorino_L}[\beta]$. The boundary components of $\Sigma_L$ will always have the linear ordering induced by $\anchor_L$, so that $H_L(p_0, p_1, p_2) = H_L(\underline p)$ with $p_{\kappa} = p_{\caS^{\anchorino_L}_{\kappa}}$ for $\kappa = 0, 1, 2$. We also write $\frt^L_{\kappa} := \frt_{\caS_{\kappa}^{\anchorino_L}}$ for the $\Tube$-actions on these three boundary components.

Two links $L_1 = (f_1, \cdots, f_k)$ and $L_2 = (f'_1\cdots, f'_l)$ are \emph{composable} if $\partial_{\ii}(L_1)=\partial_\f(L_2)$ and the composite $L_2 \wedge L_1 := (f'_1, \cdots, f'_{l-1}, f'_l, f_2, \cdots, f_k)$ is again a link. Being composable is \emph{not} a symmetric relation. Note that the right strip of $L_2 \wedge L_1$ is the union of the right strips of $L_1$ and $L_2$.

\subsection{Unitary gates} \label{subsec:unitary gates}

We construct unitaries which produce and annihilate anyon pairs on the punctures of a given link, or move an anyon from one puncture to another. The construction extends the hopping operators of \cite{christian2023lattice, green2024enriched}.

\subsubsection{Partial isometries for pair creation and hopping} \label{ssubsec:pair creation and hopping}

For each $X \in \Irr \ZC$ we fix a (non-canonical) unitary isomorphism $\zeta_X : X^* \rightarrow \bar X$. We represent $\zeta_X$ in the graphical calculus by a solid box, and $\zeta_X^{\dag}$ by an empty box:
\begin{equation} \label{eq:star to bar}
    \zeta_X = \,\,\, \adjincludegraphics[valign=c, height = 0.8cm]{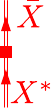} \,\,\, = \,\,\, \adjincludegraphics[valign=c, height = 0.8cm]{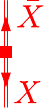} \,\,,\quad\quad \zeta^{-1}_X = \zeta^{\dag}_X = \,\,\, \adjincludegraphics[valign=c, height = 0.8cm]{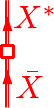} \,\,\, = \,\,\, \adjincludegraphics[valign=c, height = 0.8cm]{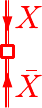}\,.
\end{equation}
We will moreover make use of the morphisms $\ev^{\dag}_X$ and $\coev_X^{\dag}$, for which we introduce the following graphical representations:
\begin{equation} \label{eq:graphical gagger duals}
    \ev_X^{\dag} = \,\,\, \adjincludegraphics[valign=c, width = 1.5cm]{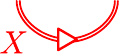} \,\,\, , \,\,\, \coev_X^{\dag} = \,\,\, \adjincludegraphics[valign=c, width = 1.5cm]{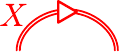} \,\, .
\end{equation}Let L be a link. We define an operator which creates an $X \bar X$ anyon pair at the punctures of $L$ by
\begin{equation} \label{eq:pair creation operator}
    \Dr_L^{(\I\I \rightarrow X \bar X)} := d_X^{-1/2} \, \Dr_{L}[ (\id_X \otimes \zeta_X) \circ \coev_X ] = d_X^{-1/2} \, \Dr_{L} \left[ \,\,\,\,\, \adjincludegraphics[valign=c, height = 0.8cm]{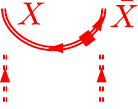} \right],
\end{equation}
where we indicate explicitly the incoming and outgoing identity strands using dotted double lines in the diagrams, so that source and target objects of the morphisms are clear. These identity strands are attached to the rest of the diagram using unitors in an arbitrary way, all ways representing the same morphism by Mac Lane's coherence theorem.

The adjoint of this operator annihilates an $X \bar X$ pair:
\begin{equation} \label{eq:pair annihilation operator}
    \Dr_L^{(X \bar X \rightarrow \I \I)} := (\Dr_L^{(\I\I \rightarrow X \bar X)})^* = d_X^{-1/2} \, \Dr_{L}[ \coev_X^{\dag} \circ (\id_X \otimes \zeta_X^{\dag}) ] = d_X^{-1/2} \, \Dr_{L} \left[ \,\,\,\,\, \adjincludegraphics[valign=c, height = 0.8cm]{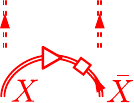} \right].
\end{equation}

Similarly, we define operators which move an $X$ anyon between the initial and final punctures of $L$:
\begin{equation} \label{eq:hopping operator}
    \Dr_L^{(X \I \rightarrow \I X)} := \Dr_{L}\left[ \adjincludegraphics[valign=c, height = 1.0cm]{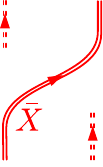} \right], \quad \Dr_L^{(\I X \rightarrow X \I)} := (\Dr_L^{(\I X \rightarrow X \I)})^* = \Dr_{L}\left[ \adjincludegraphics[valign=c, height = 1.0cm]{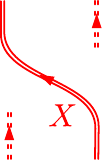} \right].
\end{equation}

Write $B_L = B_{C_L} = \prod_{f \in F_L} B_f$.
\begin{lemma} \label{lem:intertwining properties of hopping and pair creation}
    The pair creation operator $\Dr_L^{(\I\I \rightarrow X \bar X )} \in \caA_L$ is a partial isometry between $H_L(P^{\I}, p^{\I}, p^{\I})$ and $H_L(P^{\I}, p^X, p^{\bar X})$, with
    \begin{align*}
        P_L^{\I \I} &:= \frt^L_1( p^{\I} ) \frt^L_2(p^{\I}) B_L
        = (\Dr_L^{(\I\I \rightarrow X \bar X )})^* \Dr_L^{(\I\I \rightarrow X \bar X )}, \\
        P_L^{X \bar X} &:= \frt^L_0(P^{\I}) \frt^L_1( p^X ) \frt^L_2(p^{\bar X}) B_L
        = \Dr_L^{(\I\I \rightarrow X \bar X )} (\Dr_L^{(\I\I \rightarrow X \bar X )})^*.
        \intertext{Similarly, the hopping operator $\Dr_L^{(X \I \rightarrow \I X)} \in \caA_L$ is a partial isometry between $H_L(P^{\bar X}, p^{X}, p^{\I})$ and $H_L(P^{\bar X}, p^{\I}, p^{X})$, with } 
         P_L^{X \I} &:= \frt^L_1( p^{X} ) \frt^L_2(p^{\I}) B_L 
        = (\Dr_L^{(X \I \rightarrow \I X)})^* \Dr_L^{(X \I \rightarrow \I X)}, \\
        P_L^{\I X} &:= \frt^L_1( p^{\I} ) \frt^L_2(p^{X}) B_L 
        = \Dr_L^{(X \I \rightarrow \I X)} (\Dr_L^{(X \I \rightarrow \I X)})^*. 
    \end{align*}
\end{lemma}

\begin{proof}
    Note that $\Dr_L^{(\I \I \rightarrow X \bar X)}$ annihilates the orthogonal complement of $H_L(\I, p^{\I}, p^{\I}) = H_L(P^{\I}, p^{\I}, p^{\I})$ by definition. The operator
    $$ (\Dr_L^{(\I \I \rightarrow X \bar X)})^* \, \Dr_L^{(\I \I \rightarrow X \bar X)} = d_X^{-1} \, \Dr_L \left[ \adjincludegraphics[valign=c, width = 1.5cm]{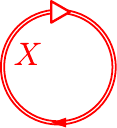} \right] $$
    acts as identity on this subspace. Indeed, the morphism in the argument of the right hand side is precisely the trace of $\id_X$ (see for example \cite[Remark~3.14]{penneys2018unitary}). This shows the first claim. 

    Similarly, $(\Dr_L^{(\I \I \rightarrow X \bar X)})^*$ annihilates the orthogonal complement of $H_L(\id, p^X, p^{\bar X})$ by definition, and it also annihilates $H_L(P^Y, p^X, p^{\bar X})$ for $Y \neq \I$. Indeed, this subspace is mapped into $H_L(P^Y, p^{\I}, p^{\I}) = \{0\}$ which is trivial by \Cref{prop:the great interface} and the fact that $\ZC( Y^* \rightarrow \I ) = \{0\}$.
    
    By noting that the operator
    $$ \Dr_L^{(\I\I \rightarrow X \bar X )} (\Dr_L^{(\I\I \rightarrow X \bar X )})^* = d_X^{-1} \Dr_L\left[ \adjincludegraphics[valign=c, width = 1.5cm]{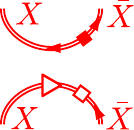} \right] $$
    acts as identity on the remaining subspace $H_L(P^{\I}, p^X, p^{\bar X})$ (again see \cite[Remark~3.14]{penneys2018unitary}), we obtain the second claim.

    The hopping operator $\Dr_L^{(X \I \rightarrow \I X)}$ annihilates the orthogonal complement of $H_L(\id, p^{X}, p^{\I}) = H_L(P^{\bar X}, p^{X}, p^{\I})$ by definition, and the operator
    $$ (\Dr_L^{(X \I \rightarrow \I X)})^* \Dr_L^{(X \I \rightarrow \I X)} = \Dr_L\left[  \id_{X \otimes \I}  \right] $$
    acts on this subspace as the identity. This shows the third equality. The proof of the last equality is identical to that of the third.
\end{proof}

\subsubsection{Unitary gates and concatenation} \label{ssubsec:definition of unitary gates}

For any link $L$ we have the projectors $P^{\I\I}_L, P_L^{\I X}, P_L^{X \I}, P_L^{X \bar X} \in \caB(\caH_L)$ defined in Lemma \ref{lem:intertwining properties of hopping and pair creation}. Note that all these projectors are all orthogonal to each other. We write
$P^{\I ?}_L := P_L^{\I \I} + P_L^{\I X}$,
$P^{X ?}_L := P_L^{X \I} + P_L^{X \bar X }$,
and $P^{??}_L = P^{X ?}_L + P^{\I ?}_L = P^{\I\I} + P_L^{X \I} + P_L^{\I X} + P_L^{X \bar X}$.
The definition of $P^{\I ?}_L$ and $P^{??}_L$ depends on the object $X$, which will always be clear from context.

For any $X \in \Irr \ZC$ and any link $L$ we define a self-adjoint unitary $u_L^X$ by
\begin{equation} \label{eq:unitary gate defined}
    u_L^X := \Dr_L^{(\I \I \rightarrow X \bar X)} + \Dr_L^{(X \bar X \rightarrow \I \I)} + \Dr_L^{(X \I \rightarrow \I X)} + \Dr_L^{(\I X \rightarrow X \I)} + (\I - P^{??}_L).
\end{equation}
Note that $u_L^{\I} = \I$. It is immediate from Lemma \ref{lem:intertwining properties of hopping and pair creation} that the unitary $u_L^X$ satisfies the following intertwining properties:
\begin{equation} \label{eq:intertwining properties of u_L}
    u^X_L P_L^{\I \I} = P_L^{X \bar X} u^X_L, \quad \quad u^X_L P_L^{ \I X} = P_L^{X \I } u^X_L,
\end{equation}
which combine into
\begin{equation}\label{eq:combined intertwining}
    u^X_L P_L^{\I ?} = P_L^{X ?} u^X_L.
\end{equation}
We also see that if $L_1$ and $L_2$ are composable links, 
then
\begin{equation}\label{eq:overlapping link}
    u_{L_1}^X P^{\I \I}_{L_2} = P^{\I ?}_{L_2} u_{L_1}^X P^{\I \I}_{L_2}.
\end{equation}
Combined, these equations show that 
\begin{equation*}
    u_{L_2}^X u_{L_1}^X P^{\I \I}_{L_2} = 
    P^{X ?}_{L_2} u_{L_2}^X u_{L_1}^X  P^{\I \I}_{L_2},
\end{equation*}
which illustrates a key feature of the unitary gates: when acting sequentially along a \emph{chain} of composable links on a state that satisfies ground state constraints on all but finitely many of those links, they will eventually start moving an $X$ anyon to infinity.

We also have the following concatenation property.
\begin{lemma} \label{lem:concatenation lemma}
    Let $X\in \Irr \ZC$ and let $L_1$ and $L_2$ be composable links with $L = L_2 \wedge L_1$. Then
    $$ u^X_{L_2} u^X_{L_1} P_L^{\I ?} = u_L^X P_L^{\I ?}. $$
\end{lemma}

\begin{proof}
    Using Lemmas \ref{lem:string insertions preserve ground state constraints} and \ref{lem:intertwining properties of hopping and pair creation} we find
    \begin{align*}
        u^X_{L_2} u^X_{L_1} P_L^{\I ?} &= \Dr_{L_2}^{(\I X \rightarrow X \I)} \Dr_{L_1}^{(\I \I \rightarrow X \bar X)} P_L^{\I \I} \quad + \quad \Dr_{L_2}^{(\I X \rightarrow X \I)} \Dr_{L_1}^{(\I X \rightarrow X \I)} P_L^{\I X},
        \intertext{and}
        u^X_L P_L^{\I ?} &= \Dr_L^{( \I \I \rightarrow X \bar X)} P_L^{\I \I} \quad + \quad  \Dr_L^{(\I X \rightarrow X \I)} P_L^{\I X}.
    \end{align*}
    It is therefore sufficient to show that
    \begin{equation} \label{eq:first concatenation equality}
        \Dr_{L_2}^{(\I X \rightarrow X \I)} \Dr_{L_1}^{(\I \I \rightarrow X \bar X)} P_L^{\I \I} = \Dr_L^{( \I \I \rightarrow X \bar X)} P_L^{\I \I},
    \end{equation}
    and
    \begin{equation} \label{eq:second concatenation equality}
        \Dr_{L_2}^{(\I X \rightarrow X \I)} \Dr_{L_1}^{(\I X \rightarrow X \I)} P_L^{\I X} = \Dr_L^{(\I X \rightarrow X \I)} P_L^{\I X}.
    \end{equation}

    Let's first derive Eq. \eqref{eq:first concatenation equality}. Consider the region $D = C^{L_1} \cup C^{L_2}$. We may describe $D$ as the region obtained from $C^L$ by removing the faces $L_1\cap L_2$ as well as the edge shared between these faces. Since $L = L_2 \wedge L_1$ is a link, the region $D$ has an associated surface $\Sigma_D$ which is homeomorphic to a disk with three holes cut out, corresponding to the punctures of $L_1$ and $L_2$ where the initial puncture of $L_1$ is the same as the final puncture of $L_2$.

    For invoking the inclusion lemma, we regard $D$ as being obtained by gluing $C^{L_1}$ and a uniquely determined region $C^{L_1, c}$ along the dual path $I_1$ cutting across the third face of $L_2$:
    \begin{center}
	   \includegraphics[width = 0.5\textwidth]{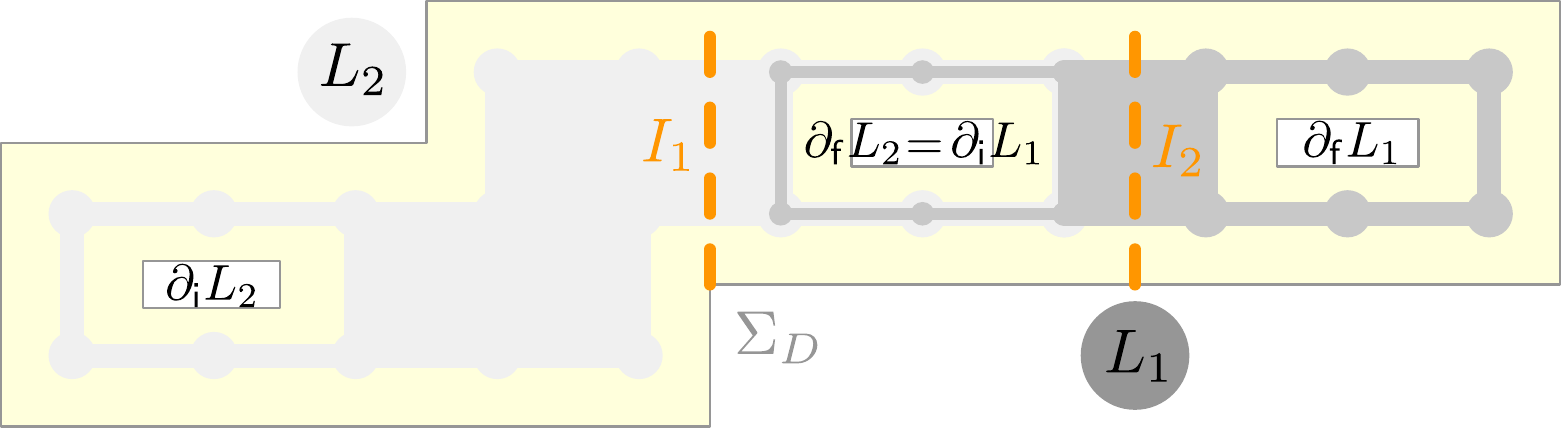} 
    \end{center}
    Similarly we may regard $D$ as being obtained by gluing $C^{L_2}$ and a uniquely determined region $C^{L_2, c}$ along the dual path $I_2$ cutting across the third to last face of $L_1$.
    The anchors $\anchor_{L_1}$ and $\anchor_{L_2}$ lie on $\Sigma_D$ as follows:
    \begin{center}
	   \includegraphics[width = 0.5\textwidth]{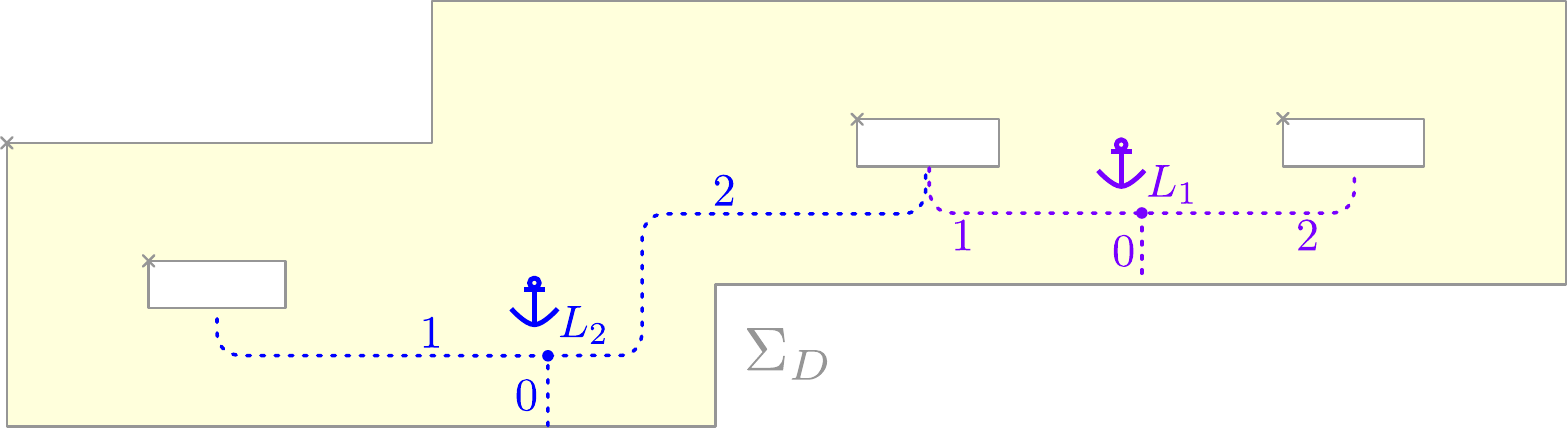} 
    \end{center}
    We pick an anchor $\anchor_D$ for $D$ in the equivalence class up to the outer boundary, uniquely determined by the property that $\anchor_D$ extends the anchor $\anchor_{L_2}$ with offset 0, and is supported on the right strip of $L$. Then the anchor $\anchor_D$ for $D$ also extends the anchor $\anchor_{L_1}$ with offset 1:
    \begin{center}
	   \includegraphics[width = 0.5\textwidth]{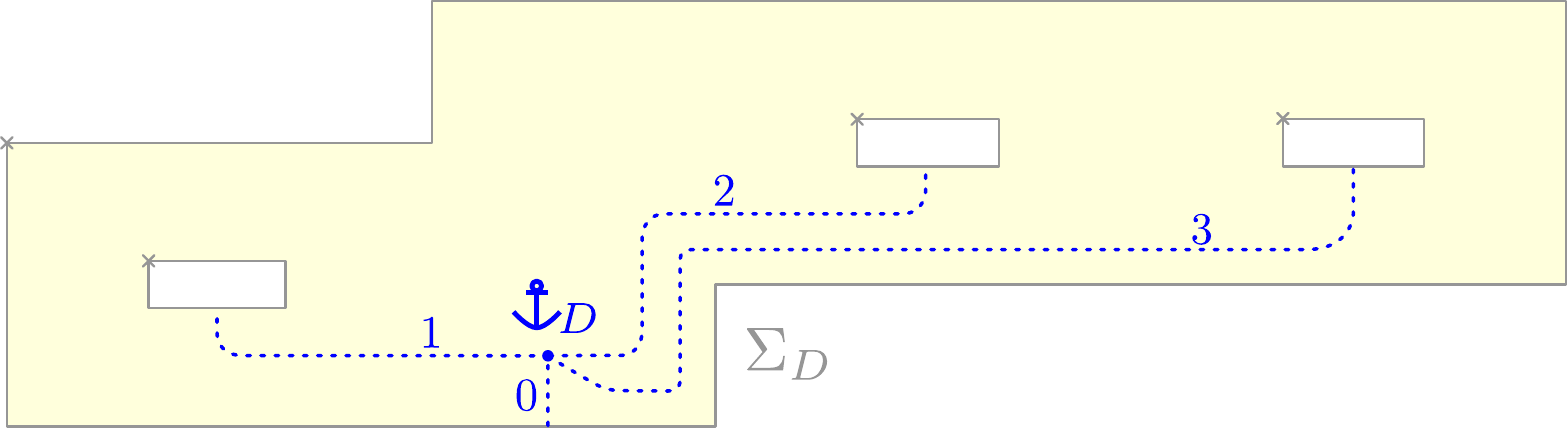} 
    \end{center}
    As we get $L$ from filling the middle puncture of $D$,
    $$B_L = B_D \frt_{\caS^{\anchorino_D}_2}(p^\1) = \frt_{\caS^{\anchorino_D}_2}(p^\1) B_D.$$
    Since $P_L^{\1\1}$ and $P_L^{\bar X \1}$ are both dominated by $B_L$, we can invoke the Inclusion Lemma \ref{lem:inclusion lemma} twice to get
    \begin{equation} \label{eq:concatenation_intermediate_1}
    \Dr_{L_2}^{(\I X \rightarrow X \I)} \Dr_{L_1}^{(\I \I \rightarrow X \bar X)} P_L^{\I \I} 
    = \Dr_D^{(\I X \bar X \rightarrow X \I \bar X)} \Dr_D^{(\I \I \I \rightarrow \I X \bar X)} P_L^{\I \I},
    \end{equation}
    where
    \begin{align*}
    \Dr_D^{(\I \I \I \rightarrow \I X \bar X)} &= d_X^{-1} \Dr_D \left[ \adjincludegraphics[valign=c, height = 1.0cm]{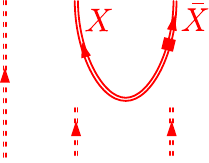} \right], \\
    \Dr_D^{(X \bar X \I \rightarrow X \I \bar X)} &= \Dr_D \left[ \adjincludegraphics[valign=c, height = 1.0cm]{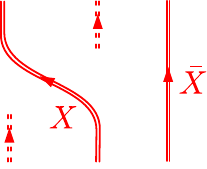} \right].
    \end{align*}
    By the choice of boundary condition $w^\1$ \eqref{eq:ground state boundary condition at puncture}, we have 
    \begin{equation} \label{eq:concatenation_intermediate_2}
    \Dr_L^{( \I \I \rightarrow X \bar X)} P_L^{\I \I} = \Dr_D^{( \I \I \I \rightarrow X \I \bar X)} P_L^{\I \I},
    \end{equation}
    where we define
    $$
    \Dr_D^{( \I \I \I \rightarrow X \I \bar X)} = d_X^{-1} \Dr_D \left[ \adjincludegraphics[valign=c, height = 1.0cm]{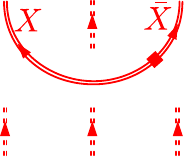} \right].
    $$
    Equation \eqref{eq:first concatenation equality} now follows from Eqs. \eqref{eq:concatenation_intermediate_1} and \eqref{eq:concatenation_intermediate_2}, and the fact that $\Dr_D$ is a representation (\Cref{lem:Drinfeld insertion multiplicativity}).
    With similar definitions, we derive \eqref{eq:second concatenation equality} by
    \begin{align*}  
    \Dr_{L_2}^{(\I X \rightarrow X \I)} \Dr_{L_1}^{(\I X \rightarrow X \I)} P_L^{\I X} 
    &= \Dr_D^{(\I X \I \rightarrow X \I \I)} \Dr_D^{(\I \I X \rightarrow \I X \I)} P_L^{\I X}  \\
    &= \Dr_D^{(\I \I X \rightarrow X \I \I)} P_L^{\I X} \\
    &= \Dr_L^{(\I X \rightarrow X \I)} P_L^{\I X}.
    \end{align*}
\end{proof}

\subsection{String operators} \label{subsec:string operators}

A \emph{chain}\footnote{visualize a bicycle chain} $\scrC = ( L_n )_{n \in \N}$ is a half-infinite sequence of links $L_n$ such that $L^{\scrC}_{\interval{m}{n}} := L_n \wedge L_{n-1} \wedge \cdots \wedge L_m$ is a well-defined link for all natural numbers $m \leq n$. We write $\scrC[\interval{m}{n}]=(L_i)_{i=m}^n$ for the finite subchains of $\scrC$ as well as $\scrC_n := \scrC[\interval{1}{n}]$, and $\partial_{\ii} \scrC = \partial_{\f} L_1$ for the initial puncture of the chain, which is the final puncture of the first link. We also write $\supp(\scrC) = \bigcup_{n \in \N} \supp(L_n)$ for the support of the chain $\scrC$. We say a face $f \in \caF$ belongs to the chain $\scrC$ if it belongs to one of its links.

For any $X \in \Irr \ZC$ and natural numbers $m \leq n$ we define unitaries $$U^X_{\scrC[\interval{m}{n}]} = u^X_{L_n} \times \cdots \times u^X_{L_m}$$ 
and automorphisms $\rho^X_{\scrC_n} := \Ad[ (U^X_{\scrC[\interval{1}{n}]})^* ]$. 

\begin{lemma} \label{lem:limiting endomorphisms}
    Let $\scrC$ be a chain.
	For any $x \in \caA^{\loc}$ there is $n_0$ large enough so that $ \rho_{\scrC_{n}}^X(x) = \rho_{\scrC_{n_0}}^X(x)$ for all $n \geq n_0$. In particular, the limit
	$$ \rho_{\scrC}^X := \lim_{n \uparrow \infty} \, \rho_{{\scrC}_n}^X,$$
	exists and defines a unital *-endomorphism of $\caA$.
    Moreover, $\rho_{\scrC}^X$ is supported on $\supp({\scrC})$ in the sense that if $x \in \caA_{\supp({\scrC})^c}$ then $\rho_{\scrC}^X(x) = x$.
\end{lemma}

\begin{proof}
    Given $x\in \caA^{\loc}$, there is a maximal $n_0$ such that $\supp(u_{L_{n_0}}^X) \cap \supp(x) \ne \emptyset$ (unless $x$ is a multiple of the identity, in which case $n_0=1$). 
    It is clear from the definition of $\rho_{{\scrC}_n}$ that $n_0$ has the desired property.
    The sequence $\rho_{{\scrC}_n}^X(x)$ is eventually constant with limit $ \rho_{\scrC}^X(x):= \rho_{{\scrC}_{n_0}}^X(x)$.
    This yields a well-defined unital *-endomorphism of $\caA^{\loc}$ which extends to the whole of $\caA$ by continuity, and which is evidently supported on $\supp({\scrC})$.
\end{proof}

\section{Irreducible Anyon Representations} \label{sec:anyon representations}

Fix a chain ${\scrC} = \{ L_n \}_{n \in \N}$ with $e = \partial_{\ii} {\scrC}$ and let $X \in \Irr \ZC$. Recall that $(\pi^{\I}, \caH, \Omega)$ is the GNS triple of the unique frustration free ground state $\omega^{\I}$ of the Levin-Wen Hamiltonian.
We define a new representation $$\pi_{\scrC}^X := \pi^{\I} \circ \rho_{\scrC}^X,$$ where $\rho_\scrC^X$ is the endomorphism defined in \Cref{lem:limiting endomorphisms}. Our goal in this section is to show that $\pi_{\scrC}^X$ is an irreducible anyon representation, that $\pi_{\scrC}^X \simeq \pi_{{\scrC}'}^X$ for any two chains ${\scrC}, {\scrC}'$, and that $\pi_{\scrC}^X$ and $\pi_{{\scrC}'}^Y$ are disjoint if $X \neq Y$.

For the remainder of this section we will usually abbreviate notation with $\rho^X = \rho_{\scrC}^X$, $\rho^X_n = \rho_{{\scrC}_n}^X$, $u^X_n = u^X_{L_n}$, $u^X_{\interval{m}{n}} = u_{L_{\interval{m}{n}}^\scrC}^X$, as well as $U^X_{\interval{m}{n}} = U^X_{\scrC[\interval{m}{n}]}$, $U^X_n = U^X_{\scrC[\interval{1}{n}]}$, and $P^{\bullet\bullet}_{\interval{m}{n}} = P^{\bullet\bullet}_{L_{\interval{m}{n}}^\scrC}$ for natural numbers $m \leq n$

\subsection{Pure anyon states belonging to \texorpdfstring{$\pi_{\scrC}^X$}{piCX}} \label{subsec:pure anyon state in the rep}

Recall from Section \ref{subsec:anyons on punctures} that for any $e \in \caE$ and any minimal projector $p \in \Tube_{\caS_e}$ we have a pure state $\omega_e^p$ uniquely characterized by the constraints
\begin{equation} \label{eq:local constraints for anyon at puncture}
    \omega_e^p( \frt_e(p) ) = \omega_e^p(B_f) = 1 \quad \forall f \in F_{C^{(e)}}.
\end{equation}

In Section \ref{subsec:Drinfeld insertions} we fixed for each $X \in \Irr \caC$ a minimal projector $p^X \in \Tube_{\caS_e}$ of type $X$. Let us write $\omega_e^X := \omega_e^{p^{\bar X}}$. We give this state the label $X$ because it is in the range of $\frt_{\partial \Sigma_{\bbD}}(P^X)$ for any sufficiently large region $\bbD$ homeomorphic to a disk with a puncture at $e$. This is  the correct notion of charge for the state $\omega_e^X$ from the point of view of sector theory.

\begin{lemma} \label{lem:pure anyon state produced by string operator}
    We have
    $$ \omega_e^X = \omega^{\I} \circ \rho^X. $$
\end{lemma}

\begin{proof}
    It is sufficient to show that the state $\omega^\1 \circ \rho^{X}$ satisfies the constraints of Eq. \eqref{eq:local constraints for anyon at puncture} with $p = p^{\bar X}$.
    Applying the Concatenation Lemma \ref{lem:concatenation lemma} inductively shows that for all $m<n \in \N$,
    \begin{equation}
        U^X_{\interval{m}{n}} P_{\interval{m}{n}}^{\1\1}
        = u^X_n \cdots u^X_m P_{\interval{m}{n}}^{\1\1}
        = u^X_{\interval{m}{n}} P_{\interval{m}{n}}^{\1\1}.
    \end{equation}
    Since $\omega^\1$ satisfies all ground state constraints, we have that
    \begin{align*}
        \omega^\1(\rho^X_n(x)) &= 
        \omega^\1(P_{\interval{1}{n}}^{\1\1} \, \rho_n^X(x) \, P_{\interval{1}{n}}^{\1\1}) \\
        &= \omega^\1((U^X_n P_{\interval{1}{n}}^{\1\1})^* x U^X_n P_{\interval{1}{n}}^{\1\1}) 
        = \omega^\1((u^X_{\interval{1}{n}})^* \, x \, u^X_{\interval{1}{n}} )
    \end{align*}
    for all $x\in \caA$.
    Using the intertwining property \eqref{eq:intertwining properties of u_L}, we have 
    $$
    \frt_e(p^{\bar X}) u^X_{\interval{1}{n}} P_{\interval{1}{n}}^{\1\1}
    = \frt_e(p^{\bar X}) P_{\interval{1}{n}}^{X \bar X} u^X_{\interval{1}{n}} 
    = u^X_{\interval{1}{n}} P_{\interval{1}{n}}^{\1\1}.
    $$
    It follows that $\omega^{\I} \big( \rho^X_n( \frt_e(p^{\bar X}) ) \big)  =1$ for all $n$, so $\omega^{\I} \big( \rho^X( \frt_e(p^{\bar X}) ) \big)  =1$.

    Let $f$ be a face which is not adjacent to $e$. Then by \Cref{lem:string insertions preserve ground state constraints}, $B_f$ commutes with $u^X_{\interval{1}{n}}$ unless $f$ is adjacent to $\partial_\ii L^\scrC_{\interval{1}{n}}$. So there is an $N$ such that $B_f$ commutes with $u^X_{\interval{1}{n}}$ for all $n\ge N$, and again we find that $\omega^{\I} \big( \rho^X( B_f ) \big)  =1$.
\end{proof}

Let us denote by $(\pi_e^X, \caH_e^X, \Omega_e^X)$ the GNS triple of $\omega_e^X$. Since $\omega_e^X$ is a pure state (Proposition \ref{prop:unique anyon states}), the representation $\pi_e^X$ is irreducible.

\subsection{Unitary equivalence of \texorpdfstring{$\pi_{\scrC}^X$}{piCX} and \texorpdfstring{$\pi_e^X$}{pieX} } \label{subsec:equivalence of reps}

Denote by $(\pi_e^X, \caH_e^X, \Omega_e^X)$ the GNS triple of the pure state $\omega_e^X$. 
Recall that $\partial_{\ii} \scrC = e$.
In this section we prove the following proposition.
\begin{proposition} \label{prop:equivalence of string and GNS reps}
    The representation $\pi_{\scrC}^X = \pi^{\I} \circ \rho^X_{\scrC}$ is irreducible, and there is a unitary equivalence,
    $$ \pi_{\scrC}^X \simeq \pi_e^X.$$
\end{proposition}

The strategy for showing that $\pi_{\scrC}^X$ is irreducible is to show that all of its vector states are pure. In fact, it is sufficient to show this for vector states corresponding to a dense subset of the unit ball of $\caH$.

Let us consider a unit vector in $\caH$ of the form $| \Psi \rangle = \pi^{\I}(O) | \Omega \rangle$ for some $O \in \caA^{\loc}$. This defines a vector state $\psi$ in $\pi^X_\scrC$ by
$$\psi(x) := \langle \Psi, \pi_{\scrC}^X(x) \Psi \rangle = \langle \Omega, \pi^{\I} \big(  O^* \rho^X(x) O \big) \, \Omega \rangle$$
for all $x \in \caA$. Our goal is to show that $\psi$ is also a vector state of the GNS representation $\pi_e^X$.

First note that the state $\psi$ can be obtained as the limit of a sequence of vector states of the vacuum representation $\pi^{\I}$ as follows. For $n \in \N$ we let
$$  |\Psi_n \rangle := \pi^{\I}(U^X_n) | \Psi \rangle. $$
These are unit vectors in $\caH$ which correspond under the vacuum representation to pure states $\psi_n$ defined by
$$ \psi_n(x) = \langle \Psi_n, \, \pi^{\I}(x) \, \Psi_n \rangle, \quad x \in \caA.$$

\begin{lemma} \label{lem:psi_n converge to psi}
    In the $w^*$-topology, $\psi_n \rightarrow \psi$.
\end{lemma}

\begin{proof}
    Take $x \in \caA^{\loc}$. By Lemma \ref{lem:limiting endomorphisms} there is $n_0$ large enough so that
    $$ \pi_{\scrC}^X(x) = \pi^{\I}( \rho^X(x) ) = \pi^{\I}( \rho^X_n(x) ) = \pi^{\I} \big(  (U^X_n)^* \, x \, U^X_n \big) $$
    for all $n \geq n_0$. It follows that
    \begin{align*}
        \psi(x) &= \langle \Psi, \pi_{\scrC}^X(x) \, \Psi \rangle = \langle \Psi, \pi^{\I} \big( (U^X_n)^* x U^X_n \big) \, \Psi \rangle \\
        &= \langle \Psi_n, \pi^{\I}(x) \, \Psi_n \rangle = \psi_n(x)
    \end{align*}
    for all $n \geq n_0$. Since the $\psi_n$ are states, and therefore all of norm one, the claim now follows by density.
\end{proof}

In a similar way we can approximate the anyon state $\omega_e^X$ by a sequence of vector states of the vacuum representation as follows. For each $n \in \N$ we let
$$ | \Omega_n \rangle = \pi^{\I}(U^X_n) | \Omega \rangle $$
and obtain pure states $\omega_n$ defined by
$$ \omega_n(x) := \langle \Omega_n, \pi^{\I}(x) \, \Omega_n \rangle = \omega^{\I}( (U^X_n)^* x U^X_n ), \quad x \in \caA. $$

\begin{lemma} \label{lem:omega_n converge to omega}
     In the $w^*$-topology, $\omega_n \rightarrow \omega_e^X$.
\end{lemma}

\begin{proof}
    Using \Cref{lem:pure anyon state produced by string operator}, the proof is identical to that of \Cref{lem:psi_n converge to psi}.
\end{proof}

Before continuing towards the proof of Proposition \ref{prop:equivalence of string and GNS reps}, let us introduce some regions that will be used often below.

Recall that for any $R > 0$ we denote by $B_R = \{ x \in \R^2 \, : \, \norm{x}_{\infty} \leq R \}$ the closed box of side length $2R$ centered on the origin. We define the region $\bbD_R$ to consist of all faces, edges, and vertices of $C^{\Z^2}$ that are contained in $B_R$. The associated surface $\Sigma_{\bbD_R}$ is homeomorphic to a disk. For any $0 < R_1 < R_2 - 1$ we define the annular region $\ann_{R_1, R_2}$ to consist of all faces, edges, and vertices contained in the closure of $B_{R_2} \setminus B_{R_1}$. For any edge $e \in \caE$ we moreover define regions $\bbD_R^{(e)}$ obtained from $\bbD_R$ by removing $e$ and its two neighbouring faces, and similarly $\ann_{R_1, R_2}^{(e)}$ obtained from $\ann_{R_1, R_2}$ by removing $e$ and its two neighbouring faces.
If $e$ belongs to $\bbD_R$ at a suitable distance from the boundary then $\Sigma_{\bbD_R^{(e)}}$ is homeomorphic to an annulus and we simply say $\bbD_R^{(e)}$ is an annulus.

Now we argue that for $n$ sufficiently large, the states $\omega_n$ and $\psi_n$ agree on all observables supported outside of some box.
To this end, let us introduce the following state spaces:

\begin{definition} \label{def:state spaces for unitary equivalence}
    For any $R > 0$, any $e \in \caE$, and any $X \in \Irr \ZC$ we let $\caS_{> R}^{(e, X)}$ be the space of states $\psi : \caA \rightarrow \C$ for which
    $$ \psi \big(  \frt_e(p^{X}) \big) = \psi \big(  B_f \big) = 1 $$
    for all faces $f\subset B_R^c$ such that $e$ is not adjacent to $f$.
\end{definition}

That is, $\caS_{> R}^{(e, X)}$ consists of states satisfying ground state constraints outside of the box of size $R$, except at the puncture at $e$ where instead the projector $\frt_e(p^{X})$ is satisfied.

Recall that $\scrC = \{ L_n \}_{n \in \N}$ and put $e_n = \partial_{\ii} L_n = \partial_{\f} L_{n+1}$ for all $n \in \N$.

\begin{lemma} \label{lem:omega_n and psi_n agree outside a ball part 1}
    There is $R > 0$ and $n_1 \in \N$ such that $\omega_n$ and $\psi_n$ both belong to $\caS_{>R}^{(e_n, X)}$ for all $n \geq n_1$.
\end{lemma}

\begin{proof}
    Recall that a link is a sequence of faces. We say a link is contained in $S \subset \R^2$ if all of its faces are subsets of $S$.
    Let $R_O>0$ be such that $\supp O \subset B_{R_O}$. 
    Let $n_0$ be maximal such that $L_{n_0}\cap B_{R_O} \ne \emptyset$.
    Then $\partial_\ii L_{n_0}$ necessarily lies outside $B_{R_O}$. 
    We choose $R> R_O$ such that $L_{n_0}$ is contained in $B_R$ and set $n_1 = n_0 + 1$.
    The point is now that $|\Psi\rangle$ satisfies ground state constraints along the support of any subchain starting from $n_1$, i.e. 
    $ \pi^\1(P_{\interval{n_1}{n}}^{\I \I}) | \Psi \rangle = | \Psi \rangle$.
    The support of the links $L_k$ where $k < n_0 = n_1 - 1$ is disjoint from this subchain, so for any $n \geq n_1$ we have
    $$ \pi^{\I}(P_{\interval{n_1}{n}}^{\I \I}) | \Psi_{n_0 - 1} \rangle 
    = \pi^{\I}(P_{\interval{n_1}{n}}^{\I \I}U^X_{n_0-1}) | \Psi \rangle = 
    | \Psi_{n_0 - 1} \rangle.$$
    The link $L_{n_0}$ overlaps with $L_{n_1}$ at the puncture at $\partial_{\f} L_{n_1}$ but as expressed by \eqref{eq:overlapping link}, $u^X_{n_0}$ will either preserve the ground state constraint or produce an $X$ anyon at the puncture when acting on $| \Psi_{n_0 - 1} \rangle$. That is,
    $$ | \Psi_{n_0} \rangle 
    = \pi^{\I}(u^X_{n_0}) | \Psi_{n_0 - 1} \rangle
    = \pi^{\I}(P_{\interval{n_1}{n}}^{\I ?}) | \Psi_{n_0} \rangle.
    $$
    Now it follows, again by inductively applying the Concatenation \Cref{lem:concatenation lemma}, that 
    $$ | \Psi_{n} \rangle 
    = \pi^{\I}\Big(u^X_{\interval{n_1}{n}} P_{\interval{n_1}{n}}^{\I ?} \Big) | \Psi_{n_0} \rangle
    = \pi^{\I}(P_{\interval{n_1}{n}}^{X ?}) | \Psi_{n} \rangle,$$
    where the last equality uses the intertwining property \eqref{eq:combined intertwining}.
    This implies that $\psi_n(\frt_{e_{n}}( p^{X})) = 1$ and $\psi_n(B_f) = 1$ for all bulk faces $f$ of $L_{\interval{n_1}{n}}^\scrC$.
    For any other face $f\subset B_{R}^c$ which is not adjacent to $e_n$, the projector $B_f$ commutes with $O$ and with $U^X_n$ (by \Cref{lem:string insertions preserve ground state constraints}) so $\psi_n(B_f)=1$. (Note in particular that the faces adjacent to $e_{n_0} = \partial_\ii L_{n_0}$ belong to $L_{n_0}$ which is contained in $B_R$.)
    Thus $\psi_n \in \caS_{>R}^{(e_n, X)}$. Taking $O=1$ shows $\omega_n \in \caS_{>R}^{(e_n, X)}$ as well.
\end{proof}

\begin{lemma} \label{lem:omega_n and psi_n agree outside a ball part 2}
    There is $R \geq 0$ and a number $n_2 \in \N$ such that $\omega_n|_{B_R^c} = \psi_n|_{B_R^c}$ for all $n \geq n_2$.
\end{lemma}

\begin{proof}
    By Lemma \ref{lem:omega_n and psi_n agree outside a ball part 1} there is $R>0$ and $n_1 \in \N$ such that $\omega_n$ and $\psi_n$ both belong to $\caS_{>R}^{(e_n, X)}$ for all $n \geq n_1$.

    Let $\caS$ be the outer boundary of an arbitrary finite connected region $C$. Then $\caS$ is the unique connected boundary component of the disk-like region $D$ obtained from $C$ by filling in its inner holes, and the restriction of $\omega^{\I}$ to $\caA_D$ corresponds to a density matrix in $\caD_{D} = \caD_{D}(P^{\I})$, where we used \Cref{lem:characterization of some caDs}. This shows that $\omega^{\I}( \frt_{\caS}(P^{\I}) ) = 1$ for any such $\caS$.
    
    For any $n\ge n_1$, take $R'>R$ large enough so that the support of $\scrC[1,n]$ is contained in $B_{R'}$ (recall that $\supp O \subset B_{R}$). Then the region $\ann_{R, R'}^{(e_n)}$ has associated surface homeomorphic to a thrice punctured sphere. (two connected boundary components from the original annulus and one from the puncture at $e_n$.) Moreover, denoting by $\caS_{R'}$ the outer boundary component of $\Sigma_{\ann_{R, R'}}$, we have that $\omega_n( \frt_{\caS_{R'}}(a) ) = \psi_n( \frt_{\caS_{R'}}(a) ) = \omega^{\I}( \frt_{\caS_{R'}}(a) )$ for any $a \in \Tube_{\caS_{R'}}$. Since $\omega^{\I}( \frt_{\caS_{R'}}(P^{\I}) ) = 1$, it follows that the restrictions of $\omega_n$ and $\psi_n$ to $\ann_{R, R'}^{(e_n)}$ are given by density matrices in
    $$\caD_{\ann_{R, R'}^{(e_n)}}(P^{\I}, \id,  p^{X} ) = \caD_{\ann_{R, R'}^{(e_n)}}(P^{\I}, P^{\bar X}, p^{X} ),$$
    where the equality follows from Lemma \ref{lem:characterization of some caDs}. Here, the first slot gives the boundary condition on the outer boundary of the annulus, the second slot gives the boundary condition on the inner boundary of the annulus, and the final slot gives the boundary condition at the puncture at $e_n$.

    Choosing $n_2\ge n_1$ such that whenever $n\ge n_2$ and $R'>R+1$ is large enough, then
    $\ann_{R+1, R'-1}^{e_n}$ is also a thrice punctured sphere, we find by \Cref{lem:restriction yields maximally mixed boundary conditions} that the restrictions of $\psi_n$, $\omega_n$ to $\ann_{R+1, R'-1}^{e_n}$ correspond to density matrices in $\caD_{\ann_{R, R'}^{(e_n)}}(\star^{\I}, \star^{\bar X}, p^{X} )$.
    But this is a singleton according to \Cref{lem:characterization of some caDs}, so 
    $$\psi_n|_{B_{R'-1}\setminus B_{R+1}} = \omega_n|_{B_{R'-1} \setminus B_{R+1}}.$$
    As $R'$ can be chosen arbitrarily large this proves the claim by density.
\end{proof}

\begin{lemma} \label{lem:unitary equivalence of omega and psi}
    There is a unitary $V \in \caA^{\loc}$ such that $\psi = \omega_e^X \circ \Ad[V]$.
\end{lemma}

\begin{proof}
    Let $R$ and $n_2$ be as in the statement of Lemma \ref{lem:omega_n and psi_n agree outside a ball part 2}. 
    Pick $n_3 \ge n_2$ such that $L_n$ is disjoint from $B_R$ for all $n\ge n_3$.
    The vacuum Hilbert space has tensor product structure $\caH = \caH_{B_R} \otimes \caH_{B_R^c}$ (because $\caH_{B_R}$ is finite dimensional), and Lemma \ref{lem:omega_n and psi_n agree outside a ball part 2} implies that the unit vectors $| \Omega_{n_3} \rangle$ and $| \Psi_{n_3} \rangle$ have the same expectation values for any observable in $\caB(\caH_{B_R^c})$. That is, the vectors $| \Omega_{n_3} \rangle$ and $| \Psi_{n_3} \rangle$ are both purifications of the same state on $\caA_{B_R^c}$, with finite dimensional purification space $\caH_{B_R}$. It follows that there exists a unitary $W \in \caB( \caH_{B_R})$ such that $W | \Psi_{n_3} \rangle = | \Omega_{n_3} \rangle$.

    Let $n > n_3$. Using that $W$ has support disjoint from $U^X_{\interval{n_3}{n}}$ we find
    \begin{align*}
        \ket{ \Psi_n} 
        = \pi^\1(U^X_{\interval{n_3+1}{n}}) \ket{\Psi_{n_3}} 
        = \pi^\1(U^X_{\interval{n_3+1}{n}}) W \ket{\Omega_{n_3}} 
        = W \pi^\1(U^X_{\interval{n_3+1}{n}}) \ket{\Omega_{n_3}}
        = W \ket{ \Omega_n}.
    \end{align*}
    Since $B_R \cap \Z^2$ is finite, there is a unitary $V \in \caA_{B_R}$ such that $W = \pi^{\I}(V)$. We then have $\psi_n = \omega_n \circ \Ad[V]$ for all $n > n_3$.
    Together with Lemmas \ref{lem:psi_n converge to psi} and \ref{lem:omega_n converge to omega}  this implies that $\psi = \omega_e^X \circ \Ad[V],$ as required.
\end{proof}

\begin{proofof}[Proposition \ref{prop:equivalence of string and GNS reps}]
    It follows from the previous Lemma that any state $\psi$, corresponding under the representation $\pi_{\scrC}^X$ to a unit vector in $\caH$ of the form $| \Psi \rangle = \pi^{\I}(O) | \Omega \rangle$ for some $O \in \caA^{\loc}$, is a vector state for the GNS representation $\pi_{e}^X$. Indeed, a representing vector is given by $\pi^X_e(V)\ket{\Omega^X_e}$ where $|\Omega^X_e\rangle$ denotes the cyclic vector of the GNS representation corresponding to the state $\omega^X_e$, and $V\in \caA$ is the unitary granted by \Cref{lem:unitary equivalence of omega and psi}. In particular, any such $\psi$ is a pure state.
    
    Since $\pi^\1(\caA^{\loc})\ket{\Omega}$ is dense in $\caH$, it follows that $\pi^X_{\scrC}$ is irreducible.
    By Lemma \ref{lem:pure anyon state produced by string operator} the state $\omega_e^X$ is a vector state for both irreducible representations $\pi_{\scrC}^X$ and $\pi_e^X$, so we conclude that $\pi_{\scrC}^X \simeq_{u.e.} \pi_e^X$.
\end{proofof}

\subsection{Disjointness} \label{subsec:disjointness of reps}

We no longer keep $\scrC$ and $X$ fixed.
Recall that $\pi_e^X$ is the GNS representation of the pure state $\omega_e^X$.

\begin{proposition} \label{prop:disjointness}
	Let $X, Y \in \Irr \ZC$ and $e, e' \in \caE$. Then $\pi_e^X \simeq_{u.e.} \pi_{e'}^Y$ if and only if $X = Y$.
\end{proposition}

\begin{proof}
    The statement follows from showing that $\pi_e^X \simeq_{u.e.} \pi_{e'}^X$, and that $X\ne Y$ implies that $\pi_e^X$ and $\pi_e^Y$ are disjoint.
    Assume $X\ne Y$.
    Let $R>0$ be such that $\Sigma_{\bbD_R^{(e)}}$ is homeomorphic to an annulus.
    As noted in the proof of \Cref{prop:unique anyon states}, the restriction $\omega_e^X|_{\caA_{B_R}}$ corresponds to the unique density matrix in $\caD_{\bbD_{R}^{(e)}}(p^{\bar X}, \star^{X})$.
    In particular $\omega_e^X(\frt_{R}(P^{X}))=1$, writing $\frt_{R}$ for the Tube action on the outer boundary of $\bbD^{(e)}_{R}$.
    Similarly, $\omega_e^Y(\frt_{R}(P^{Y}))=1$.
    By orthogonality, $P^{X}P^{Y}=0$ so
    $$  \abs{ \omega_{e}^X \big(  \frt_{R}(P^{X}) \big) - \omega_{e}^{Y} \big(  \frt_{R}( P^{X} )  \big)  } = 1.  $$
    Since $R$ can be taken arbitrarily large, it  follows from Corollary 2.6.11 of \cite{bratteli2012operator} that the GNS representations $\pi_e^X$ and $\pi_e^Y$ are disjoint.
    
    Now we show that $\pi_e^X \simeq_{u.e.} \pi_{e'}^X$.
    Assume that there is a link $L$ with $\partial_\ii L= e$ and $\partial_\f L=e'$. Then $\omega^X_{e'} = \omega^X_e \circ \Ad[(u_L^X)^*]$, so the GNS representations of these pure states are unitarily equivalent. In case there is no such link we obtain the result by considering an intermediate equivalence with $\pi_{e''}^X$ for a suitable edge $e''$ such that there are links $L_1$ and $L_2$ with $e = \partial_{\ii} L_1, e'' = \partial_{\f} L_1 = \partial_{\ii} L_2$ and $e' = \partial_{\f} L_2$. Such an intermediate $e''$ can always be found.
\end{proof}

\subsection{Superselection criterion} \label{subsec:superselection criterion}
\begin{proposition} \label{prop:superselection criterion for piX}
	The representations $\pi_e^X$ for all $e\in \caE$ and $X\in \Irr \ZC$ satisfy the superselection criterion.
\end{proposition}

\begin{proof}
	We must show that for any cone $\Lambda$ there is a unitary equivalence
	$$ \pi_e^X|_{\Lambda} \simeq_{u. e} \pi^{\I}|_{\Lambda}. $$

	Let ${\scrC}$ be a chain supported in the complement of $\Lambda$, so $\rho_{\scrC}^X(x) = x$ for all $x \in \caA_{\Lambda}$. Then we have
	$$ \pi_{\scrC}^X|_{\Lambda} = (\pi^{\I} \circ \rho_{\scrC}^X)|_{\Lambda} = \pi^{\I}|_{\Lambda}. $$
    Let $e' = \partial_{\ii} {\scrC}$. The unitary equivalence $\pi_{\scrC}^X \simeq_{u.e.} \pi_{e'}^X$ from Proposition \ref{prop:equivalence of string and GNS reps} together with $\pi_e^X \simeq_{u.e.} \pi_{e'}^X$ from Proposition \ref{prop:disjointness} gives the desired equivalence.
\end{proof}

\section{Completeness} \label{sec:completeness}

We show that every irreducible anyon representation $\pi$ is isomorphic to one of the anyon representations $\pi_e^X$ constructed in Section \ref{sec:anyon representations}.

\subsection{Anyon representations contain locally excited vector states}

Let $\pi : \caA \rightarrow \caB(\caH)$ be an irreducible anyon representation.

For any $S \subset \R^2$, let $S^F \subset \caF$ be the set of faces such that $B_f \in \caA_S$. For bounded $S \subset \R^2$, define
$$ P_S := \prod_{f \in S^F} \, B_f. $$

We now construct analogous projectors for infinite regions $S$ in the von Neumann algebra $\pi( \caA )'' = \caB(\caH)$.

We say a non-decreasing sequence $\{ S_n \}$ of subsets of $\R^2$ converges to $S \subset \R^2$ if $\bigcup_{n} S_n = S$. Let $\{S_n\}$ be such a non-decreasing sequence of bounded subsets of $\R^2$ which converges to a possibly unbounded $S \subset \R^2$. Then the sequence of projectors $\{ \pi( P_{S_n} ) \}$ is non-increasing and converges in the strong operator topology to the orthogonal projector $p_S$ onto the intersection of the ranges of the $\pi( P_{S_n} )$, see for example \cite[Thm 4.32(a)]{Weidmann}. In particular, $p_S$ is independent of the particular sequence $\{S_n\}$, and if $S$ is bounded then $p_S = \pi(P_S)$.

The following proposition is an adaptation of \cite[Proposition 5.2]{bols2025classification} to Levin-Wen models.
\begin{proposition} \label{prop:anyon rep contains finitely excited vector state}
	Any anyon representation $\pi : \caA \rightarrow \caH$ has a vector state $\psi$ for which there is $R \geq 0$ such that
	$$  \psi(B_f) = 1  $$
	for all $f \in \caF$ such that $B_f \in \caA_{B_{R}^c}$.
\end{proposition}

\begin{proof}
	Let $\Lambda_1, \Lambda_2$ be two cones such that $\Lambda_1 \cup \Lambda_2 = \R^2$, and such that any $B_f$ belongs to $\caA_{\Lambda_1}$ or $\caA_{\Lambda_2}$ (or both).

	Since $\pi$ satisfies the superselection criterion, there are unitaries $U_i : \caH \rightarrow \caH$ for $i = 1, 2$ such that
	$$ \pi|_{\Lambda_i} = (\Ad[U_i] \circ \pi^{\I})|_{\Lambda_i}. $$
	It follows that
	$$ \omega^{\I}(x) = \langle \Omega^{\I}, \pi^{\I}(x) \, \Omega^{\I} \rangle = \langle \Omega^{\I}, U_i^* \pi(x) U_i \, \Omega^{\I} \rangle    $$
	for all $x \in \caA_{\Lambda_i}$. Define states $\omega_i$ by $\omega_i(x) := \langle \Omega_i, \pi(x) \Omega_i \rangle$ for $x \in \caA$, where $\ket{\Omega_i} = U_i\ket{\Omega^{\I}} \in \caH$. The states $\omega_i$ are vector states of $\pi$ and satisfy $\omega_i(x) = \omega^{\I}(x)$ for all $x \in \caA_{\Lambda_i}$.

	Let $\Lambda_i^{> n} := \Lambda_i \setminus B_{n}$ and $\Lambda_i^{n, n+m} := \Lambda_i^{>n} \setminus \Lambda_i^{>n+m}$ for all $m, n \in \N$. Then the sequence $m \mapsto \Lambda_i^{n, n+m}$ is a non-decreasing sequence of bounded sets converging to $\Lambda_i^{> n}$. We have
	$$1 = \omega^{\I}\big( P_{\Lambda_i^{n, n+m}} \big) = \omega_i \big( P_{\Lambda_i^{n, n+m}} \big) = \langle \Omega_i, \pi \big( P_{\Lambda_i^{n, n+m}} \big)  \Omega_i \rangle,$$
	where we used that all these projectors are supported in $\Lambda_i$. It follows that
	$$\langle \Omega_i, p_{\Lambda_i^{> n}} \, \Omega_i \rangle = 1. $$
    From Corollary 2.6.11 of \cite{bratteli2012operator} we obtain an $N \in \mathbb{N}$ such that
	$$|\omega_1(O) - \omega_2 (O)| \leq \frac{1}{2} \norm{O}$$
	for all $O \in \caA^{\loc} \cap \caA_{B_N^c}$. This implies that
    $ \omega_1(P_{\Lambda_2^{N, n}}) > 1/2$
    for all $n > N$, so by continuity $\langle \Omega_1, p_{\Lambda_2^{>N}} \Omega_1 \rangle \ge 1/2$. 
    Since $B_{N}^c \subset \Lambda_1^{>N} \cup \Lambda_2^{>N}$, we have $ p_{B_{N}^c } \ge p_{\Lambda_2^{>N}} p_{\Lambda_1^{>N}}$ so
    $$\langle \Omega_1, p_{B_{N}^c } \Omega_1 \rangle \ge \langle \Omega_1,p_{\Lambda_2^{>N}} p_{\Lambda_1^{>N}} \Omega_1 \rangle \ge \frac{1}{2}. $$
    It follows that $p_{B_{N}^c} \ket{\Omega_1} \neq 0$, so we can define a normalized vector
	$$\ket{\Psi} := \frac{p_{B_{N}^c} \ket{\Omega_1}}{||p_{B_{N}^c} \ket{\Omega_1}||} \in \caH.$$ 
    The corresponding vector state of $\pi$ is defined by $\psi(x):=\langle \Psi, \pi(x) \, \Psi \rangle$ for all $x \in \caA$.
    To finish the proof, we verify that $\psi(B_f) = 1$ whenever $B_f \in \caA_{B_{N}^c}$. We have
	$$ \psi(B_f) = \frac{\langle \Omega_1, p_{B_{N}^c} \pi(B_f) p_{B_{N}^c} \Omega_1 \rangle }{ \norm{  p_{B_{N}^c} \ket{\Omega_1} }^2}
	= \frac{\langle \Omega_1, p_{B_{N}^c} \, \Omega_1 \rangle }{  \norm{p_{B_{N}^c} \ket{\Omega_1} }^2} = 1.
	$$
	This concludes the proof.
\end{proof}

\subsection{Proof of completeness} \label{subsec:proof of completeness}

\begin{proposition} \label{prop:completeness of anyon sectors}
    Let $\pi : \caA \rightarrow \caB(\caH)$ be an irreducible anyon representation. Then there is a unique $X \in \Irr \ZC$ so that $\pi \simeq_{u.e.} \pi_e^X$ for any $e \in \caE$.
\end{proposition}

\begin{proof}
    By Proposition \ref{prop:anyon rep contains finitely excited vector state}, $\pi$ has a vector state $\psi$ for which there is $R \geq 1$ such that $\psi(B_{f}) = 1$ for all faces $f$ whose vertices belong to $B_{R}^c$. Consider the annular regions $\ann_{R, R'}$ for $R' \geq R + 1$ and denote by $\frt_{R}$ the $\Tube$ actions on the inner boundary of annuli $\Sigma_{\ann_{R, R'}}$.
    Note that $\frt_R(1) = \sum_{X} \frt_R(P^X)$ is the projector which enforces string-net constraints along the boundary of $B_{R}^c$. 
    This constraint is also enforced by $\prod_{f\in I_R}B_f$ where $I_R \subset \caF$ is the finite set of faces on the boundary of $B_{R}^c$, so
    $$
    \sum_{X}  \psi(\frt_R(P^X)) =\psi\left( \prod_{f\in I_R} B_f \right) = 1.
    $$
    Since each of the terms of this sum is positive, there is some $X\in \Irr \ZC$ such that
    $$ \psi_X(\bullet) = \psi(\frt_R(P^{\bar X}))^{-1}\psi(\frt_R(P^{\bar X})  \bullet \frt_R(P^{\bar X}) ),$$
    is a well-defined state.
    It is clear that $\psi_X$ is also a vector state of $\pi$.
    It satisfies
    $$ \psi_X(B_f) = \psi_X \big( \frt_R(P^{\bar X}) \big) = 1 $$
    for all faces $f$ whose vertices belong to $B_R^c$. It follows that the restriction of $\psi_X$ to any annular region $\ann_{R, R'}$ corresponds to a density matrix in $\caD_{\ann_{R, R'}}( P^{\bar X}, P^{X} )$. It then follows from Lemma \ref{lem:restriction yields maximally mixed boundary conditions} that the restriction of $\psi_X$ to any annular region $\ann_{R+1, R'}$ for $R' \geq R+2$ corresponds to a density matrix in $\caD_{\ann_{R+1, R'}}( \star^{\bar X}, \star^{X})$. Since $R'$ can be chosen arbitrarily large, it follows that $\psi_X|_{B_{R+1}^c}$ belongs to $\caS^{\star^{\bar X}}_{C_{>R+1}}$ (recall Definition \ref{def:state spaces}), where $C_{>R+1}$ is the infinite region consisting of all faces, edges, and vertices belonging to the closure of $B_{R+1}^c$.

    Let $e \in \caE$ be given. By choosing $R>0$ above large enough we may without loss of generality assume that $\bbD^{(e)}_R$ is an annulus. Then the definition of $\omega_e^X = \omega_e^{p^{\bar X}}$ as the unique state in $\caS_{C^{(e)}}^{p^{\bar X}}$ together with \Cref{lem:restrictions of infinite volume states} implies that $\omega_e^X|_{\caA_{B_{R+1}^c}} \in \caS_{C_{>R+1}}^{\star^{\bar X}}$.
    By Proposition \ref{prop:unique anyon states}, the state space $\caS_{C_{>R+1}}^{\star^{\bar X}}$ is a singleton, so $\psi_X$ and $\omega_e^X$ have identical restrictions to $B_{R+1}^c$. Since $\pi$ and $\pi_e^X$ are irreducible it now follows from Corollary 2.6.11 of \cite{bratteli2012operator} that $\pi \simeq_{u.e.} \pi_e^X$.
    Uniqueness of $X$ follows from disjointness, \Cref{prop:disjointness}.
\end{proof}

\subsection{Proof of the classification Theorem} \label{subsec:proof of main theorem}

We are now ready to prove the main result of this paper.

\begin{proofof}[Theorem \ref{thm:classification of anyon sectors}]
	Fix an edge $e \in \caE$. By Proposition \ref{prop:superselection criterion for piX} the representations $\{ \pi_e^X \}_{X \in \Irr \ZC}$ satisfy the superselection criterion. Since these representations are the GNS representations of pure states $\omega_e^X$, they are irreducible anyon representations. Moreover, by Proposition \ref{prop:disjointness} all these representations are disjoint from each other. We therefore have for each $X \in \Irr \ZC$ a distinct irreducible anyon sector $[\pi_e^X]$. Finally, by \Cref{prop:completeness of anyon sectors} any irreducible anyon representation belongs to one of these sectors.
\end{proofof}

\appendix

\section{Proof of Proposition \ref{prop:characterisation of skein modules}} \label{app:proof of characterization of skein modules}

\subsection{Orthonormal basis for skein modules on disks and annuli} \label{subsec:orthonormal bases for some skein modules}

In the course of the proof of Proposition \ref{prop:characterisation of skein modules}, as well as in Appendix \ref{app:proof of skein subspace isomorphism}, we will often make use of the gluing law Eq. \eqref{eq:gluing formula}. In preparation, let us here give some explicit bases of some skein modules on disks and annuli.
\begin{enumerate}
    \item \underline{The disk $D_1$ with one marked boundary point} : Suppose the boundary point is labelled by $a \in \Irr \caC$. If $a \neq \I$ then any string diagram for this skein module evaluates to zero, \ie $A(D_1 ; a) = \{ 0 \}$ if $a \neq \I$. If $a = \I$ then any string diagram evaluates to a complex number and we have $A(D_1 ; \I) \simeq \C$. A unit vector is represented by the empty string diagram on the disk. \label{enum:simple_skein_bases_1}

    \item \label{skein basis 2} \underline{The disk $D_2$ with one incoming and one outgoing boundary point} : Suppose the boundary points are labelled by $a, b \in \Irr \caC$. If $a \neq b$ then any string diagram for this skein module evaluates to zero and we have $A(D_2, a \cup b) = \{0\}$. If $a = b$ then any string diagram can be interpreted as a morphism from $a$ to $a$, so $A(D_2 ; a \cup a ) \simeq \C$ and is spanned by the unit vector
    \begin{equation} \label{eq:ONB of A(D_2;aa)}
        \frac{1}{d_a^{1/2}} \, \adjincludegraphics[valign=c, width = 1.2cm]{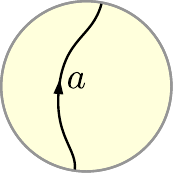}.
    \end{equation} 
    \item \underline{The annulus $\ann$ with no marked boundary points} : An orthonormal basis is given by
    \begin{equation} \label{eq:ONB of ann}
        \left\lbrace \adjincludegraphics[valign=c, width = 1.6cm]{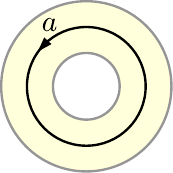} \right\rbrace_{a \in \Irr \caC}.
    \end{equation}
    Indeed, any string diagram on $\ann$ can by isotopy be put in the form of the left hand side of
    \begin{equation*}
        \adjincludegraphics[valign=c, width = 2.0cm]{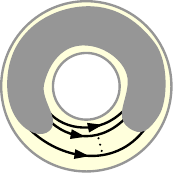} = \sum_{a \in \Irr \caC} d_a \,  \adjincludegraphics[valign=c, width = 2.0cm]{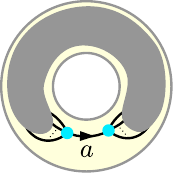} = \sum_{a \in \Irr \caC} d_a \lambda_a \,  \adjincludegraphics[valign=c, width = 2.0cm]{annzerobasis.pdf}, 
    \end{equation*}
    where the grey blob contains some arbitrary string diagram. The first equality is \eqref{eq:decomposition into simples}, and for the second equality we note that the grey blob together with the coloured vertices evaluates to some $\lambda_a \id_a \in \caC(a \rightarrow a) \simeq \C$. We conclude that Eq. \eqref{eq:ONB of ann} is indeed a basis. To see that it is orthonormal, compute
    \begin{equation*}
        \left( \adjincludegraphics[valign=c, width = 1.6cm]{annzerobasis.pdf}, \adjincludegraphics[valign=c, width = 1.6cm]{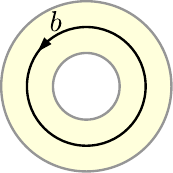} \right) = Z_{\ann \times I} \left( \adjincludegraphics[valign=c, width = 2.0cm]{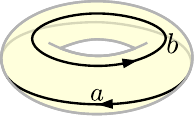} \right) = \frac{\delta_{a b}}{d_a} Z_{B^3} \left( \adjincludegraphics[valign=c, width = 2.0cm]{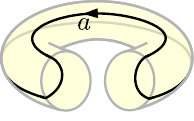} \right) = \delta_{a b},
    \end{equation*}
    where we evaluate $Z_{\ann \times I}$ using the gluing law \eqref{eq:gluing formula} by cutting the solid torus transversally to the $a$ and $b$ strings. The summation of \eqref{eq:gluing formula} runs over the unique basis vector \eqref{eq:ONB of A(D_2;aa)}.
\end{enumerate}

\subsection{Orthonormal basis of skein modules on punctured spheres}

Let $\Sigma$ be an extended surface homeomorphic to a sphere with $m$ holes cut out, and let $\anchor$ be an anchor for $\Sigma$.
Let $\mathbf{X}:\Bd(\Sigma) \to \Irr \ZC$, writing $X_\kappa = \mathbf{X}(\caS^{\anchorino}_\kappa)$.
Let $\mathbf{\underline{a}}$ be a boundary condition on $A(\Sigma)$ specified by a function on $\Bd(\Sigma)$ defined by 
$$\underline a_\kappa = \mathbf{\underline{a}}(\caS_\kappa^{\anchorino}) : m_{\caS_\kappa^{\anchorino}}\to \Irr(\caC).$$
For given $\mathbf{X},\mathbf{\underline{a}}$, let $\mathbf{i}:\Bd(\Sigma)\to \N$ be such that $i_\kappa = \mathbf{i}(\caS_\kappa^{\anchorino})$ belongs to the index set of the chosen basis of $\caC(X_\kappa \to \otimes\underline a_\kappa)$.
The set of such functions $\mathbf i$ label the corresponding product basis of 
$$\caC(X_1 \to \otimes \underline a_1) \otimes \cdots \otimes \caC(X_m \to \otimes \underline{ a}_m),$$
and we refer to $\mathbf{i}$ as indexing the $(\mathbf{X,\underline{a}})$-basis.
For every $m$-tuple consisting of $X_1, \ldots, X_m \in \Irr \ZC$, fix an orthonormal basis $\{ \al_l \}$ of $\ZC( \I \rightarrow X_1 \otimes \cdots \otimes X_m)$ with respect to the trace inner product. 

\begin{lemma} \label{lem:skein_basis}
    The elements,
    \begin{equation} \label{eq:skein basis}
    e^{{\anchorino};\mathbf{X,\underline{a}}}_{\mathbf{i},l} := 
    \Phi_{\Sigma}^{\anchor}( \al_l \otimes w^{X_1 \, \underline{a}_1}_{i_1} \otimes \cdots \otimes w^{X_m \, \underline{a}_m}_{i_{m}} )
    \end{equation}
    ranging over all $\mathbf{X} : \Bd(\Sigma)\to \Irr \ZC$, all boundary conditions $\underline{\mathbf{a}}$ on $\Sigma$, all $\mathbf{i}$ indexing the $(\mathbf{X,\underline{a}})$-basis, and all $l$ indexing the basis of $\ZC(\I \to X_1 \otimes \cdots \otimes X_m)$,
    form an orthonormal basis of $A(\Sigma)$ w.r.t. the TQFT inner product. In particular, the map $\Phi_{\Sigma}^{\anchor}$ of Proposition \ref{prop:characterisation of skein modules} is an isomorphism of Hilbert spaces.
\end{lemma}

\begin{proof}
    Throughout this proof, we drop the anchor from notation. For any $\mathbf{X, \underline{a},i},l$, let
    \begin{equation}\label{eq:ONB tilde of Sigma_m}
        \tilde e^{\mathbf{X,\underline{a}}}_{\mathbf{i},l}
        = \left( \prod_{\kappa=1}^m d_{X_{\kappa}}^{-1/2} \right) \, \caD^{1-m} \, \times \, e^{\mathbf{X,\underline{a}}}_{\mathbf{i},l} \, 
        = \adjincludegraphics[valign=c, width = 5.0cm]{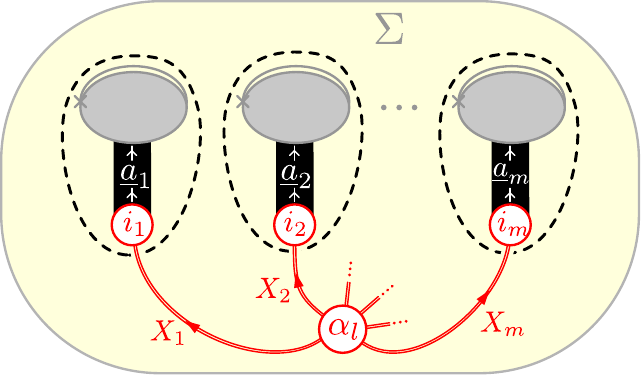} \in A(\Sigma)
    \end{equation}
    Using the resolution of identity from Eq. \eqref{eq:Tube_n matrix unit properties} we find that an arbitrary element of $A(\Sigma)$ can always be written as a linear combination of elements of the form on the left hand side of Eq. \eqref{eq:diagram simplification}, where we present $\Sigma$ as a disk with $m-1$ holes cut out of the interior, and the grey blob contains an arbitrary string diagram.

    \begin{equation} \label{eq:diagram simplification}
        \adjincludegraphics[valign=c, width = 3.0cm]{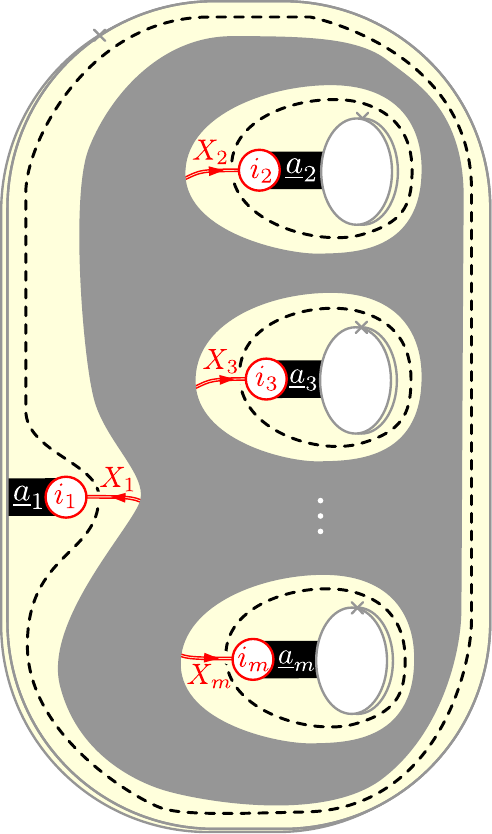} = \adjincludegraphics[valign=c, width = 3.0cm]{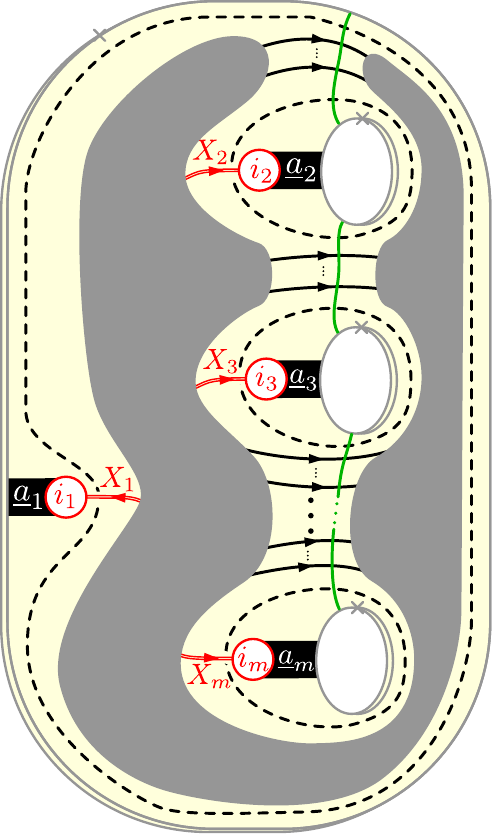} = \adjincludegraphics[valign=c, width = 3.0cm]{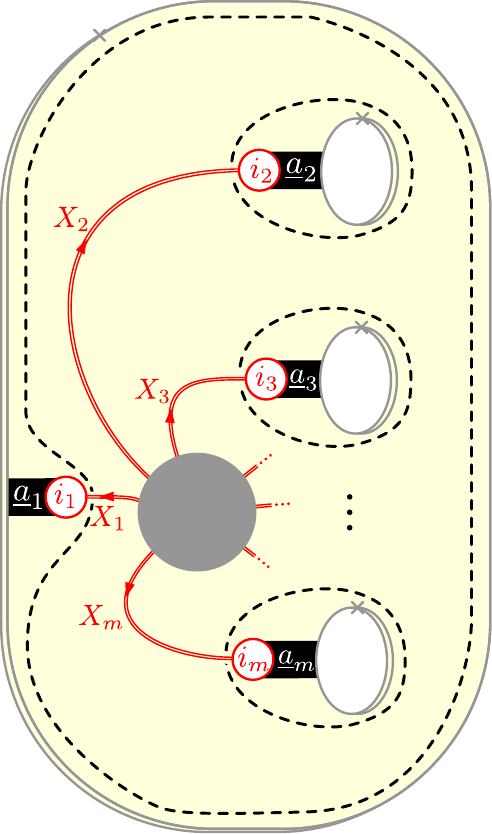} = \adjincludegraphics[valign=c, width = 3.0cm]{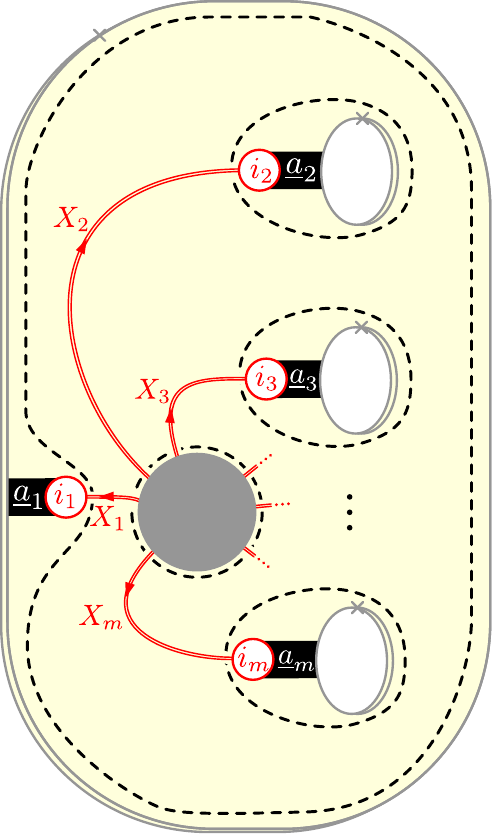}
    \end{equation}
    The string diagram contained in the grey blob may be moved past the punctures by isotopy and the cloaking property. In more detail, draw green lines connecting boundary components $\caS_{\kappa}$ and $\caS_{\kappa + 1}$ for $\kappa = 1, \cdots, m-1$ such that these green lines do not intersect the anchor. This condition makes the green lines unique up to isotopy of $\Sigma$, keeping the anchor fixed. By isotopy invariance one can always arrange the string diagram such that the green lines do not contain vertices of the string diagram and such that edges of the string diagram intersect the green lines transversally, as shown in the second panel of Eq. \eqref{eq:diagram simplification}. In the next step we use the cloaking property of the dotted lines and isotopy to deform the grey blob to a contractible region as in the third panel.
    Finally, we can insert a new small dotted loop and use the cloaking property of the pre-existing dotted lines to obtain a string diagram as in the final panel of Eq. \eqref{eq:diagram simplification}. In this final diagram, the grey blob surrounded by the dotted loop represents a morphism in $\ZC(\I \rightarrow X_1 \otimes \cdots \otimes X_m)$ which can be written as a linear combination of the basis elements $\{ \al_l \}$. In all, we conclude that vectors $\tilde e^{\mathbf{X,\underline{a}}}_{\mathbf{i},l}$ span the skein module $A(\Sigma)$.

    Let us now show that these vectors form an orthogonal family w.r.t. the TQFT inner product. This follows from the following computation, where we again present $\Sigma$ as a disk with $m-1$ holes cut out:
    \begin{align*}
        \bigg( \tilde e^{ \mathbf X', \mathbf{\underline{a}}' }_{\mathbf{j}, l'}, & \tilde e^{\mathbf{X,\underline{a}}}_{\mathbf{i},l}\bigg)_{A(\Sigma)} = \left( \prod_{\kappa = 1}^m \delta_{\underline{a}_{\kappa}, \underline{a}'_{\kappa}} \right) \, Z_{\Sigma \times I}\left(  \adjincludegraphics[valign=c, width = 6.0cm]{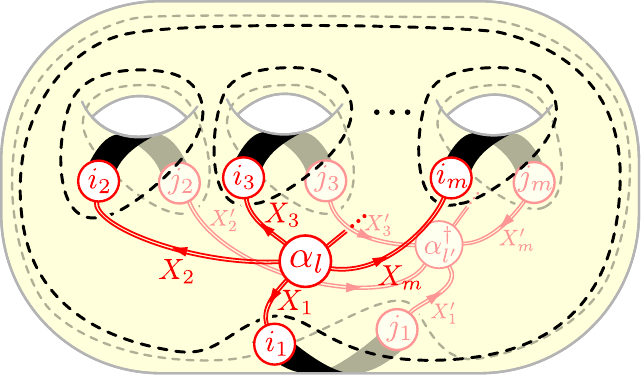}  \right) . \\
        \intertext{For each connected boundary component of $\Sigma$ we now have a pair of dotted lines going around the corresponding hole of $\Sigma 
        \times I$. Using the cloaking property of the dotted lines we pull one of each pair over the other to encircle the $w^{X_{\kappa} \underline{a}_{\kappa}}_{i_{\kappa}}$ and $(w^{X'_{\kappa} \underline{a}'_{\kappa}}_{j_{\kappa}})^{\dag}$ morphisms, and use Eq. \eqref{eq:useful identity} to obtain}
        &= \left( \prod_{\kappa = 1}^m \frac{  \delta_{\underline{a}_{\kappa}, \underline{a}'_{\kappa}} \delta_{X_{\kappa}, X'_{\kappa}} \delta_{i_{\kappa} \, j_{\kappa}} }{ d_{X_{\kappa}}  } \right) \, Z_{\Sigma \times I}\left(  \adjincludegraphics[valign=c, width = 6.0cm]{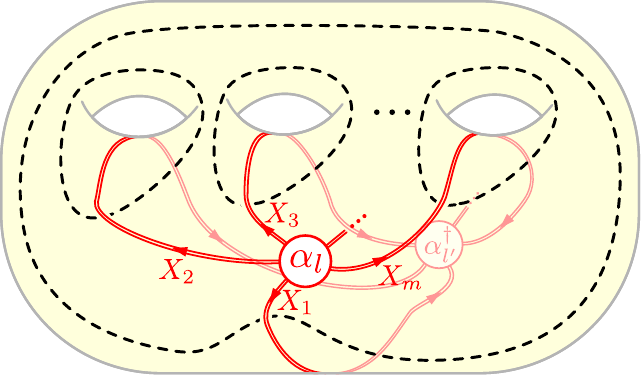}  \right). \\
        \intertext{The resulting diagram only contains dotted lines and lines labelled by objects of $\ZC$. We may contract the outer dotted line to a point using the cloaking property and so get rid of it. Now evaluate $Z_{\Sigma \times I}$ using the gluing formula \eqref{eq:gluing formula} by cutting $\Sigma \times I$ along $m-1$ disks transversal to the remaining dotted loops to obtain a 3-ball and summing over bases Eq. \eqref{eq:ONB of A(D_2;aa)}. We find}
        &= \frac{1}{\caD^{2(m-1)}} \, \left( \prod_{\kappa = 1}^m \frac{  \delta_{\underline{a}_{\kappa}, \underline{a}'_{\kappa}} \delta_{X_{\kappa}, X'_{\kappa}} \delta_{i_{\kappa}, j_{\kappa}} }{ d_{X_{\kappa}}  } \right) \\
        & \quad\quad\quad\quad\quad \times Z_{B^3}\left(  \adjincludegraphics[valign=c, width = 6.0cm]{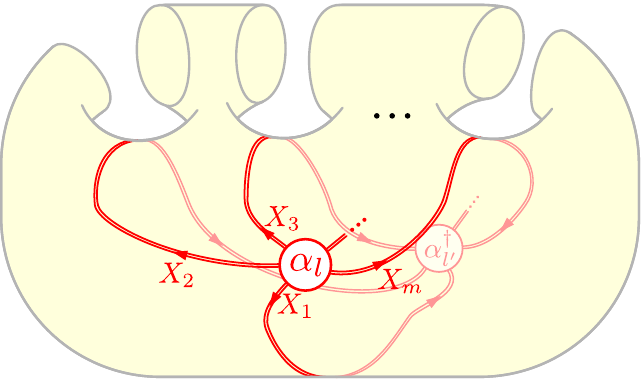}  \right) \\
        &= \frac{\delta_{l l'}}{\caD^{2(m-1)}} \, \left( \prod_{\kappa = 1}^m \frac{  \delta_{\underline{a}_{\kappa}, \underline{a}'_{\kappa}} \delta_{X_{\kappa}, X'_{\kappa}} \delta_{i_{\kappa}, j_{\kappa}}  }{ d_{X_{\kappa}}  } \right).
    \end{align*}
    Comparing to \eqref{eq:ONB tilde of Sigma_m}, we conclude that $e^{\mathbf{X,\underline{a}}}_{\mathbf{i},l}$ are the elements of an orthonormal basis of $A(\Sigma)$ with respect to the TQFT inner product.
\end{proof}

\subsection{Intertwining Tube-actions} \label{subsec:proof of skein characterisation}

\begin{proofof}[Proposition~\ref{prop:characterisation of skein modules}]
    By Lemma \ref{lem:skein_basis} we have that $\Phi_{\Sigma}^{\anchor}$ is an isomorphism of Hilbert spaces. It remains to verify that $\Phi_{\Sigma}^{\anchor}$ intertwines the $\Tube$-actions on $A(\Sigma)$ described in Section \ref{subsec:Tube algebras and their actions} with those on the morphism spaces given by given by Eq. \eqref{eq:Tube_n_module}. This can be seen by using the cloaking property of the dotted line as follows. For $\caS^{\anchorino}_{\kappa} \in \Bd(\Sigma)$, take $a \in \Tube_{\caS_{\kappa}^{\anchorino}}$ corresponding to an $f \in \caC( c \otimes \chi^{\otimes \caS_{\kappa}^{\anchorino}} \rightarrow \chi^{\otimes \caS_{\kappa}^{\anchorino}} \otimes c)$ as on the left-hand side of Eq. \eqref{eq:Tube_n_module}. Such $a$ span the $\Tube$ algebra. Take further $\al \in \ZC( \I \rightarrow X_1 \otimes \cdots \otimes X_m )$ and $w_{\kappa} \in \caC(X_{\kappa} \rightarrow \chi^{\otimes \caS_{\kappa}^{\anchorino}})$, and write $K := \left( \prod_{\kappa=1}^m d_{X_{\kappa}}^{1/2} \right) \, \caD^{m-1}$. Then
    \begin{align*}
        \Phi_{\Sigma}^{\anchor} &\left(  \al \otimes w_1 \otimes \cdots \otimes \big(  a \triangleright w_{\kappa} \big) \otimes \cdots \otimes w_{m} \right) 
        =
        K \times \adjincludegraphics[valign=c, width = 4.0cm]{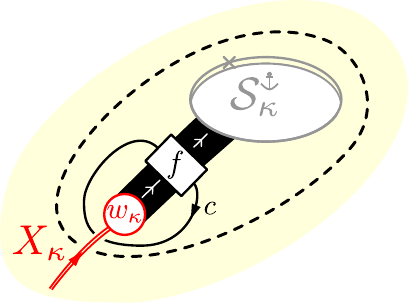} \\
        &= K \times \adjincludegraphics[valign=c, width = 4.0cm]{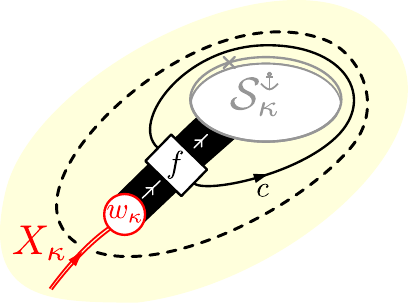}
        = 
        a \triangleright_{\kappa} \, \Phi_{\Sigma}^{\anchor} \left( \al \otimes w_1 \otimes \cdots \otimes  w_{\kappa} \otimes \cdots \otimes w_{m} \right),
    \end{align*}
    where we only depicted the string diagrams near the boundary component $\caS_{\kappa}^{\anchorino}$.
\end{proofof}

\section{Proof of Proposition \ref{prop:isomorphism of skein subspace and skein module}} \label{app:proof of skein subspace isomorphism}

In order to prove the Proposition, we must compare the TQFT inner product on $A_{\Sigma_C}$ with the skein inner product on $H_C$. To do this, we first compare the TQFT inner product on $A(\Sigma_{C_1}) \simeq H_{C_1}$ with the skein inner product on the string-net subspace $H_{C_1}$:
\begin{lemma} \label{lem:TQFT vs skein on string-net space}
    Let $C$ be a finite region and let $\phi, \psi \in H_{C_1}$. Then
    $$  \big( \sigma_{C_1}(\phi), \sigma_{C_1}(\psi) \big)_{A(\Sigma_{C_1})} = \langle \phi, \psi \rangle. $$
    \ie $\sigma_{C_1} : H_{C_1} \rightarrow A(\Sigma_{C_1})$ is an isometry.
\end{lemma}

\begin{proof}
    Let $\phi_x = \bigotimes_{v \in V_C} \phi_{x(v)}$ and $\phi_y = \bigotimes_{v \in V_C} \phi_{y(v)}$ be product states in $H^b_{C_1}$ corresponding under $\pi_{C_1}$ to string diagrams $x$ and $y$ on $\Sigma_{C_1}$. We represent
    \begin{align*}
    	(x, y)_{A(\Sigma_{C_1})} &= Z_{\Sigma_{C_1} \times I} \left( \adjincludegraphics[valign=c, height = 1.8cm]{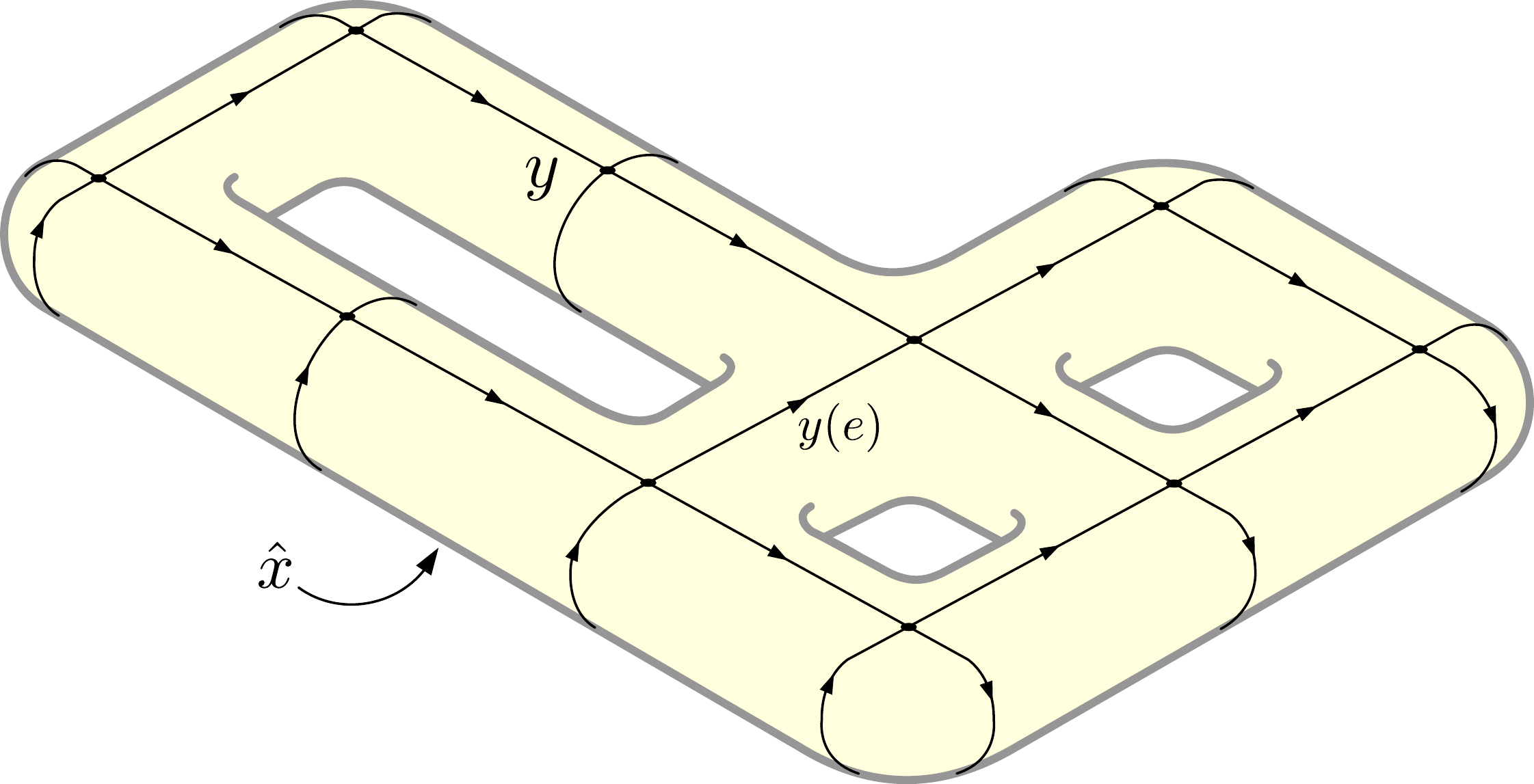} \right). \\
    	\intertext{We cut $\Sigma_{C_1} \times I$ transversal to edges of $C_1$ and use the gluing formula \eqref{eq:gluing formula}, summing over bases Eq \eqref{eq:ONB of A(D_2;aa)} to obtain}
    	      (x, y)_{A(\Sigma_{C_1})} &= \left( \prod_{e \in E_C} \frac{\delta_{x(e), y(e)}}{d_{x(e)}}\right) \, Z_{\bigsqcup_{v\in V_C} B^3} \left( \adjincludegraphics[valign=c, height = 1.8cm]{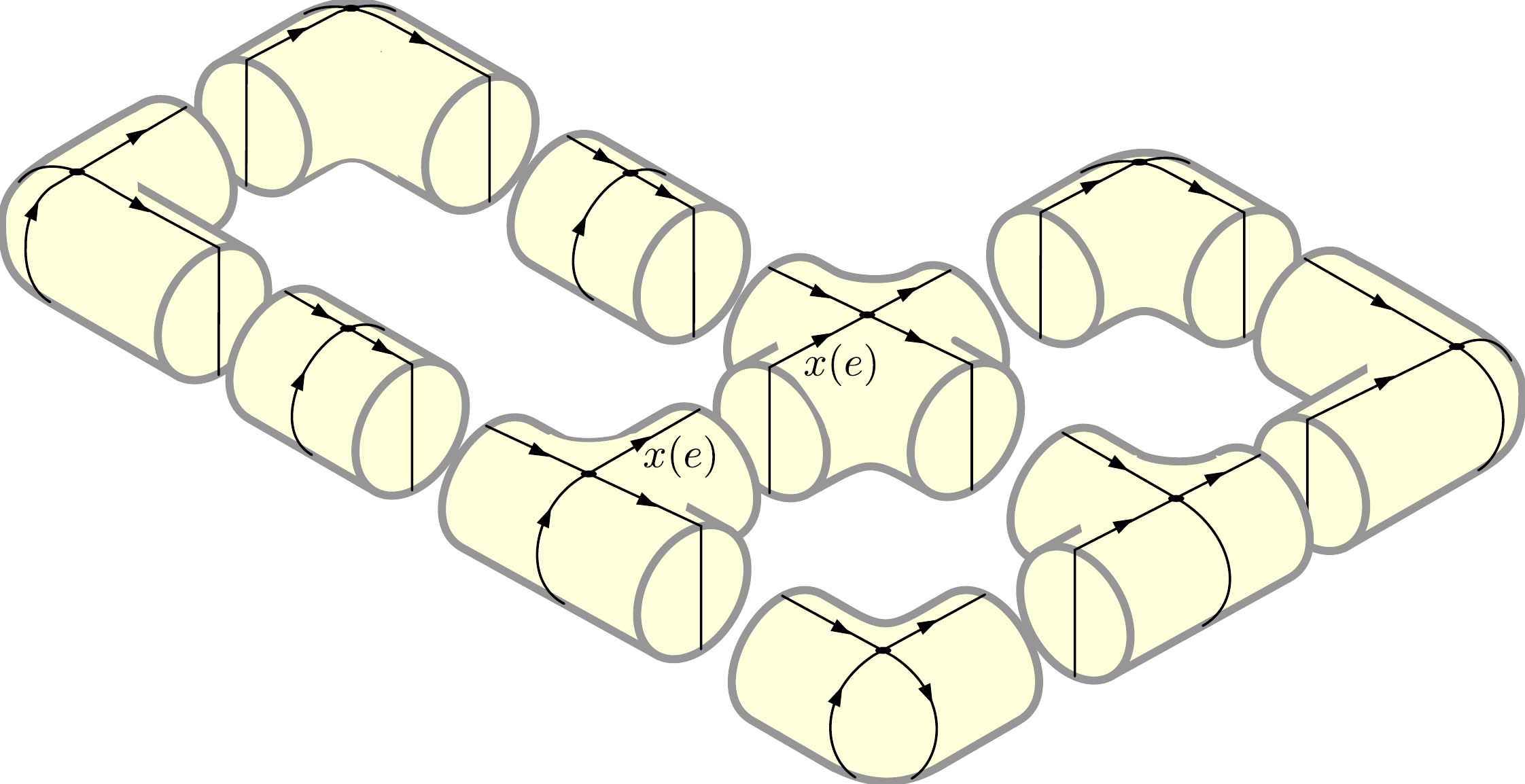} \right). \\
    	       \intertext{Using the product property of the partition function we now find}
    	      (x, y)_{A(\Sigma_{C_1})} &= \left(\prod_{e \in E_C} \frac{\delta_{x(e), y(e)}}{d_{x(e)}}\right) \, \prod_{v \in V_C} \left( \prod_{e \in \vec \partial v} d_{x(e)}^{1/2}  \right) \langle \phi_{x(v)}, \phi_{y(v)} \rangle \\
    	       &= \left( \prod_{e \in \vec \partial C} d_{x(e)}^{1/2} \right) \langle \phi_x, \phi_y \rangle = d_{\partial x}^{1/2} \langle \phi_x, \phi_y \rangle.
    \end{align*}
    Recall that $\sigma_{C_1}(\phi_x) = d_{\partial x}^{-1/4} [x]_{\Sigma_{C_1}}$ and likewise for $\phi_y$. Since $H_{C_1}^b$ is spanned by vectors of this form, and since both inner products vanish whenever boundary labellings of $\phi$ and $\psi$ do not match, we conclude that $\sigma_{C_1}$ is an isometry.
\end{proof}

We proceed to consider the TQFT inner product on $A(\Sigma_C)$.
\begin{lemma} \label{lem:TQFT vs skein on skein subspaces}
    Let $C$ be a finite region and $\phi, \psi \in H_{C}$. Then
    $$ \big( \sigma_C(\phi), \sigma_C(\psi) \big)_{A(\Sigma_C)}  =  \langle \phi, \psi \rangle,$$
    \ie $\sigma_C|_{H_C} : H_C \rightarrow A(\Sigma_C)$ is an isometry.
\end{lemma}

\begin{proof}
    Let $\phi_x = \bigotimes_{v \in V_C} \phi_{x(v)}$ and $\phi_y = \bigotimes_{v \in V_C} \phi_{y(v)}$ be product states in $H^b_{C_1}$ corresponding under $\pi_{C_1}$ to string diagrams $x$ and $y$ on $\Sigma_{C_1}$ with boundary condition $b$. 
    We relate the TQFT inner product on $A(\Sigma_C)$ to the one on $A(\Sigma_{C_1})$ by using the gluing formula \eqref{eq:gluing formula}.
    Indeed $\Sigma_C \times I$ is obtained from $\Sigma_{C_1} \times I$ by patching the holes corresponding to the internal faces $F_C$. A hole is patched by gluing a 3-ball along an annulus around the hole, so the gluing formula gives rise to a summation over basis elements of the annulus with no marked points \eqref{eq:ONB of ann} for each hole. 
    Using the product property of the partition function, the partition function on $\Sigma_C \times I$ is in this way written as a sum of products of a partition function on $\Sigma_{C_1}$ with partition functions on the glued in 3-balls. The 3-balls have basis elements \eqref{eq:ONB of ann} placed on a belt on their surfaces, which evaluate to the quantum dimensions $d_a$. In this way we find
    \begin{align*}
        ([x], [y])_{A(\Sigma_C)} &= Z_{\Sigma_C \times I} \left( \adjincludegraphics[valign=c, height = 2.0cm]{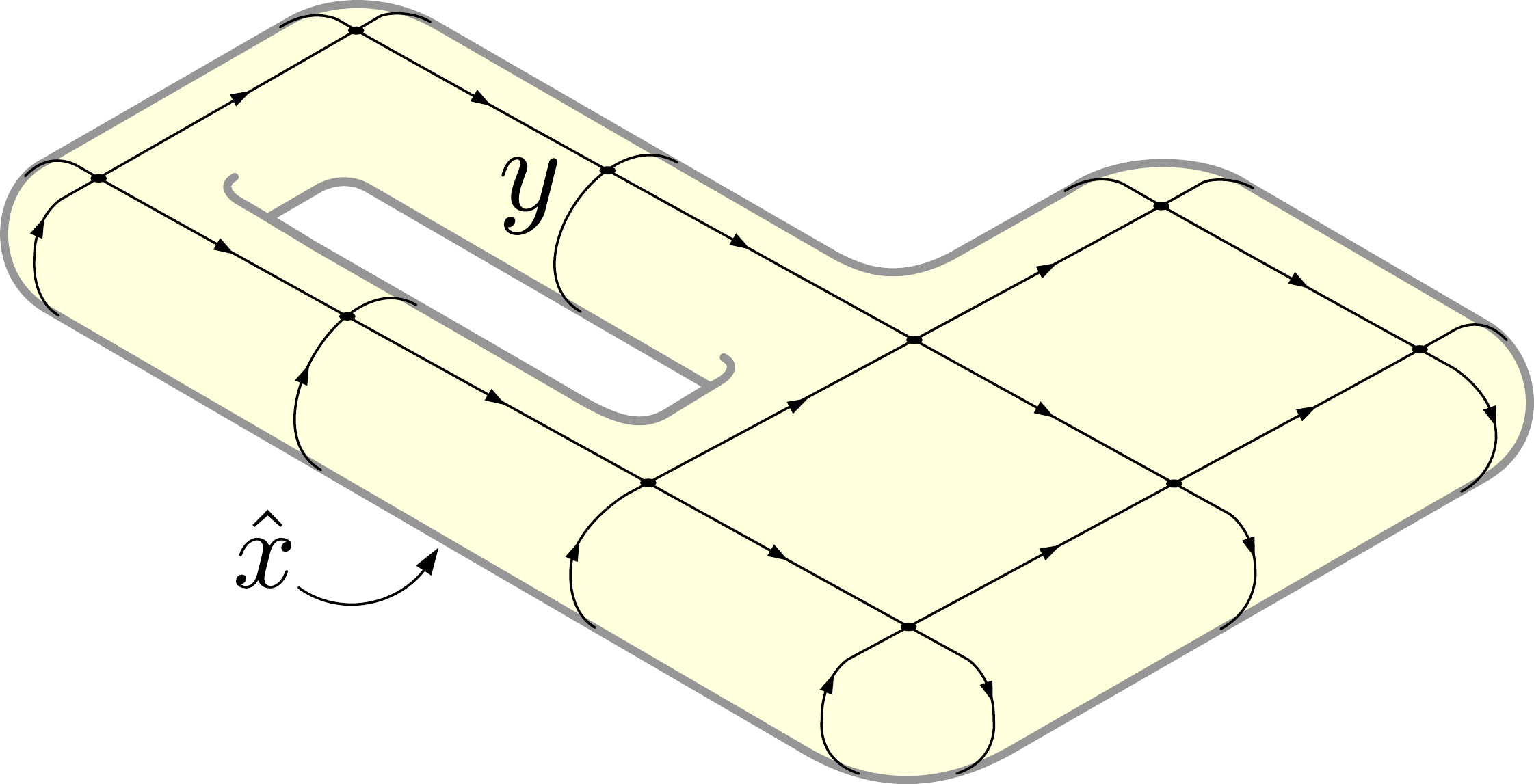} \right) \\
        &= \sum_{(a_{f}) \in (\Irr \caC)^{F_C}}  \, Z_{(\Sigma_{C_1} \times I) \sqcup B^3 \sqcup \cdots \sqcup B^3} \left(  \adjincludegraphics[valign=c, height = 2.0cm]{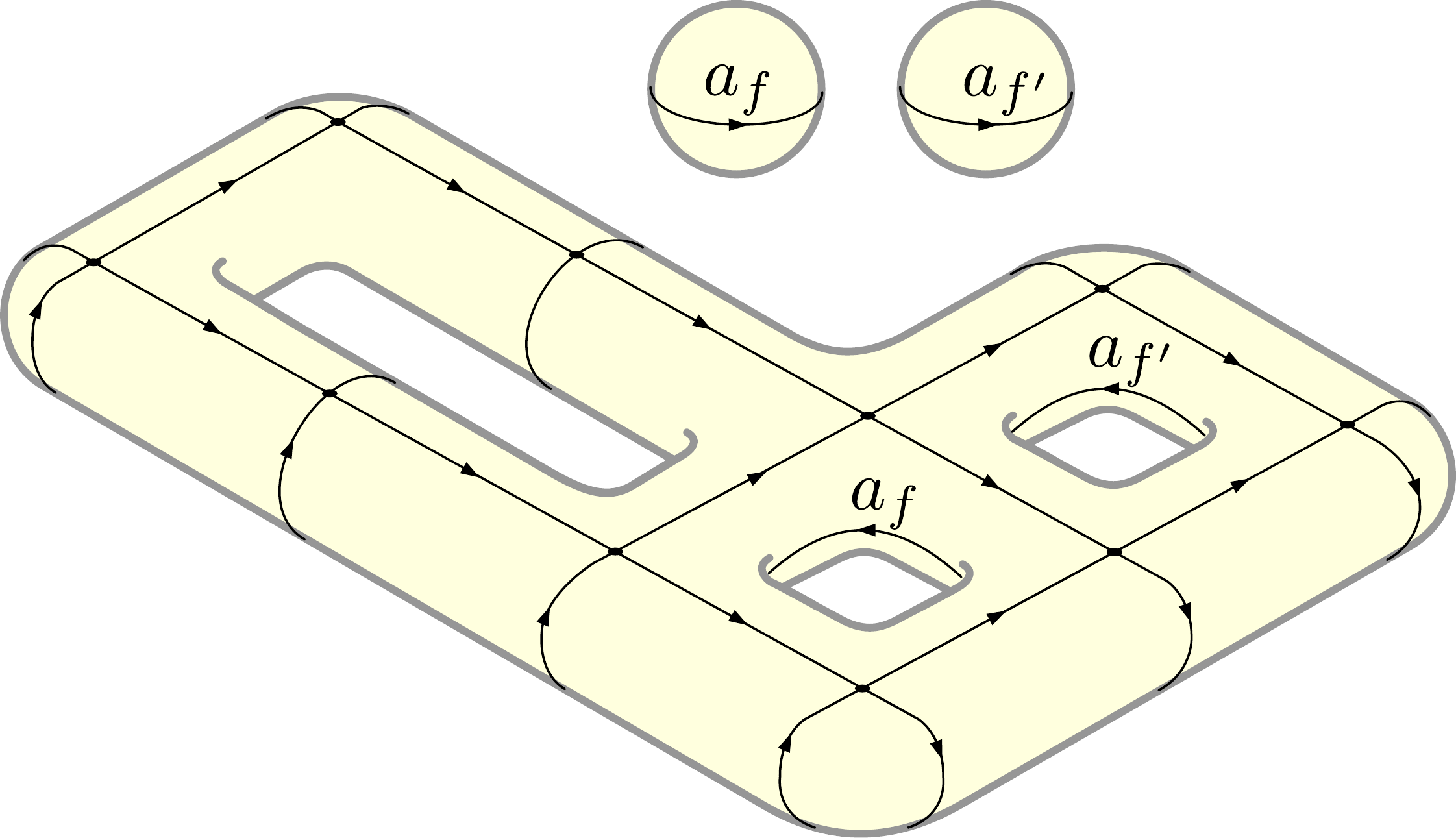}  \right) \\
        \intertext{ }
        &= \caD^{2 \abs{F_C}} \, Z_{\Sigma_{C_1} \times I} \left( \adjincludegraphics[valign=c, height = 2.0cm]{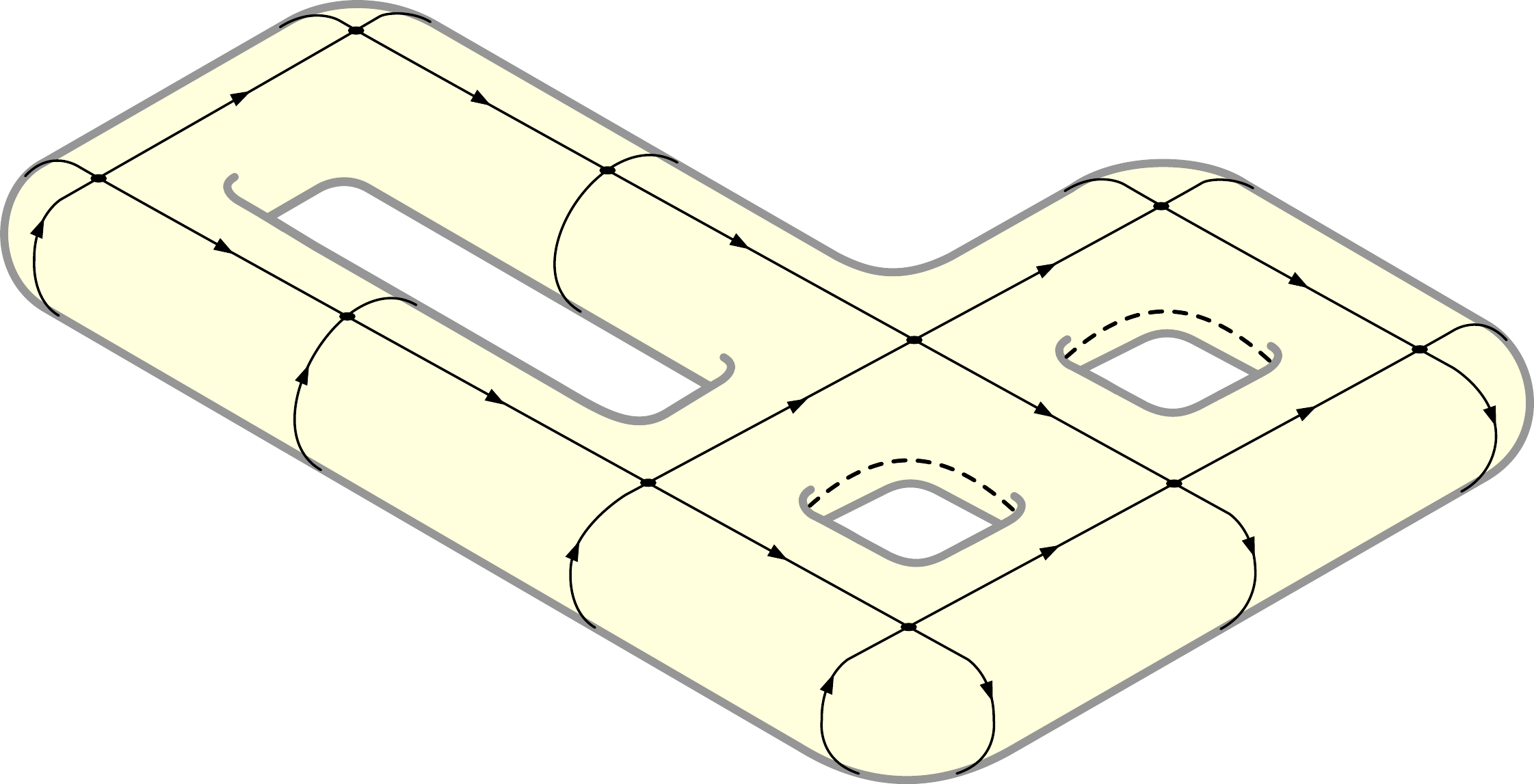} \right)\\
        &= \caD^{2 \abs{F_C}} d_{\partial x}^{1/2} \, \left\langle \phi_x,  B_C \phi_y \right\rangle,
    \end{align*}
    where we used Lemma \ref{lem:TQFT vs skein on string-net space} in the last step to relate the TQFT inner product on $A(\Sigma_{C_1})$ to the skein inner product. Recalling the definition of $\sigma_C$, we obtain $( \sigma_C(\phi_x), \sigma_C(\phi_y) )_{A(\Sigma_C)} = \langle B_C \phi_x, B_C \phi_y \rangle$. Since any $\phi, \psi \in H_C^b$ can be expressed as linear combinations of vectors of the form $\phi_x$ and $\phi_y$, and both inner products vanish whenever the boundary conditions do not match, this concludes the proof.
\end{proof}

\begin{lemma}\label{lem:sigma intertwines Tube actions}
    Let $C$ be a finite region. 
    For all $\phi \in H_{C_1}$, and all $a \in \Tube_\caS$ with $\caS \in \Bd(\Sigma_C)$, 
    \begin{equation*}
        \sigma_C \big(  \frt_{\caS}(a) \phi \big) = a \triangleright_{\caS} \sigma_C(\phi).
    \end{equation*}
\end{lemma}
\begin{proof}
  Recall that $\sigma_C$ is defined on the whole of $H_{C_1}$. Consider a product state $\phi_x \in H_{C_1}$ corresponding to a string diagram $x$ on $\Sigma_{C_1}$. Let $a = [y] \in \Tube_{\caS}$ corresponding to a string diagram $y$ with boundary condition $y_{\hat S}$ matching $x_{\caS}$. 
  Schematically,
    \begin{align*}
        \phi_x = \adjincludegraphics[valign=c, width = 4.0cm]{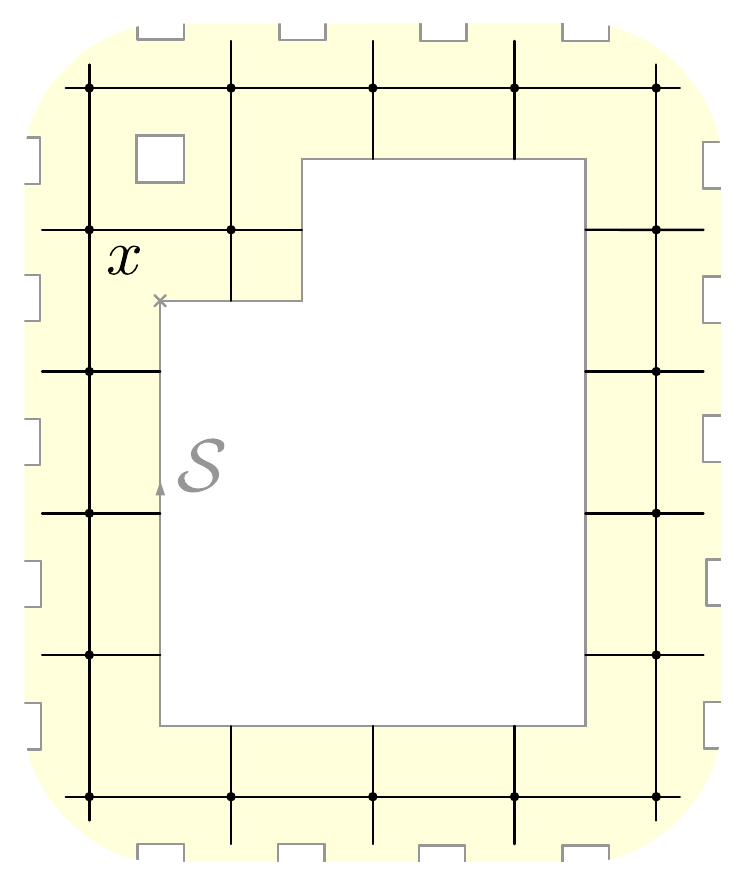}, \quad  a = \adjincludegraphics[valign=c, width = 2.5cm]{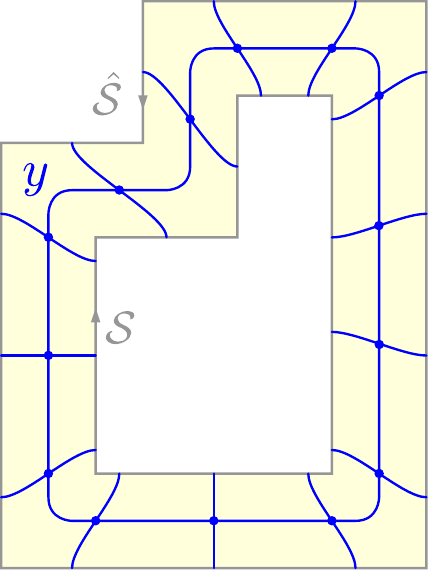}.
    \end{align*}
    Since $\phi_x$ is a product state it factorises as $\phi_x=\phi_\caS \otimes \phi_{\caS^c}$, where $\phi_\caS \in H_{C^\caS}$, and $\phi_{\caS^c}\in \caH_{C_0\setminus C^\caS_0}$. 
    The locality of local relations implies,
    $$
    \frt_{\caS}(a) \phi_x = \sigma_{C_1^{\caS}}^{-1} \big(  a \triangleright_{\caS} \sigma_{C^{\caS}}( \phi_\caS)  \big) \otimes \phi_{\caS^c} = \sigma_{C_1}^{-1} \big(  a \triangleright_{\caS} \sigma_{C_1}( \phi_x)  \big).
    $$
    In the graphical representation (Convention \ref{conv:graphical representation}), this is illustrated by
    $$
        \frt_{\caS}(a) \phi_x = \left( \frac{d_{y_{\caS}}}{d_{x_{\caS}}} \right)^{1/4} \,\adjincludegraphics[valign=c, width = 4.0cm]{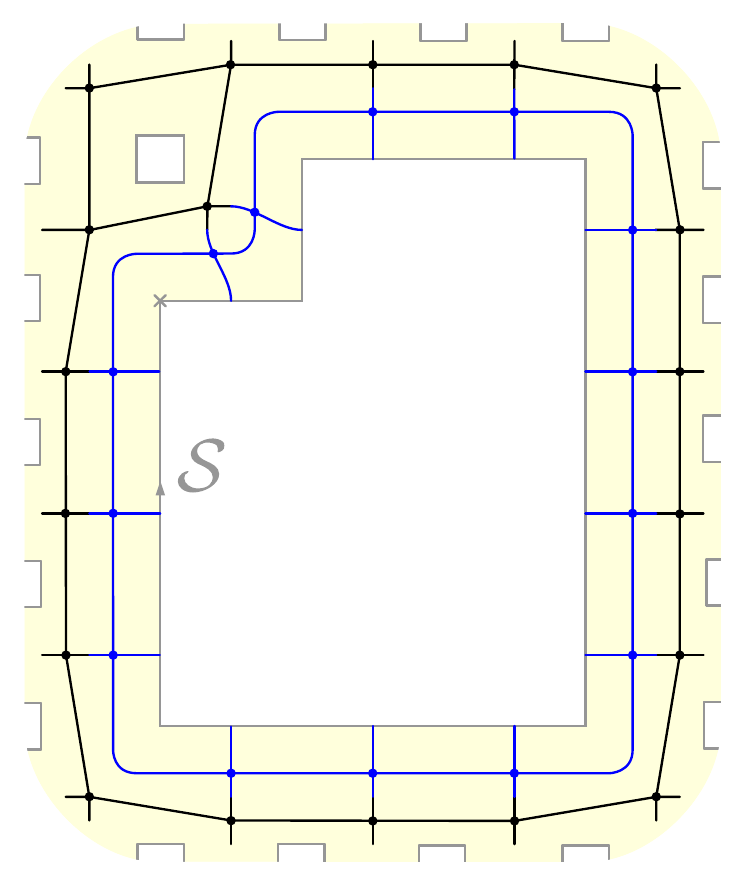}.
    $$
    As a result, 
    \begin{equation*}
        \sigma_C \big(  \frt_{\caS}(a) \phi_x \big) = \caD^{-\abs{F_C}} d_{x_{\caS}}^{-1/4} \big(  a \triangleright_{\caS} [x]_{\Sigma_C}  \big) =  a \triangleright_{\caS} \sigma_C(\phi_x).
    \end{equation*}
    If the boundary conditions $x_{\caS}$ and $y_{\hat \caS}$ do not match, then this equality still holds because either side vanishes. The claim now follows because vectors of the form $\phi_x$ span $H_{C_1}$.
\end{proof}

\begin{lemma} \label{lem:sigma_C absorbs B_C}
    Let $C$ be a finite region, then
    $$ \sigma_C( \phi )  = \sigma_C ( B_C \phi )$$
    for all $\phi \in H_{C_1}$. 
    The same holds true with $\sigma_C$ replaced by $\pi_C$.
\end{lemma}

\begin{proof}
    Let us first prove the claim for $\pi_C$. It is sufficient to show that $\pi_C( \phi_x )  = \pi_C ( B_C \phi_x )$ for any string-net diagram $x$ on $\Sigma_{C_1}$.
    
    Recalling that $B_f$ acts on the string-net subspace $H_{C^f}$ as a $\Tube$ action, we find using \Cref{lem:sigma intertwines Tube actions} that $B_C\phi_x = \pi_{C_1}^{-1} \big( [\tilde x]_{\Sigma_{C_1}} \big)$ where $\tilde x$ is the string diagram on $\Sigma_{C_1}$ obtained from $x$ by inserting dotted loops around the holes corresponding to the faces of $C$. Since $\pi_C = \iota_C\circ \pi_{C_1}$,
    $$ \pi_C \big( B_C \phi_x \big)
    = \iota_C([\tilde x]_{\Sigma_{C_1}}) = [\tilde x]_{\Sigma_C} = [x]_{\Sigma_C} = \pi_C(\phi),$$
    where in the penultimate step we noted that, in $\Sigma_C$, we can contract all the dotted loops to a point, i.e. remove them by local relations on $\Sigma_C$.
    
    The claim for $\sigma_C$ now follows by noting that $B_C$ preserves boundary labellings, and $\sigma_C$ differs from $\pi_C$ only by factors depending on the boundary labelling.
\end{proof}

\begin{proofof}[Proposition~\ref{prop:isomorphism of skein subspace and skein module}]
    As a composition of surjective maps $\iota_C$ and $\sigma_{C_1}$, the map $\sigma_C$ is also surjective. Surjectivity of the restriction $\sigma_C|_{H_C}$ follows from \Cref{lem:sigma_C absorbs B_C}. 
    \Cref{lem:TQFT vs skein on skein subspaces} shows that $\sigma_C|_{H_C}$ is an isometry, and \Cref{lem:sigma intertwines Tube actions} shows that $\sigma_C$ is an intertwiner. 
\end{proofof}


\bibliographystyle{unsrturl}
\bibliography{bib}

\end{document}